\documentclass[12pt]{article}
\usepackage[a4paper,margin=1in]{geometry}
\usepackage{setspace}
\usepackage{mathtools}
\usepackage{amssymb}
\usepackage{amsthm}
\usepackage{algorithm}
\usepackage{algpseudocode}
\usepackage{graphicx}
\usepackage{placeins}
\graphicspath{{figures/}{./}}
\usepackage{booktabs}
\usepackage[authoryear,round]{natbib}
\usepackage[hidelinks]{hyperref}
\usepackage{orcidlink}
\usepackage{enumitem}
\usepackage{caption}
\newcommand{\alttext}[1]{%
  \begingroup
  \captionsetup{font=onehalfspacing,justification=raggedright,
    singlelinecheck=false,skip=0pt}%
  \caption*{\textit{Alt text:} #1}%
  \endgroup
}

\allowdisplaybreaks
\DeclareMathOperator*{\argmax}{argmax}

\algrenewcommand\algorithmicrequire{\textbf{Input:}}
\algrenewcommand\algorithmicensure{\textbf{Output:}}

\theoremstyle{plain}
\newtheorem{theorem}{Theorem}
\newtheorem{lemma}{Lemma}
\newtheorem{corollary}{Corollary}
\newtheorem{proposition}{Proposition}
\theoremstyle{definition}
\newtheorem{assumption}{Assumption}
\theoremstyle{remark}
\newtheorem{remark}{Remark}

\newtheorem{innlemma}{Lemma}
\newenvironment{lemmaS}[1]
  {\renewcommand\theinnlemma{#1}\innlemma}
  {\endinnlemma}

\newtheorem{inntheorem}{Theorem}
\newenvironment{theoremS}[1]
  {\renewcommand\theinntheorem{#1}\inntheorem}
  {\endinntheorem}

\newtheorem{inncorollary}{Corollary}
\newenvironment{corollaryS}[1]
  {\renewcommand\theinncorollary{#1}\inncorollary}
  {\endinncorollary}

\newtheorem{innassumption}{Assumption}
\newenvironment{assumptionS}[1]
  {\renewcommand\theinnassumption{#1}\innassumption}
  {\endinnassumption}

\makeatletter
\let\arxiv@maketitle\maketitle
\let\arxiv@@maketitle\@maketitle
\let\arxiv@title\title
\let\arxiv@author\author
\let\arxiv@date\date
\let\arxiv@thanks\thanks
\let\arxiv@and\and
\newcommand{\arxivrestoretitlecommands}{%
  \global\let\maketitle\arxiv@maketitle
  \global\let\@maketitle\arxiv@@maketitle
  \global\let\title\arxiv@title
  \global\let\author\arxiv@author
  \global\let\date\arxiv@date
  \global\let\thanks\arxiv@thanks
  \global\let\and\arxiv@and
}
\makeatother

\title{\Large Optimal multiple testing under family-wise error control:\\
elementary symmetric polynomials and a scalable algorithm}
\author{%
  Prasanjit Dubey\,\orcidlink{0000-0002-3667-5507}%
  \thanks{Corresponding author. 755 Ferst Dr NW, Atlanta, GA 30332, USA. Email: \href{mailto:pdubey31@gatech.edu}{pdubey31@gatech.edu}}%
  \qquad
  Xiaoming Huo\,\orcidlink{0000-0003-0101-1206}\\[6pt]
  \normalsize H.~Milton Stewart School of Industrial and Systems Engineering,\\
  \normalsize Georgia Institute of Technology, Atlanta, GA 30332, USA%
}
\date{}

\begin{document}

\maketitle

\begin{abstract}
\normalsize\onehalfspacing
\noindent Family-wise error rate control is essential when even one false rejection is costly, but distribution-free procedures do not exploit information in a specified alternative model.
Existing dual theory characterises the maximum-average-power procedure under strong family-wise error control, but its optimal multipliers had been computed only up to $K = 3$.
Under an exchangeable product model of independent $p$-values with a common non-increasing density under the alternative, we develop a general-$K$ statistical methodology that removes this computational barrier.
An elementary-symmetric-polynomial representation yields global monotonicity of every constraint function and converts the coupled multiplier problem into monotone coordinate-wise searches.
The resulting algorithm, \emph{symmetric-polynomial optimal testing} (SPOT), uses bisection coordinate descent and has polynomial per-sweep cost in $K$ and the Monte Carlo size.
Under stated conditions, its exact population objective values converge to the optimum and every limit point is optimal, without contraction or strong-convexity assumptions. Additional local regularity gives linear convergence of the exact population iterates, and an achieved $O_p(N^{-1/2})$ residual gives $\sqrt{N}$-consistent Monte Carlo output.
In a truncated-normal scaling experiment, the relative average-power gain over Hommel's method increases from 15\% at $K = 3$ to 83\% at $K = 12$.
Applications with up to 21 hypotheses demonstrate SPOT's practical reach.
\end{abstract}


\bigskip

\section{Introduction}
\label{sec:introduction}

Multiple-hypothesis testing with family-wise error rate (FWER) control remains essential when even a single false rejection carries high cost: in clinical trials with multiple endpoints \citep{Pocock87, Bretz09, Vickerstaff19}, in economic policy evaluation \citep{Romano05, list16, Viviano25}, and in large-scale experimentation \citep{Johari22}.
Given $K$ hypotheses $H_1, \ldots, H_K$ and a significance level $\alpha$, the goal is to find the testing policy $\vec{D} = (D_1, \ldots, D_K) \in \{0,1\}^K$ that maximises power while ensuring ${\rm pr}(V > 0) \le \alpha$ for every configuration of true and false nulls, where $V$ counts the false rejections.

Standard procedures \citep{Bonferroni1936, holm, Hochberg1988, hommel} provide valid family-wise error rate control (Bonferroni and Holm under arbitrary dependence; Hochberg and Hommel under independence or Simes-type positive dependence, both satisfied by the product model studied here). This distribution-free validity is a genuine strength, but it can leave power unused when a credible alternative model is available.
Their power loss tends to grow with $K$ in the settings we study because they base rejection decisions on the ordered $p$-values $u_{(1)} \le \cdots \le u_{(K)}$ and fixed critical constants: Bonferroni compares each $u_{(k)}$ to $\alpha/K$, Holm to $\alpha/(K - k + 1)$, and Hommel selects the most liberal among several such thresholds.
In every case the rejection decisions depend on the data only through the ordered $p$-values and these fixed, rank-based critical constants, never through the magnitude of the joint likelihood: Bonferroni decides each hypothesis from $u_{(k)}$, $K$, and $\alpha$ alone, while the step-down and step-up procedures (Holm, Hochberg, Hommel) couple the decisions across hypotheses only through the ordered comparisons.

Optimality of family-wise-error-rate-controlling procedures has been studied from several angles \citep{Spjtvoll1972OnTO, westfall98, Lehmann1, Dobriban, Rosenblum}.
\citet{RHPA22} made the key advance: they formulated the most powerful test as an infinite-dimensional binary program, proved strong duality, and showed that an optimal dual vector $\mu^* = (\mu_0^*, \ldots, \mu_{K-1}^*)$ exists under a non-redundancy assumption.

The framework of \citet{RHPA22} establishes the existence of $\mu^*$ but does not give a general method to compute it: they solve for $\mu^*$ only up to $K = 3$ hypotheses, and to our knowledge no scalable, convergence-analysed procedure for computing it had been demonstrated beyond that; without $\mu^*$ the optimal policy cannot be implemented.
The challenge is structural: the optimum is characterised by $K$ coupled nonlinear complementarity conditions, $\mu_\gamma \ge 0$, $F_\gamma(\mu) \le \alpha$, and $\mu_\gamma\{F_\gamma(\mu) - \alpha\} = 0$ for $\gamma = 0, \ldots, K-1$, where $F_\gamma(\mu) = {\rm FWER}_\gamma(\vec{D}^\mu)$ is the family-wise error rate of the $\mu$-induced policy under the configuration $\vec{h}_\gamma$ in which the first $\gamma$ hypotheses are alternatives and the remaining $K-\gamma$ are nulls (defined formally in \S\ref{sec:mht}).
Each $F_\gamma$ depends on $\mu$ through an integral over the ordered simplex $Q = \{\vec{u} \in [0,1]^K : u_1 \le \cdots \le u_K\}$, and the policy $\vec{D}^\mu$ is a discontinuous function of $\mu$ (it involves indicator functions).
The coordinate-bisection construction and contraction-conditional convergence analysis of \citet{DubeyHuo2024K3} are specific to $K = 3$; beyond that case, neither a global monotonicity theorem covering all multiplier coordinates nor the coordinate-wise crossing structure had been established.

We solve the computational problem for $K$ exchangeable hypotheses whose $p$-values are independent with a common alternative density (the product model of Section~\ref{sec:pvalue}).
We call this algorithm \emph{symmetric-polynomial optimal testing} (SPOT).
The key insight is algebraic: the family-wise error rate constraint coefficients $b_{\gamma,k}(\vec{u})$ in the dual formulation, which weight the contribution of rejecting hypothesis $k$ to the error rate under configuration $\vec{h}_\gamma$, admit a closed-form factorisation through \emph{elementary symmetric polynomials}:
\[
b_{\gamma,k}(\vec{u}) = \gamma!\,(K{-}\gamma)!\; \prod_{j=1}^{k-1} g(u_j) \;\cdot\; e_{\gamma-k+1}\bigl\{g(u_{k+1}), \ldots, g(u_K)\bigr\},
\]
where $g$ is the alternative $p$-value density and $e_m$ is the $m$th elementary symmetric polynomial (Lemma~\ref{lem:general_k_coefficients}).
Since $g \ge 0$ and elementary symmetric polynomials of non-negative arguments are non-negative, every $b_{\gamma,k} \ge 0$.
That the error coefficients are non-negative was already noted by \citet[\S2.4]{RHPA22}. Writing them in the closed form of Lemma~\ref{lem:general_k_coefficients} makes this non-negativity transparent for every $K$, and the non-negativity in turn yields a \emph{global monotonicity theorem} (Theorem~\ref{thm:general_k_monotonicity}), new for general $K$: the functions $F_\gamma(\mu)$ are simultaneously non-increasing in every component of $\mu$, without case analysis (the proof uses only $b_{\gamma,k}\ge0$).
This explicit form is also what the smoothness and Monte Carlo analyses of \S\S\ref{sec:algorithm}--\ref{sec:mc_error} exploit.
Monotonicity guarantees that each coordinate-wise constraint has a well-defined leftmost crossing point (Corollary~\ref{cor:unique_root}), enabling SPOT, the bisection coordinate-descent scheme in Algorithm~\ref{alg:general_k}.
We prove that its exact (population) form drives the dual objective to its optimum, with every limit point an optimal dual vector, under Assumptions~\ref{as:assumption3}--\ref{as:assumption5} and no contraction or strong-convexity hypothesis.
The exact population iterates converge to $\mu^*$ at a linear rate under standard local non-degeneracy conditions (strict complementarity, continuous differentiability of the target functions, and a nonsingular active Jacobian), and the Monte Carlo implementation is analysed separately in \S\S\ref{sec:algorithm}--\ref{sec:mc_error}.

Together, Theorem~\ref{thm:general_k_monotonicity} and SPOT (Algorithm~\ref{alg:general_k}) turn the existence characterisation of \citet{RHPA22} into an implementable optimal policy for arbitrary $K$, with population convergence and Monte Carlo error theory; we run SPOT for up to $K = 21$, including on real clinical-trial and replication data (\S\ref{sec:realdata}).
The elementary symmetric polynomial formulation replaces the companion's explicit $K = 3$ coefficient and indicator-function calculations with a single formula valid for every $K$ (Remark~\ref{rem:k3_relation}).
Table~\ref{tab:scaling} previews the main empirical finding: while the power of Hommel's method drops from $0.555$ to $0.334$ as $K$ grows from $3$ to $12$, the policies computed by SPOT maintain nearly flat power ($0.638$ to $0.610$), a relative gain growing from 15\% to 83\% that reflects the optimal policy's model-based likelihood weighting in place of fixed, distribution-free critical constants.

\section{Problem formulation}
\label{sec:formulation}

\subsection{Setup and notation}
\label{sec:mht}

We assess $K$ hypotheses $\{H_1, \ldots, H_K\}$ simultaneously.
The unknown set of true nulls is $\mathcal{N} \subseteq [K]$, where $[K] = \{1, \ldots, K\}$.
A testing policy $\vec{D}(\vec{X}) = (D_1, \ldots, D_K) \in \{0,1\}^K$ maps data $\vec{X} = (X_1, \ldots, X_K)$ to rejection decisions, with $D_k = 1$ indicating rejection of $H_k$.

We parameterise the configuration space by vectors $\vec{h} = (h_1, \ldots, h_K) \in \{0,1\}^K$, where $h_k = 1$ indicates that $H_k$ is false (the alternative holds) and $h_k = 0$ indicates that $H_k$ is true (the null holds).
The configuration $\vec{h}_\gamma = (1,\ldots,1,0,\ldots,0)^{\rm T}$ with exactly $\gamma$ alternatives and $K-\gamma$ nulls plays a central role.

Under configuration $\vec{h}_\gamma$, the average power is
\[
\Pi_\gamma(\vec{D}) = \gamma^{-1} \mathbb{E}_{\vec{h}_\gamma}\!\left(\sum_{k=1}^\gamma D_k\right), \quad 1 \le \gamma \le K.
\]
We also report the \emph{any-discovery power} under the all-alternatives configuration,
\[
\Pi_{\rm any}(\vec{D}) = {\rm pr}_{\vec{h}_K}\!\left(\sum_{k=1}^K D_k \ge 1\right),
\]
which is the probability of making at least one correct rejection when every hypothesis is an alternative \citep{DubeyHuo2024K3}.
The family-wise error rate is
\[
{\rm FWER}_\gamma(\vec{D}) = {\rm pr}_{\vec{h}_\gamma}(V > 0), \quad 0 \le \gamma < K,
\]
where $V = \sum_{k=\gamma+1}^K D_k$ counts the false rejections, i.e., the number of null hypotheses that are incorrectly rejected.
The case $\gamma = 0$ corresponds to all nulls being true (the global null), and $\gamma = K-1$ to exactly one null remaining.
We seek the policy maximising $\Pi_K$ subject to ${\rm FWER}_\gamma \le \alpha$ for all $\gamma = 0, \ldots, K-1$.
Maximising $\Pi_K$, the average power under the all-alternatives configuration, suits families expected to be signal-rich; the criterion targets that configuration and does not imply dominance for configurations with $\gamma < K$.

Throughout, $\mathcal{L}_{\vec{h}}(\vec{X})$ denotes the joint density (likelihood) of the data $\vec{X}$ when the configuration of true and false nulls is $\vec{h}$, and $\Lambda_k = g(u_k)$ denotes the likelihood ratio at coordinate $k$, with the $p$-value $u_k$ and the alternative density $g$ defined in \S\ref{sec:pvalue}.

\begin{assumption}[$\vec{h}$-Exchangeability]
\label{as:assumption1}
The $K$ tests are $\vec{h}$-exchangeable: $\mathcal{L}_{\vec{h}}(\vec{X}) = \mathcal{L}_{\sigma(\vec{h})}\{\sigma(\vec{X})\}$ for all permutations $\sigma \in S_K$.
\end{assumption}

\begin{sloppypar}
Exchangeability requires that the joint distribution of test statistics depends on the null/alternative labels only through their configuration, not through the identity of the individual hypotheses.
This holds, for example, when the test statistics $X_1, \ldots, X_K$ are conditionally independent given $\vec{h}$, with a common null distribution and a common alternative distribution.
\end{sloppypar}

\begin{assumption}[Arrangement-increasing]
\label{as:assumption2}
For $i \neq j$ with $(\Lambda_i - \Lambda_j)(h_i - h_j) \le 0$, we have $\mathcal{L}_{\vec{h}}(\vec{X}) \le \mathcal{L}_{\vec{h}}\{\sigma_{ij}(\vec{X})\}$, where $\Lambda_k = g(u_k)$ is the likelihood ratio and $\sigma_{ij}$ is the transposition of coordinates $i$ and $j$.
\end{assumption}

This assumption states that the joint density is arrangement-increasing in the likelihood ratios and the alternative indicators.
Under these two assumptions, it suffices to search over symmetric, likelihood-ratio-ordered policies \citep[Theorem~1]{RHPA22}: policies that reject a leading block of hypotheses when sorted by their likelihood ratios.
Symmetry and exchangeability also make ${\rm FWER}_\gamma$ depend on the configuration only through the number $\gamma$ of alternatives; the configurations $\vec{h}_0, \ldots, \vec{h}_{K-1}$ therefore exhaust the non-vacuous family-wise error constraints across all $2^K$ configurations, while $\vec{h}_K$ (all alternatives) has no true null and hence family-wise error rate zero.

\subsection{\texorpdfstring{$P$}{P}-value domain and likelihood ratios}
\label{sec:pvalue}

Let $u_k \in [0,1]$ be a $p$-value for $H_k$: uniformly distributed on $[0,1]$ under the null $H_{0k}$ and with density $g$ under the alternative $H_{Ak}$, so that $g(u)$ is the likelihood ratio of the $p$-value relative to its uniform null distribution.
For a one-sided test with $u_k = F_0(X_k)$, where $X_k \sim F_0$ under the null and $X_k \sim F_A$ under the alternative, this is $g(u) = f_A\{F_0^{-1}(u)\}/f_0\{F_0^{-1}(u)\}$; for two-sided or other tests, $g$ is the corresponding density of the alternative $p$-value.
We work in the \emph{product model}: given the configuration $\vec{h}$, the $p$-values $u_1, \ldots, u_K$ are mutually independent, each null $p$-value uniform on $[0,1]$ and each alternative $p$-value distributed as $g$.
This product likelihood satisfies Assumptions~\ref{as:assumption1}--\ref{as:assumption2} and yields the closed-form coefficients of Lemma~\ref{lem:general_k_coefficients}.
We further assume that $g$ is non-increasing on $[0,1]$ (the monotone-likelihood-ratio condition), so that smaller $p$-values carry larger likelihood ratios.
Exchangeability and the likelihood-ratio ordering then reduce the search to the ordered simplex
\[
Q = \{\vec{u} \in [0,1]^K : u_1 \le u_2 \le \cdots \le u_K\},
\]
in which, by the monotone likelihood ratio, ordering by increasing $p$-value coincides with ordering by decreasing likelihood ratio, so an optimal policy rejects a leading block of the smallest $p$-values.

A symmetric, likelihood-ratio-ordered policy is fully characterised by a function $l^*(\vec{u}) \in \{0, 1, \ldots, K\}$: reject the hypotheses with the $l^*$ smallest $p$-values and retain the rest ($l^* = 0$ rejects nothing).

\subsection{Lagrangian dual}
\label{sec:dual}

Since the search is restricted to symmetric, likelihood-ratio-ordered policies, a policy is characterised by a decision function $D_k(\vec{u}) \in \{0,1\}$ for each position $k$ and data point $\vec{u} \in Q$, where $D_k(\vec{u}) = 1$ means ``reject hypothesis at position $k$'' (the hypothesis with the $k$th smallest $p$-value).

Power and error admit linear representations in the policy \citep[equations~(4) and (7)]{RHPA22}:
\begin{equation}
\label{eq:power_general}
\Pi_K(\vec{D}) = \int_{Q}\sum_{k=1}^{K} a_{k}(\vec{u})\, D_{k}(\vec{u})\, d\vec{u}, \qquad
{\rm FWER}_\gamma(\vec{D}) = \int_{Q}\sum_{k=1}^{K} b_{\gamma,k}(\vec{u})\, D_{k}(\vec{u})\, d\vec{u},
\end{equation}
where the coefficient functions are defined as follows.
The \emph{power weight} $a_k(\vec{u})$ is the marginal contribution to power from rejecting the hypothesis at sorted position $k$ when the data are $\vec{u}$:
\[
a_k(\vec{u}) = (K{-}1)!\prod_{j=1}^K g(u_j),
\]
which is the same for all $k$ (by symmetry, each position contributes equally to power under the all-alternatives configuration $\vec{h}_K$).
The \emph{error weight} $b_{\gamma,k}(\vec{u})$ is the marginal contribution to the family-wise error rate under configuration $\vec{h}_\gamma$ from rejecting position $k$; its closed form is given in Lemma~\ref{lem:general_k_coefficients} below.

Both $a_k$ and $b_{\gamma,k}$ are non-negative on $Q$.
For $a_k$, this is immediate since $g(u) \ge 0$ for all $u$.
For $b_{\gamma,k}$, non-negativity follows from the elementary symmetric polynomial representation (Lemma~\ref{lem:general_k_coefficients}): each $b_{\gamma,k}$ is a product of non-negative quantities ($g$-values, factorials, and an elementary symmetric polynomial of non-negative arguments).
The non-negativity of $b_{\gamma,k}$ was already observed by \citet[\S2.4]{RHPA22}; the closed form in Lemma~\ref{lem:general_k_coefficients} makes it transparent for every $K$, and this non-negativity is what yields the global monotonicity theorem of \S\ref{sec:monotonicity}, which underlies all our subsequent results.

The constrained optimisation problem is to maximise $\Pi_K(\vec{D})$ subject to ${\rm FWER}_\gamma(\vec{D}) \le \alpha$ for $\gamma = 0, \ldots, K-1$.
Introducing Lagrange multipliers $\mu = (\mu_0, \ldots, \mu_{K-1})^{\rm T} \ge 0$, and writing $\mu_{-\gamma}$ for the vector $\mu$ with the $\gamma$th component removed, the Lagrangian is
\[
  L(\vec{D}, \mu) = \sum_{\gamma=0}^{K-1}\mu_\gamma\alpha + \int_Q \sum_{i=1}^K D_i(\vec{u})\, R_i(\mu, \vec{u})\, d\vec{u},
\]
where
\[
R_i(\mu, \vec{u}) = a_i(\vec{u}) - \sum_{\gamma=0}^{K-1}\mu_\gamma\, b_{\gamma,i}(\vec{u})
\]
is the \emph{net benefit} of rejecting hypothesis $i$ at $\vec{u}$: the power gain minus the penalty-weighted error contribution.
For each $\mu$, maximising the Lagrangian over likelihood-ratio-ordered policies yields the induced decision
\[
  D_i^\mu(\vec{u}) = I\bigl(i \le l^*(\mu, \vec{u})\bigr), \qquad
  l^*(\mu, \vec{u}) = \max\Bigl(\argmax_{0 \le l \le K} S_l(\mu, \vec{u})\Bigr),
\]
where $S_l(\mu, \vec{u}) = \sum_{i=1}^l R_i(\mu, \vec{u})$ is the cumulative net benefit (with $S_0 \equiv 0$), and the argmax selects the number of rejections that maximises it, ties broken by the largest maximiser (under Assumption~\ref{as:assumption3}, ties occur only on a Lebesgue-null set of $\vec{u}$; see Supplementary Lemma~\ref{res:C1}).
The dual problem reduces to $\min_{\mu \ge 0} L(\vec{D}^\mu, \mu)$.

Under a non-redundancy condition \citep[Assumption~3]{RHPA22}, strong duality yields a dual minimiser $\mu^*$ whose induced policy $\vec D^{\mu^*}$ is primal-optimal; complementary slackness gives, for each $\gamma$, either $\mu_\gamma^* > 0$ and ${\rm FWER}_\gamma(\vec{D}^{\mu^*}) = \alpha$, or $\mu_\gamma^* = 0$ and ${\rm FWER}_\gamma(\vec{D}^{\mu^*}) \le \alpha$.
The challenge is to compute $\mu^*$.

\section{Main results}
\label{sec:main_result}

\subsection{Family-wise error rate coefficients via elementary symmetric polynomials}
\label{sec:coefficients}

We begin by recalling elementary symmetric polynomials.
The $m$th elementary symmetric polynomial of $(x_1, \ldots, x_n)$ is
\[
e_m(x_1, \ldots, x_n) = \sum_{\substack{S \subseteq [n] \\ |S|=m}} \prod_{i \in S} x_i,
\]
with the convention $e_0 = 1$ and $e_m = 0$ for $m < 0$ or $m > n$.
Key properties that we use below are: (i) $e_m(x_1, \ldots, x_n) \ge 0$ whenever all $x_i \ge 0$; (ii) $e_m$ is non-decreasing in each argument when all arguments are non-negative; and (iii) the recurrence $e_m(x_1, \ldots, x_n) = e_m(x_1, \ldots, x_{n-1}) + x_n \cdot e_{m-1}(x_1, \ldots, x_{n-1})$, which allows $O(nm)$ computation.

\begin{lemma}
\label{lem:general_k_coefficients}
Under the product $p$-value model of Section~\ref{sec:pvalue} (which entails Assumptions~\ref{as:assumption1}--\ref{as:assumption2}), the family-wise error rate constraint coefficients are
\begin{equation}
\label{eq:b_lk_general}
b_{\gamma,k}(\vec{u}) = \gamma!\,(K{-}\gamma)!\; \prod_{j=1}^{k-1} g(u_j) \;\cdot\; e_{\gamma-k+1}\bigl\{g(u_{k+1}), \ldots, g(u_K)\bigr\},
\end{equation}
where $b_{\gamma,k} = 0$ whenever $\gamma < k{-}1$ or $\gamma{-}k{+}1 > K{-}k$.
The power coefficients are $a_k(\vec{u}) = (K{-}1)!\prod_{j=1}^K g(u_j)$, identical for all~$k$.
\end{lemma}

The proof is given in Supplementary Lemma~\ref{res:lemma1}.

\begin{remark}
\label{rem:structure}
The formula~\eqref{eq:b_lk_general} exhibits a factorisation: $b_{\gamma,k}$ depends on the data through a \emph{prefix product} (the likelihood ratios below position $k$) and a \emph{tail symmetric polynomial} (the likelihood ratios above position $k$).
The null hypothesis at position $k$ itself contributes nothing (its $p$-value is uniform with density $1$), creating a clean separation between the evidence below and above the position under consideration.
\end{remark}

\subsection{Optimality conditions}
\label{sec:optimality}

We impose regularity on the alternative density.

\begin{assumption}[Lower Lipschitz]
\label{as:assumption3}
There exists $c_3 > 0$ such that $|g(u) - g(u')| \ge c_3|u - u'|$ for all $u, u' \in [0,1]$.
\end{assumption}

\begin{assumption}[Strict positivity]
\label{as:assumption4}
There exists $c_4 > 0$ such that $g(u) \ge c_4$ for all $u \in [0,1]$.
\end{assumption}

\begin{assumption}[Upper bound]
\label{as:assumption5}
There exists $c_5 > 0$ such that $g(u) \le c_5$ for all $u \in [0,1]$.
\end{assumption}

Assumption~\ref{as:assumption3} ensures that the alternative density differs sufficiently from the uniform (the likelihood ratio is not nearly constant); Assumption~\ref{as:assumption4} bounds the likelihood ratio away from zero (the finite maximisation over $l$ already makes the policy well-defined; the lower bound enters the Lipschitz-stability constants of the Monte Carlo analysis); Assumption~\ref{as:assumption5} is a boundedness condition.

Only the truncated normal among our simulation models satisfies Assumptions~\ref{as:assumption3}--\ref{as:assumption5}.
The mixture-normal and beta alternatives retain a non-increasing likelihood ratio, and hence the population optimality interpretation, but violate regularity conditions used in the convergence analysis.
The Student-$t$ alternative also violates likelihood-ratio monotonicity and is retained only as an out-of-model stress test.
Model-specific details and the precise scope of each guarantee are given in Supplementary Section~\ref{sec:supp_models}.

\begin{theorem}
\label{thm:optimality_general_k}
Under Assumptions~\ref{as:assumption3}--\ref{as:assumption5}, a minimiser $\mu^*$ of the dual satisfies, for each $\gamma = 0, \ldots, K{-}1$:
\[
F_\gamma(\mu^*) := {\rm FWER}_\gamma(\vec{D}^{\mu^*}) = \alpha \quad \text{if } \mu_\gamma^* > 0, \qquad
F_\gamma(\mu^*) \le \alpha \quad \text{if } \mu_\gamma^* = 0.
\]
\end{theorem}

The proof is given in Supplementary Theorem~\ref{res:theorem1}.

\subsection{Global monotonicity}
\label{sec:monotonicity}

The key structural result makes SPOT possible.

\begin{theorem}
\label{thm:general_k_monotonicity}
In the product-model setting of Lemma~\ref{lem:general_k_coefficients}, let $\mu,\mu' \in [0,\infty)^K$ satisfy $\mu \le \mu'$ componentwise.
Then $D_k^{\mu'}(\vec{u}) \le D_k^{\mu}(\vec{u})$ for all $k$ and $\vec{u} \in Q$, and consequently $F_\gamma(\mu') \le F_\gamma(\mu)$ for every $\gamma = 0, \ldots, K{-}1$.
\end{theorem}

The proof is given in Supplementary Theorem~\ref{res:theorem2}.

The monotonicity theorem has several immediate consequences.

\begin{corollary}
\label{cor:unique_root}
Under Assumptions~\ref{as:assumption3} and~\ref{as:assumption5}, for fixed $\mu_{-\gamma}$ the function $\mu_\gamma \mapsto F_\gamma(\mu_\gamma; \mu_{-\gamma})$ is non-increasing and continuous.
If $F_\gamma(0; \mu_{-\gamma}) > \alpha$, then the crossing point $\mu_\gamma^*(\mu_{-\gamma}) = \inf\{\mu_\gamma \ge 0 : F_\gamma(\mu_\gamma; \mu_{-\gamma}) \le \alpha\} \in (0, \infty)$ satisfies $F_\gamma(\mu_\gamma^*; \mu_{-\gamma}) = \alpha$, with $F_\gamma > \alpha$ for $\mu_\gamma < \mu_\gamma^*$ and $F_\gamma \le \alpha$ for $\mu_\gamma \ge \mu_\gamma^*$, and is computable by bisection.
\end{corollary}

The proof is given in Supplementary Corollary~\ref{res:corollary1}.

\begin{remark}
\label{rem:k3_relation}
For $K{=}3$, Theorem~\ref{thm:general_k_monotonicity} subsumes the own-coordinate and cross-coordinate monotonicity results proved individually in \citet{DubeyHuo2024K3} via case analysis of indicator functions; the present proof uses only $b_{\gamma,i} \ge 0$, which holds for all~$K$.
The two approaches are complementary: the explicit Lagrangian decomposition of the $K{=}3$ paper gives finer-grained structural insight for small~$K$, while the elementary symmetric polynomial formulation gives a single, dimension-free proof for arbitrary~$K$.
\end{remark}

\subsection{SPOT and its convergence}
\label{sec:algorithm}

Algorithm~\ref{alg:general_k} presents SPOT.
Corollary~\ref{cor:unique_root} makes the population crossing point well defined; the finite-sample bisection instead acts on the empirical target $\hat{F}_\gamma$, a non-increasing step function by the pathwise monotonicity of Theorem~\ref{thm:general_k_monotonicity}, whose crossing is bracketed and bisected directly (\S\ref{sec:mc_error}).

\begin{algorithm}[!t]
\caption{SPOT: bisection coordinate descent for computing the model-based dual multipliers}
\label{alg:general_k}
\normalsize\onehalfspacing
\begin{algorithmic}[1]
\Require $K$, $\alpha$, alternative density $g$, and Monte Carlo size $N$
\Statex \hspace{\algorithmicindent}Bisection tolerance $\delta$, iterate tolerance $\varepsilon$, sweep cap $T_{\max}$,
\Statex \hspace{\algorithmicindent}and initial upper bracket $U_{\max}$
\State Initialise $\vec{\mu}^{(0)} \gets \vec{0} \in \mathbb{R}^K$ and $s_N \gets \{\alpha(1-\alpha)/N\}^{1/2}$
\State Draw and fix an independent $N$-sample batch under each $\vec{h}_\gamma$, $\gamma=0,\ldots,K{-}1$
\Statex \textbf{Bisection update:} $\Call{UpdateCoordinate}{\gamma,\mu_{-\gamma}}$ returns $0$ if
\Statex \hspace{\algorithmicindent}$\hat F_\gamma(0;\mu_{-\gamma})\le\alpha$; otherwise it brackets the leftmost crossing from
\Statex \hspace{\algorithmicindent}$U_{\max}$ and returns the midpoint of a bracket of width at most $\delta$.
\For{$t=1,\ldots,T_{\max}$}
  \State $\vec{\mu}^{(t)} \gets \vec{\mu}^{(t-1)}$
  \For{$\gamma=0,\ldots,K{-}1$}
    \State $\mu_\gamma^{(t)} \gets \Call{UpdateCoordinate}{\gamma,\mu_{-\gamma}^{(t)}}$
  \EndFor
  \State $r^{(t)} \gets \max_\gamma\bigl|\min\{\mu_\gamma^{(t)},\alpha-\hat F_\gamma(\vec\mu^{(t)})\}\bigr|$
  \State $d_1^{(t)} \gets \|\vec\mu^{(t)}-\vec\mu^{(t-1)}\|_2$
  \State $d_2^{(t)} \gets \|\vec\mu^{(t)}-\vec\mu^{(t-2)}\|_2$ if $t\ge3$
  \If{$r^{(t)}\le s_N$ \textbf{or} $d_1^{(t)}<\varepsilon$ \textbf{or} ($t\ge3$ \textbf{and} $d_2^{(t)}<\varepsilon$)}
    \State \textbf{break}
  \EndIf
\EndFor
\Ensure $\hat{\vec\mu} \gets \vec\mu^{(t)}$ and $D_k^{\hat\mu}(\vec u) \gets I\{k\le l^*(\hat\mu,\vec u)\}$
\end{algorithmic}
\end{algorithm}

Each coordinate update initialises a local bracket at $[0,U_{\max}]$, advances and doubles its upper endpoint until $\hat{F}_\gamma \le \alpha$, and bisects to bracket width $\delta$ on the fixed configuration-$\gamma$ batch $\{\vec{u}^{(n)}\}_{n=1}^N$, reused across all sweeps (common random numbers); the full subroutine is given in Supplementary Section~\ref{sec:supp_impl}.
Because Theorem~\ref{thm:general_k_monotonicity} applies pathwise to every $\vec{u}^{(n)}$, the use of common random numbers ensures that $\hat{F}_\gamma(\mu_\gamma)$ is non-increasing in $\mu_\gamma$ for any fixed batch, so the bisection is well-defined even with finite Monte Carlo samples.
Because each configuration's batch is fixed for the entire run, each $\hat{F}_\gamma$ is a single empirical function and the outer recursion is deterministic given the batches; this is the setting analysed in \S\ref{sec:mc_error} and the Supplementary Material.

\medskip\noindent\textit{Convergence analysis.}

\begin{theorem}[Global convergence of the exact population recursion]
\label{thm:algo_convergence}
Under Assumptions~\ref{as:assumption3}--\ref{as:assumption5}, the dual objective $h(\mu) = \max_{\vec{D}} L(\vec{D}, \mu)$ is convex and continuously differentiable, and the exact population version of SPOT (Algorithm~\ref{alg:general_k} run with the exact target $F_\gamma$ in place of its Monte Carlo estimate $\hat{F}_\gamma$, and with each coordinate minimised exactly rather than to a positive bisection tolerance) performs exact cyclic coordinate minimisation of $h$ over $\mu \ge 0$.
Its iterates satisfy $h(\vec{\mu}^{(t)}) \downarrow \min_{\mu \ge 0} h(\mu)$, and every limit point of $\{\vec{\mu}^{(t)}\}$ is an optimal dual vector $\mu^*$.
No contraction, strong-convexity, Lipschitz-gradient, or active-Jacobian-nonsingularity hypothesis is required (beyond Assumptions~\ref{as:assumption3}--\ref{as:assumption5}).
\end{theorem}

The proof, in Supplementary Theorem~\ref{res:C5}, applies global convergence theory for exact cyclic coordinate minimisation to the continuously differentiable convex dual on its compact sublevel set; coordinate-minimiser uniqueness is not required.

\begin{proposition}[Local linear rate]
\label{prop:linear_rate}
Under Assumptions~\ref{as:assumption3}--\ref{as:assumption5}, let $\mu^*$ be a dual optimum with active set $A = \{\gamma : \mu_\gamma^* > 0\}$ at which strict complementarity holds, the target functions $F_\gamma$ are continuously differentiable, and the active Jacobian $J = [\partial F_\gamma/\partial \mu_\delta(\mu^*)]_{\gamma,\delta \in A}$ is nonsingular.
Then $h$ is locally strongly convex in the active coordinates, and the exact population iterates converge $R$-linearly: $\|\vec{\mu}^{(t)} - \mu^*\|_2 \le C\,\rho^{\,t}$ for some $\rho \in (0,1)$, $C < \infty$, and all large $t$.
\end{proposition}

The nonsingular-Jacobian condition of Proposition~\ref{prop:linear_rate} replaces the global contraction hypothesis of the $K{=}3$ analysis of \citet{DubeyHuo2024K3} by a local, verifiable non-degeneracy condition.
In the scaling runs of \S\ref{sec:convergence} the stopping criterion is met within $4$--$5$ sweeps for every $K \le 12$, with the successive-iterate changes $\|\mu^{(t)} - \mu^{(t-1)}\|_2$ decreasing geometrically (Supplementary Section~\ref{sec:supp_impl}).
The corresponding Monte Carlo guarantee is Theorem~\ref{thm:mc_accuracy}, conditional on its local-regularity and achieved-residual conditions; the centred Gaussian limit additionally requires the stricter $o_p(N^{-1/2})$ residual of Supplementary Theorem~\ref{res:C11}.

\medskip\noindent\textit{Computational complexity.}

\begin{sloppypar}
Each outer iteration performs $K$ bisection updates, each requiring $O(\log(U_{\rm final}/\delta))$ function evaluations, where $U_{\rm final}$ is the final upper bracket after any doubling.
The coefficient tensors $a_k$ and $b_{\gamma,k}$ are precomputed once from the fixed batch at total cost $O(K^4 N)$; with them cached across all sweeps, each outer iteration then costs $O(K^3 N \log(U_{\rm final}/\delta))$. The cache memory footprint and the no-cache variant, which recomputes the tensors each sweep, are detailed in Supplementary Section~\ref{sec:supp_impl}.
\end{sloppypar}

In practice, the dominant cost is evaluating the empirical targets over the Monte Carlo batches.
With $N = 10^5$ and $\delta = 10^{-4}$, SPOT alone takes about $3.1$ seconds for $K = 3$ and $48$ seconds for $K = 12$ on a single core; the full experiment, which also includes the evaluation phase and the baseline procedures, takes about $11$ and $133$ seconds respectively.
SPOT is embarrassingly parallel across Monte Carlo samples.

\subsection{Monte Carlo error}
\label{sec:mc_error}

SPOT (Algorithm~\ref{alg:general_k}) replaces each target function $F_\gamma$ by a Monte Carlo estimate
\begin{equation}
\label{eq:fhat_mc}
\hat{F}_\gamma(\mu) = \frac{1}{N}\sum_{n=1}^{N} I\{V^{(n)}(\vec{D}^{\mu}) > 0\},
\end{equation}
where $V^{(n)}(\vec{D}^\mu)$ is the number of false rejections incurred by the policy $\vec{D}^\mu$ on the $n$th Monte Carlo dataset generated under configuration $\vec{h}_\gamma$. One batch per configuration is drawn before the outer loop, independently across configurations, and reused for all bisection evaluations and all sweeps (common random numbers). Each $\hat{F}_\gamma$ is therefore a fixed empirical target function throughout the run.
Under configuration $\vec{h}_\gamma$, let $\vec{u}(\vec{X})$ be the sorted $p$-values obtained from the data $\vec{X}$ and write $\rho(\vec{X})$ for the sorted position of the first true null; a false rejection occurs precisely when the induced policy rejects down to that position, so
\begin{equation}
\label{eq:mc_indicator}
I\{V > 0\} = I\{l^*(\mu, \vec{u}(\vec{X})) \ge \rho(\vec{X})\},
\end{equation}
which is the event averaged in \eqref{eq:fhat_mc}.
The position $\rho(\vec{X})$ depends on which hypotheses are truly null, not on the sorted $p$-values alone, since a null $p$-value can sort ahead of an alternative; the fixed-index shortcut $I\{l^* \ge \gamma + 1\}$ is generally incorrect.
The population theory of \S\ref{sec:algorithm} does not, however, certify the sweep-level convergence of the recursion: the $K$ empirical targets come from independent batches and so are not the gradient of a single empirical objective. We therefore monitor convergence numerically: the implementation stops when the residual $\max_\gamma|\min\{\mu_\gamma, \alpha - \hat{F}_\gamma(\mu)\}|$ falls below $\{\alpha(1-\alpha)/N\}^{1/2}$, the Monte Carlo standard error of an FWER estimate at the constraint level (the primary criterion). As safeguards for the near-null alternatives whose empirical targets can jump across $\alpha$, it also stops when successive sweeps change $\mu$ by less than $\varepsilon$ in norm or $\mu$ returns to within $\varepsilon$ of its value two sweeps earlier (a period-two check), or, failing these, at the sweep cap $T_{\max}$ (\S\ref{sec:convergence}).

\begin{theorem}[Monte Carlo accuracy conditional on the achieved residual]
\label{thm:mc_accuracy}
Fix $K$ and let $N \to \infty$, and write $A = \{\gamma : \mu^*_\gamma > 0\}$ for the active set at the dual optimum $\mu^*$. Suppose Assumptions~\ref{as:assumption3}--\ref{as:assumption5} hold and, at $\mu^*$, strict complementarity, the local single-crossing condition of Supplementary Corollary~\ref{res:C8}, local continuous differentiability of $F = (F_0, \ldots, F_{K-1})$, and nonsingularity of the active Jacobian hold; suppose also that the bisection tolerance sequence satisfies $\sup_N \delta_N \le \bar{\delta} < \infty$.
Then, for each $\gamma$ and each fixed bounded set $M \subset [0, \infty)^K$, the Monte Carlo estimator satisfies the uniform bound
\begin{equation}
\label{eq:mc_uniform}
\sup_{\mu \in M} |\hat{F}_\gamma(\mu) - F_\gamma(\mu)| = O_p(N^{-1/2}).
\end{equation}
Suppose in addition that SPOT's output sequence $\hat{\mu} = \hat{\mu}_N$ satisfies the complementarity-residual bound
\begin{equation}
\label{eq:mc_residual}
r_N = \max_\gamma \bigl|\min\{\hat{\mu}_\gamma,\ \alpha - \hat{F}_\gamma(\hat{\mu})\}\bigr| = O_p(N^{-1/2}).
\end{equation}
Then
\begin{enumerate}
\item[(a)] $\|\hat{\mu} - \mu^*\| = O_p(N^{-1/2})$;
\item[(b)] $F_\gamma(\hat{\mu}) = \alpha + O_p(N^{-1/2})$ for every $\gamma \in A$;
\item[(c)] ${\rm pr}\{F_\gamma(\hat{\mu}) < \alpha\} \to 1$ for every $\gamma \notin A$.
\end{enumerate}
The conclusion is conditional on the residual bound~\eqref{eq:mc_residual}; it does not assert sweep-level convergence of the empirical recursion, which the population theory of \S\ref{sec:algorithm} does not certify.
\end{theorem}

Supplementary Theorem~\ref{res:C11} gives the proof and also establishes, under a stricter $o_p(N^{-1/2})$ residual, a centred Gaussian limit for $\sqrt{N}(\hat{\mu}_A - \mu^*_A)$ and exact recovery of the active set $A$ by the thresholded estimator $\{\gamma : \hat{\mu}_\gamma > \tau_N\}$ for any $\tau_N \to 0$ with $\sqrt{N}\,\tau_N \to \infty$.
For each fixed, non-random $\mu$, the estimator \eqref{eq:fhat_mc} is an unbiased average of $N$ Bernoulli indicators, with standard error $\{F_\gamma(\mu)[1 - F_\gamma(\mu)]/N\}^{1/2}$; at a constraint-level point where $F_\gamma(\mu) = \alpha$ this is $\{\alpha(1-\alpha)/N\}^{1/2} \approx 6.9\times 10^{-4}$ for $\alpha = 0.05$, $N = 10^5$ (unbiasedness need not survive adaptive same-batch evaluation at the fitted $\hat{\mu}$).
\begin{corollary}[Conservative reduced-level targeting]
\label{cor:conservative}
Fix $0 < \eta < 1/2$ and let $z_\eta$ be the upper-$\eta$ standard normal quantile. Under the conditions of Theorem~\ref{thm:mc_accuracy}, suppose SPOT targets the reduced level $\alpha_N = \alpha - z_\eta\{\alpha(1-\alpha)/N\}^{1/2}$ and its final output $\hat{\mu}$ attains a reduced-target complementarity residual
\[
r_N^{(\alpha_N)} = \max_\gamma \bigl|\min\{\hat{\mu}_\gamma,\ \alpha_N - \hat{F}_\gamma(\hat{\mu})\}\bigr| = o_p(N^{-1/2}).
\]
Then ${\rm pr}\{F_\gamma(\hat{\mu}) \le \alpha\} \to 1 - \eta$ at each active constraint $\gamma \in A$, and replacing $\eta$ by $\eta/K$ gives simultaneous control,
\begin{equation}
\label{eq:simultaneous}
\liminf_{N \to \infty} {\rm pr}\{F_\gamma(\hat{\mu}) \le \alpha \text{ for all } \gamma \in A\} \ge 1 - \eta.
\end{equation}
If $\tilde{\mu}$ is a nominal-target output satisfying $r_N = O_p(N^{-1/2})$ as in Theorem~\ref{thm:mc_accuracy}, the power cost of the reduced target is
\begin{equation}
\label{eq:power_cost}
\bigl|\Pi_K(\vec D^{\hat{\mu}}) - \Pi_K(\vec D^{\tilde{\mu}})\bigr| = O_p(N^{-1/2}).
\end{equation}
\end{corollary}

The proof is given in Supplementary Corollary~\ref{res:C12}.
For the policies computed by SPOT, the maximum empirical family-wise error rates in Table~\ref{tab:scaling} range from $0.050$ to $0.052$ after rounding and are within $0.003$ of $\alpha = 0.05$.

\section{Simulation study}
\label{sec:simulations}

We evaluate the policies computed by SPOT for $K \in \{3, 4, 5, 6, 8, 10, 12\}$, comparing them with Bonferroni, Holm, Hochberg, and Hommel at $\alpha = 0.05$.
All simulations use $N_{\rm eval} = 50{,}000$ Monte Carlo replications for power and FWER evaluation and $N_{\rm opt} = 100{,}000$ samples for the coordinate-descent optimisation.
All reported standard errors are Monte Carlo standard errors, computed as $\hat{\sigma}/\sqrt{N_{\rm eval}}$.

\subsection{The power gap grows with \texorpdfstring{$K$}{K}}
\label{sec:sim_scaling}

The central empirical finding is that the power advantage \emph{increases} with $K$.
Table~\ref{tab:scaling} reports average power $\Pi_K$ and maximum empirical family-wise error rate under the truncated normal model at $\theta = -2.0$.

The gain over Hommel's method grows from $0.083$ (15\%) at $K{=}3$, to $0.169$ (37\%) at $K{=}6$, to $0.276$ (83\%) at $K{=}12$.
While Hommel's $\Pi_K$ drops from $0.555$ to $0.334$ over this range, the power of the policies computed by SPOT remains nearly flat ($0.638$ to $0.610$), reflecting their model-based likelihood weighting with adaptive joint rejection boundaries.

The maximum empirical FWER for the policies computed by SPOT across the settings of Table~\ref{tab:scaling} is $0.052$.
With $N_{\rm eval} = 50{,}000$, the Monte Carlo standard error of an FWER estimate near $0.05$ is $\sqrt{0.05 \times 0.95/50{,}000} \approx 0.001$.
Across the $112$ robustness settings, the largest reported maximum FWER is $0.053$; across the level-sensitivity runs, the largest excess over nominal is $0.003$. These deviations are on the $N^{-1/2}$ Monte Carlo scale. The $0.053$ maximum occurs in an out-of-model Student-$t$ setting, so Theorem~\ref{thm:mc_accuracy} does not apply.

\begin{table}[htbp]
\caption{Scaling with $K$ under the truncated normal alternative ($\theta=-2.0$, $\alpha=0.05$; $N_{\rm eval}=50{,}000$ per configuration): SPOT sustains average power $\Pi_K$ while standard procedures steadily lose power, with maximum empirical FWER near $0.05$. Parentheses give pointwise binomial 95\% confidence intervals for its FWER at the empirically maximising configuration (unadjusted for this selection). Hochberg is omitted; its power lies between Holm's and Hommel's.}
\label{tab:scaling}
\centering
\begin{tabular}{@{}llrrrrr@{}}
\toprule
$K$ & Metric & Bonferroni & Holm & Hommel & SPOT & 95\% CI \\
\midrule
3 & $\Pi_3$ & $0.449$ & $0.529$ & $0.555$ & $\mathbf{0.638}$ & \\
  & max FWER & $0.048$ & $0.048$ & $0.049$ & $0.051$ & $(0.049,0.053)$ \\
4 & $\Pi_4$ & $0.405$ & $0.484$ & $0.518$ & $\mathbf{0.633}$ & \\
  & max FWER & $0.048$ & $0.048$ & $0.049$ & $0.050$ & $(0.049,0.052)$ \\
5 & $\Pi_5$ & $0.372$ & $0.446$ & $0.483$ & $\mathbf{0.630}$ & \\
  & max FWER & $0.049$ & $0.049$ & $0.049$ & $0.052$ & $(0.050,0.054)$ \\
6 & $\Pi_6$ & $0.347$ & $0.415$ & $0.454$ & $\mathbf{0.623}$ & \\
  & max FWER & $0.049$ & $0.049$ & $0.050$ & $0.052$ & $(0.051,0.054)$ \\
8 & $\Pi_8$ & $0.309$ & $0.366$ & $0.404$ & $\mathbf{0.617}$ & \\
  & max FWER & $0.049$ & $0.049$ & $0.049$ & $0.050$ & $(0.048,0.051)$ \\
10 & $\Pi_{10}$ & $0.283$ & $0.330$ & $0.365$ & $\mathbf{0.614}$ & \\
   & max FWER & $0.048$ & $0.048$ & $0.048$ & $0.050$ & $(0.048,0.052)$ \\
12 & $\Pi_{12}$ & $0.262$ & $0.303$ & $0.334$ & $\mathbf{0.610}$ & \\
   & max FWER & $0.047$ & $0.047$ & $0.048$ & $0.051$ & $(0.049,0.053)$ \\
\bottomrule
\end{tabular}
\end{table}
\FloatBarrier

Figure~\ref{fig:scaling} displays the power advantage and computation time as functions of~$K$.

\begin{figure}[htbp]
\centering
\includegraphics[width=30pc]{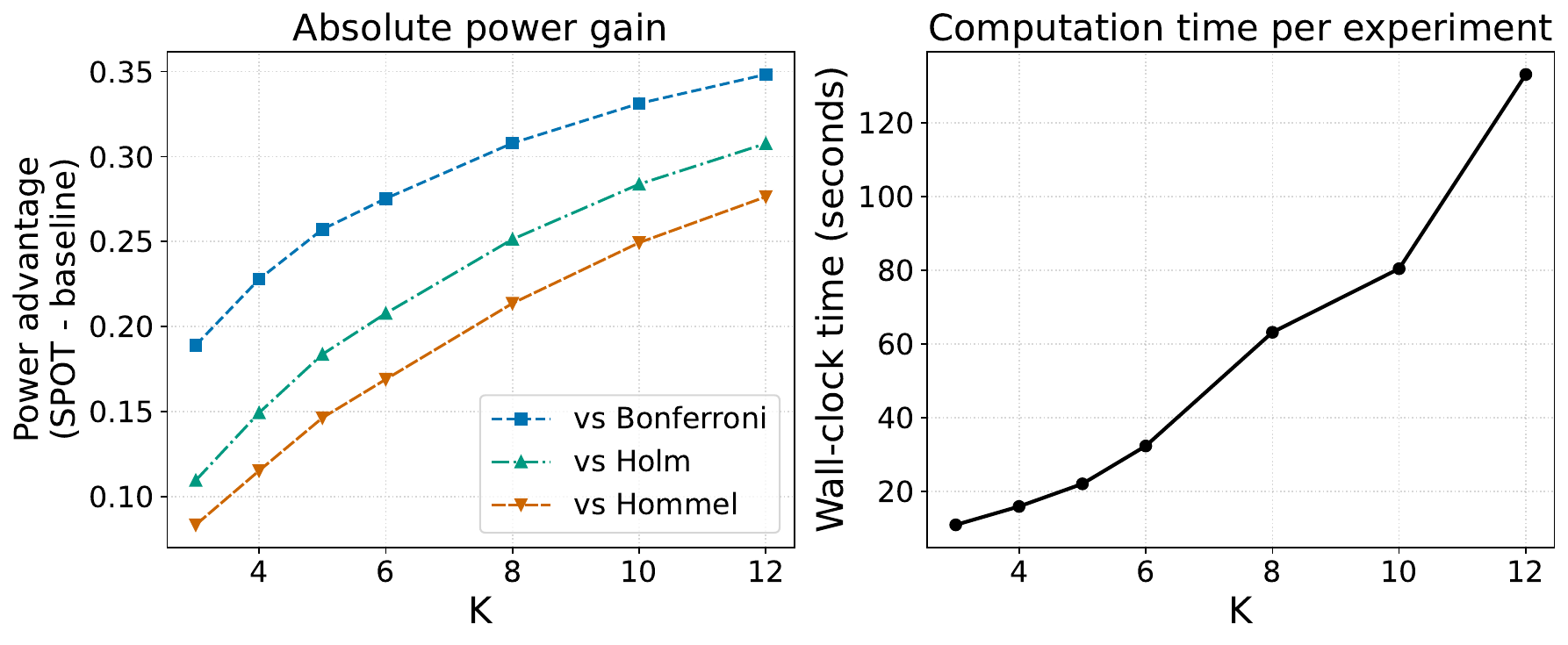}
\caption{Absolute power gain over standard procedures at $\theta=-2.0$ (left) and total wall-clock time per experiment (right) as functions of $K$.
The timing, which is consistent with polynomial scaling, was measured on a single CPU core and includes both the coordinate-descent optimisation ($N_{\rm opt} = 100{,}000$ samples) and the evaluation phase ($N_{\rm eval} = 50{,}000$ samples) of SPOT and the baseline procedures.}
\label{fig:scaling}
\alttext{Left: The average power advantage of SPOT grows with the number of hypotheses. Right: computation time increases approximately polynomially.}
\end{figure}

\subsection{Why the power gap grows}
\label{sec:why_gap_grows}

Standard step-up and step-down procedures compare each ordered $p$-value $u_{(k)}$ to a fixed critical constant depending only on $k$ and $K$ (Holm: $\alpha/(K-k+1)$; Hommel: the most liberal of several such comparisons).
For an exchangeable problem the ordered $p$-values carry all the information available to a symmetric procedure; the distinction is what each does with them.
The baselines use prespecified, distribution-free constants, whereas the optimal decision $l^*(\mu^*, \vec{u})$ maximises the cumulative net benefits $\sum_{i=1}^l R_i(\mu^*, \vec{u})$, in which every $R_i$ weighs all $K$ likelihood ratios $g(u_k)$ through the coefficients $b_{\gamma,k}$, producing data-adaptive joint rejection boundaries.
As $K$ grows the scope for this joint weighting widens, and with it the power gap: in the settings studied, several very small $p$-values can provide enough joint-likelihood evidence to support additional borderline rejections. A step-down procedure such as Holm also rejects when the leading $p$-values are tiny, but its critical constants $\alpha/(K - k + 1)$ are fixed in advance and do not adapt to how small those $p$-values are; it is this adaptation to the joint likelihood magnitudes, not the sequential processing, that the baselines lack.

\subsection{Robustness across distributional settings}
\label{sec:sim_robustness}

We examine four models spanning different tail behaviours and violations of the boundedness assumption: a \emph{truncated normal} ($\theta \in \{-1, -1.5, -2, -2.5, -3, -3.5, -4\}$), the only model satisfying all of Assumptions~\ref{as:assumption3}--\ref{as:assumption5}; a two-sided \emph{mixture normal} and a \emph{beta} ($\theta < 1$) alternative, whose densities are unbounded as $u \to 0$ and so violate Assumption~\ref{as:assumption5}; and, as an out-of-model stress test, a heavy-tailed \emph{Student-$t$} alternative whose likelihood ratio is U-shaped, violating the monotone-likelihood-ratio condition as well.
Full model definitions, parameter grids, and the numerical-truncation details are given in Supplementary Section~\ref{sec:supp_models}.

Across all four models and all $K$ tested, SPOT attains the highest average power $\Pi_K$ of all methods at every signal strength; the gap over the baselines widens from $K = 3$ to $K = 6$ almost everywhere, the sole exception being the strongest truncated normal signal, where both gaps fall below $0.001$.
Figures~\ref{fig:power_curves_K3} and~\ref{fig:power_curves_K6} plot the power curves at $K = 3$ and $K = 6$, with the average power $\Pi_K$ in the left column and the any-discovery power $\Pi_{\rm any}$ in the right: in every family SPOT (solid line) attains the highest average power across the signal range, and comparing the two figures shows the advantage widening as $K$ grows. Hommel is the strongest baseline for average power $\Pi_K$ at all $112$ settings, yet SPOT leads it throughout; the intermediate cases $K = 4$ and $K = 5$ are in Supplementary Figures~\ref{fig:supp_K4} and~\ref{fig:supp_K5}.
The any-discovery power $\Pi_{\rm any}$ (right column) is not the optimised objective: for the out-of-model Student-$t$ alternative at larger degrees of freedom the standard procedures can slightly exceed SPOT on this secondary metric, while in the three monotone-likelihood-ratio product-model families SPOT leads on both.
Convergence diagnostics for all $112$ runs are reported in \S\ref{sec:convergence}.

\begin{figure}[tbp]
\centering
\includegraphics[width=0.972\linewidth]{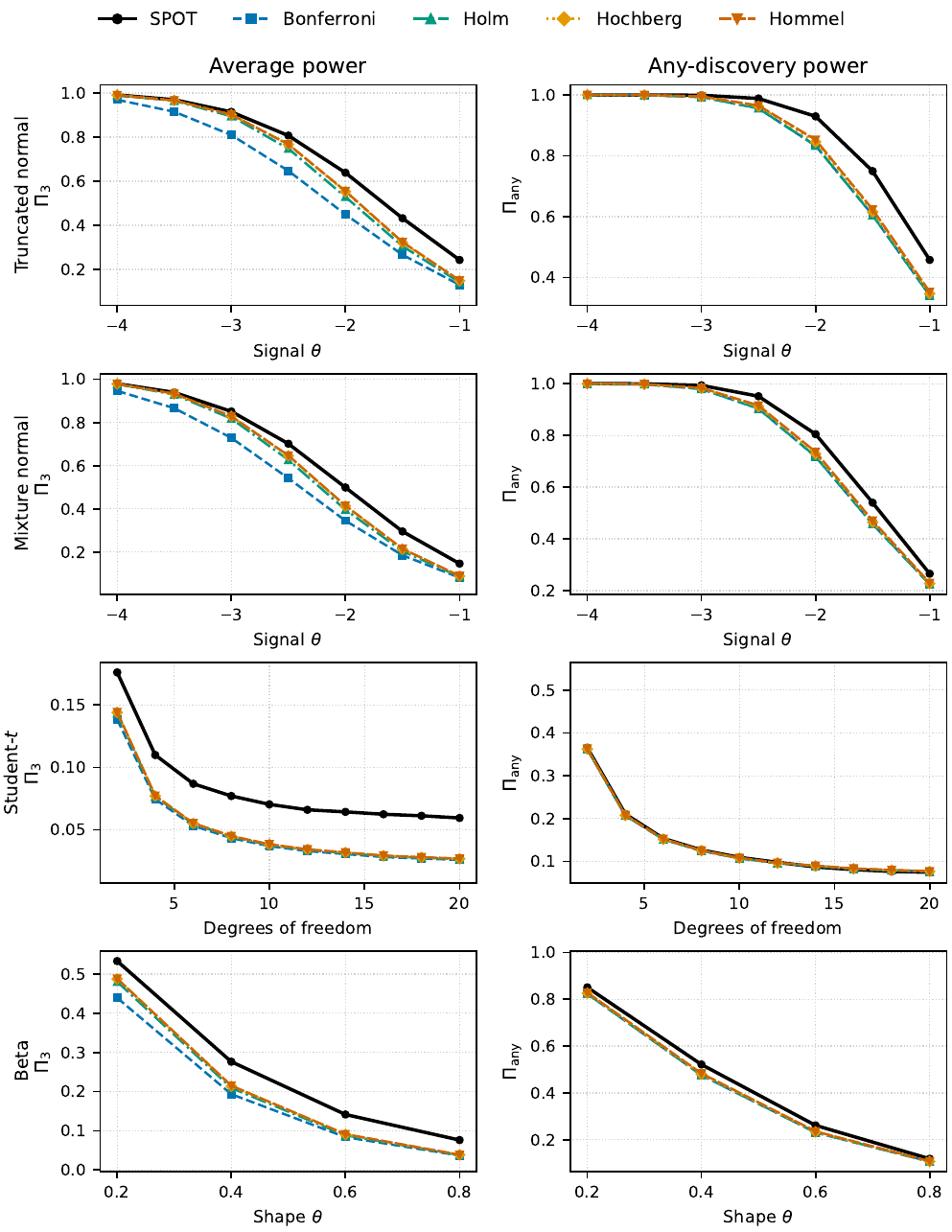}
\caption{Power at $K = 3$ and $\alpha = 0.05$ across four alternatives: optimised average power $\Pi_3$ (left) and any-discovery power $\Pi_{\rm any}$ (right). SPOT leads on $\Pi_3$ throughout.}
\label{fig:power_curves_K3}
\alttext{Left panels: SPOT leads average power throughout. Right panels: SPOT leads except for near ties under Student t.}
\end{figure}

\begin{figure}[tbp]
\centering
\includegraphics[width=0.972\linewidth]{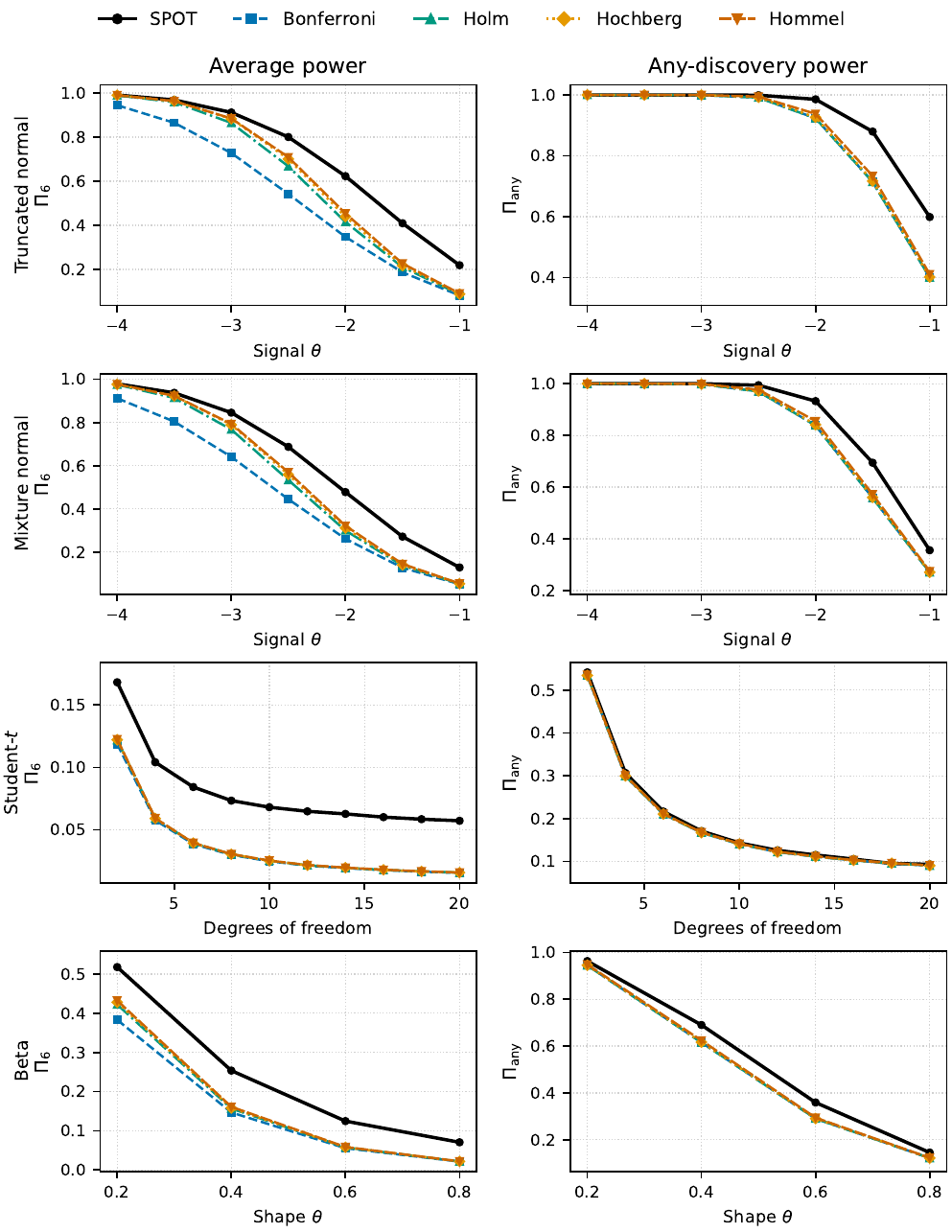}
\caption{Power at $K = 6$, with the layout and scales of Figure~\ref{fig:power_curves_K3}. SPOT leads average power throughout, with a larger margin over Hommel except in the saturated truncated-normal setting ($\theta = -4$).}
\label{fig:power_curves_K6}
\alttext{Left panels: SPOT leads average power throughout. Right panels: baselines sometimes lead under Student t.}
\end{figure}
\FloatBarrier

Although the mixture, beta, and Student-$t$ models fall partly or wholly outside the stated conditions, they show favourable numerical behaviour beyond the proved regime; these experiments do not extend the convergence or optimality guarantees, and the Student-$t$ model additionally lacks the likelihood-ratio-ordered optimality interpretation.

\subsection{Convergence diagnostics}
\label{sec:convergence}

The residual criterion was met within four to five sweeps for every scaling run with $K \le 12$, with $\|\mu^{(t)} - \mu^{(t-1)}\|_2$ decreasing geometrically at empirical successive-iterate ratios of about $0.1$ to $0.5$ after the initial transient, consistent with the linear rate of Proposition~\ref{prop:linear_rate}.
Across the full $112$-setting grid of \S\ref{sec:sim_robustness}, $92$ runs met the one-standard-error residual threshold (median $5$ sweeps among these $92$; $7$ sweeps across all $112$); the remaining $20$, all out-of-model Student-$t$ settings, terminated at residuals of one to three Monte Carlo standard errors.
These are finite-sample diagnostics, not a verification of Theorem~\ref{thm:mc_accuracy}; its asymptotic and local-regularity conditions are not checked here.
Full convergence trajectories (Supplementary Figure~\ref{fig:supp_convergence}) and timings and outer-iteration counts (Supplementary Table~\ref{tab:supp_timing}) are given in Supplementary Section~\ref{sec:supp_impl}; per-setting diagnostics are stored with the released results.

\subsection{Sensitivity analyses}
\label{sec:sensitivity_N}

The power advantage remains substantial across the significance levels examined (Table~\ref{tab:sensitivity_alpha}) and stable across the Monte Carlo sample sizes considered.
Repeating the $K = 6$ truncated normal experiment at $\alpha \in \{0.01, 0.05, 0.10\}$ leaves the relative gain over Hommel's method substantial at every level (78\%, 37\%, and 21\% respectively; Table~\ref{tab:sensitivity_alpha}), largest at the most stringent level, where the joint allocation of the error budget matters most.
Rerunning SPOT with $N_{\rm opt} \in \{50{,}000,\; 100{,}000,\; 200{,}000\}$ changes $\Pi_6$ by less than $0.003$ and the maximum empirical FWER by at most $0.001$, supporting the adequacy of $N_{\rm opt} = 100{,}000$ in this $K = 6$ setting.

\begin{table}[!htbp]
\caption{Sensitivity to the significance level $\alpha$ for the truncated normal alternative at $\theta = -2.0$, $K = 6$ (average power $\Pi_6$). The relative power gain over Hommel's method remains substantial across the levels examined and is largest at the most stringent one. The final column reports SPOT's maximum empirical family-wise error rate over the $K$ null configurations, which stays within $0.003$ of the nominal level $\alpha$ throughout.}
\label{tab:sensitivity_alpha}
\centering
\begin{tabular}{lrrrrr}
\toprule
$\alpha$ & Bonferroni & Holm & Hommel & SPOT & SPOT: max FWER \\
\midrule
$0.01$ & $0.176$ & $0.192$ & $0.202$ & $\mathbf{0.359}$ & $0.011$ \\
$0.05$ & $0.347$ & $0.415$ & $0.454$ & $\mathbf{0.623}$ & $0.052$ \\
$0.10$ & $0.449$ & $0.563$ & $0.622$ & $\mathbf{0.750}$ & $0.102$ \\
\bottomrule
\end{tabular}
\end{table}
\FloatBarrier

\section{Applications}
\label{sec:realdata}

We present three applications to published data at $K = 5$, $6$, and $21$ (a $K = 3$ application is in \citealp{DubeyHuo2024K3}), illustrating the changes in rejection decisions and the role of exchangeability.
Throughout, two-sided $p$-values are modelled with the mixture-normal density of \S\ref{sec:sim_robustness}, $g(u) = \exp(-\theta^2/2)\cosh\{\theta\,\Phi^{-1}(1-u/2)\}$, at a common standardised effect $\theta$ varied in sensitivity analyses.

\subsection{Replicability of social science experiments}
\label{sec:camerer}

\citet{Camerer18} replicated $21$ experimental studies originally published in \emph{Nature} and \emph{Science} between 2010 and 2015, in a two-stage design: Stage~1 was powered at $90\%$ to detect $75\%$ of the original effect size in a two-sided test, and studies that did not replicate in Stage~1 (that is, that failed to attain two-sided $p < 0.05$ in the original direction) proceeded to a larger second stage.
We run SPOT with $K = 21$ jointly on the Stage-1 replication $p$-values of \emph{all} $21$ studies, as reported in Supplementary Table~3 of \citet{Camerer18}.
The family is defined outcome-independently: it comprises every replication conducted in Stage~1, so no hypothesis enters the analysis because of its observed outcome.
Six $p$-values reported as $<0.001$ are set to $0.001$; the rejection sets reported below are unchanged when these six values are instead set to $10^{-10}$, for every assumed effect considered (a check of the two endpoint substitutions, not of the entire interval, since the optimal rejection regions need not be monotone in an individual $p$-value).

\medskip\noindent\textit{Study design.}
Each replication is a two-sided test of the original finding, treated as valid and continuously calibrated under its null (so the null $p$-value is uniform); we model the alternative at a common effect $\theta = 2.0$, varied below.
Exchangeability is an idealisation: the studies share the protocol but not their true effects, so the rejections are model-based decisions, not study-specific truth claims.

\medskip\noindent\textit{Results.}
The rejection decisions are summarised in Table~\ref{tab:camerer}, with the full study-level breakdown in Supplementary Table~\ref{tab:supp_camerer}.
Bonferroni ($\alpha/21 \approx 0.0024$) rejects the six studies with $p = 0.001$; Holm, Hochberg, and Hommel additionally reject the Derex et al.\ replication ($p = 0.003$), for seven rejections.
The policy returned by SPOT rejects nine: the seven baseline rejections plus the Karpicke--Blunt ($p = 0.006$) and Nishi et al.\ ($p = 0.011$) replications.
The two additional rejections reflect the strength-borrowing mechanism of \S\ref{sec:why_gap_grows}: the six $p$-values at $0.001$ contribute overwhelming evidence through the joint likelihood, offsetting the penalty from the moderate $p$-values.
At the reported effect $\theta = 2.0$, the twelve studies with $p \ge 0.022$ remain unrejected, the borrowed strength being insufficient to lift $p$-values this moderate.
The exact population procedure is designed for strong family-wise error control under the working exchangeable product model. Because the mixture-normal alternative used here lies outside the regularity conditions of Theorem~\ref{thm:mc_accuracy} and its local conditions are not verified for this $g$, the reported finite-Monte-Carlo decisions should be interpreted as model-based numerical decisions rather than as theorem-certified family-wise error control.

\medskip\noindent\textit{Sensitivity to the assumed effect.}
As the assumed common effect varies over $\theta \in \{1.5, 2.0, 2.5, 3.0\}$, the rejection count decreases through $12, 9, 8, 8$ (Table~\ref{tab:camerer}; Supplementary Table~\ref{tab:supp_camerer_sensitivity}), since a larger assumed effect concentrates $g$ near $u = 0$ so that moderate $p$-values carry less evidence; at $\theta = 1.5$ the twelve rejections coincide exactly with the twelve studies that replicated individually in Stage~1 of \citet{Camerer18}.
The dependence of the decisions on the assumed $g$ argues for reporting such a sensitivity analysis whenever SPOT is used in practice.

\begin{table}[!htbp]
\caption{Distilled rejection decisions for the $21$ Stage-1 replications of \citet{Camerer18} ($K = 21$, $\alpha = 0.05$), grouped by $p$-value. R, rejection; $\cdot$, no rejection; $\dagger$, reported as ${<}0.001$ and set to $0.001$. Holm, Hochberg, and Hommel coincide. The last four columns give SPOT under the two-sided mixture-normal alternative at the assumed common effect $\theta$; the reported effect is $\theta = 2.0$. Full study-level decisions and effect-size sensitivity are in Supplementary Tables~\ref{tab:supp_camerer} and~\ref{tab:supp_camerer_sensitivity}.}
\label{tab:camerer}
\centering
\setlength{\tabcolsep}{5pt}\renewcommand{\arraystretch}{1.1}
\begin{tabular}{@{}lccccccc@{}}
\toprule
 & & & & \multicolumn{4}{c}{SPOT, $\theta =$} \\
\cmidrule(l){5-8}
Studies & $p$-value & Bonf. & Holm--Hommel & $1.5$ & $2$ & $2.5$ & $3$ \\
\midrule
Six studies & $0.001\,\dagger$ & R & R & R & R & R & R \\
Derex et al.\ (2013) & $0.003$ & $\cdot$ & R & R & R & R & R \\
Karpicke and Blunt (2011) & $0.006$ & $\cdot$ & $\cdot$ & R & R & R & R \\
Nishi et al.\ (2015) & $0.011$ & $\cdot$ & $\cdot$ & R & R & $\cdot$ & $\cdot$ \\
Three studies & $0.022$--$0.025$ & $\cdot$ & $\cdot$ & R & $\cdot$ & $\cdot$ & $\cdot$ \\
Nine studies & $\ge 0.089$ & $\cdot$ & $\cdot$ & $\cdot$ & $\cdot$ & $\cdot$ & $\cdot$ \\
\addlinespace
Total rejected & & $6$ & $7$ & $12$ & $9$ & $8$ & $8$ \\
\bottomrule
\end{tabular}
\end{table}
\FloatBarrier

\subsection{Boundary cases: SPRINT and psychology replications}
\label{sec:sprint}\label{sec:osc}

Two further applications illustrate the boundary of the method's usefulness.
In the SPRINT blood-pressure trial \citep{SPRINT15} ($K = 5$ cardiovascular endpoints), where the endpoints are measured on the same participants and have heterogeneous effects (so the product model's independence and common-effect assumptions fail), the policy returned by SPOT rejects the same two endpoints (heart failure and cardiovascular death, $p = 0.002$ and $0.005$) as every standard procedure.
In this case the null-like endpoints (acute coronary syndrome at $p = 0.99$, stroke at $p = 0.50$) drive the cumulative net benefit negative, so the borderline myocardial-infarction endpoint ($p = 0.19$) is not lifted and the strength-borrowing mechanism offers no advantage. Because the product-model assumptions fail here, no family-wise error guarantee is claimed for this application; the agreement of the observed rejection sets is descriptive.

In a six-study subset of the Reproducibility Project: Psychology \citep{OSC15} ($K = 6$), where the replication $p$-values are uniformly large (smallest $0.079$, far from the Bonferroni threshold $0.05/6 \approx 0.008$), the standard methods reject no hypotheses and SPOT likewise returns no rejections, because there is no strong signal from which the likelihood-based rule can borrow strength.
The contrast with the Camerer application, where six overwhelming $p$-values anchor the joint likelihood and lift moderate $p$-values into rejection, is instructive: joint modelling amplifies collective evidence but does not manufacture individual signal.
Endpoint lists, the hierarchical SPRINT analysis (which recovers the two Group-A rejections while confining exchangeability to more plausible groups), and the OSC subset construction are given in Supplementary Section~\ref{sec:supp_apps}.

\section{Discussion}
\label{sec:conclusion}

\subsection{Summary of contributions}

We have shown that the dual constraint coefficients of \citet{RHPA22} admit a closed form through elementary symmetric polynomials of the likelihood ratios (Lemma~\ref{lem:general_k_coefficients}), whose non-negativity yields global monotonicity of all target functions for arbitrary $K$, without case analysis (Theorem~\ref{thm:general_k_monotonicity}).
SPOT makes the optimal dual multipliers computationally accessible for arbitrary~$K$: its exact population recursion drives the dual objective to its optimum with every limit point optimal, needing no contraction or strong-convexity hypothesis (Theorem~\ref{thm:algo_convergence}); its exact population iterates converge linearly near the optimum under the local non-degeneracy conditions of Proposition~\ref{prop:linear_rate}; and its Monte Carlo output is $\sqrt{N}$-consistent under the conditions of Theorem~\ref{thm:mc_accuracy} and its achieved $O_p(N^{-1/2})$ residual (\S\ref{sec:mc_error}).
The power gains are substantial and, in the truncated-normal scaling experiment, grow with~$K$ when the exchangeable product model of \S\ref{sec:pvalue} holds, reaching an 83\% gain over Hommel at $K{=}12$ (Table~\ref{tab:scaling}) at zero additional cost in data collection; when it is misspecified, as in the SPRINT application (\S\ref{sec:sprint}), the fitted rejection set happens to coincide with those of the standard procedures, and no family-wise error guarantee is claimed.
SPOT therefore makes model-assisted optimal family-wise error control computationally accessible: any policy induced by a limit point of the exact population recursion is optimal within the exchangeable product model, while its Monte Carlo output additionally requires the local conditions and achieved residual of Theorem~\ref{thm:mc_accuracy}.

\subsection{Limitations}

The main restriction is the exchangeable product model of Section~\ref{sec:pvalue}: the $p$-values are assumed mutually independent, all hypotheses share the same null distribution and the same alternative density $g$, and $g$ is non-increasing, so that ordering by $p$-value coincides with ordering by likelihood ratio. This independence and common-$g$ structure of the product model (\S\ref{sec:pvalue}), beyond the exchangeability of Assumption~\ref{as:assumption1}, is what delivers the closed-form coefficients of Lemma~\ref{lem:general_k_coefficients}.
This can be plausible for prespecified, non-overlapping subgroups with comparable designs, but excludes settings where hypotheses have different effect sizes or different sample sizes.
The theory treats $g$ and its parameters as fixed and correctly specified. If $g$ is estimated, even from independent training data or by sample splitting, additional robustness or asymptotic analysis is needed to account for estimation error and preserve the stated guarantees; reuse of the testing $p$-values introduces further dependence.

Convergence of the exact population recursion's objective is global under Assumptions~\ref{as:assumption3}--\ref{as:assumption5}, needing no contraction, strong-convexity, Lipschitz-gradient, or active-Jacobian-nonsingularity hypothesis (Theorem~\ref{thm:algo_convergence}); the \emph{linear rate} of Proposition~\ref{prop:linear_rate} additionally requires strict complementarity, continuous differentiability of the target functions, and a nonsingular active Jacobian at $\mu^*$, standard non-degeneracy conditions we do not verify a priori but that are consistent with the geometric decrease seen in the truncated-normal scaling runs of Supplementary Figure~\ref{fig:supp_convergence}.
Establishing them under primitive conditions on $g$ would strengthen the rate guarantee.

Finally, the Monte Carlo evaluation introduces sampling error. Under the conditions of Theorem~\ref{thm:mc_accuracy}, an empirical output whose achieved complementarity residual is $O_p(N^{-1/2})$ is $\sqrt{N}$-consistent for $\mu^*$, and its realised FWER at active constraints is $\alpha + O_p(N^{-1/2})$. Under the stricter residual conditions of Corollary~\ref{cor:conservative}, reduced-level targeting gives marginal control with asymptotic probability $1-\eta$, or simultaneous control with asymptotic probability at least $1-\eta$ after replacing $\eta$ by $\eta/K$, at an $O_p(N^{-1/2})$ power cost (\S\ref{sec:mc_error}).

\subsection{Extensions}

Several directions for future work are natural.

\emph{Relaxing exchangeability.}
A hierarchical approach \citep{RHPA22} could extend applicability: partition hypotheses into groups that are each exchangeable, with independent $p$-values and a common non-increasing alternative density (so that the within-group product model of \S\ref{sec:pvalue} holds), apply the optimal policy within each at level $\alpha/G$, and aggregate; this preserves the within-group guarantees while accommodating heterogeneity across groups, the Bonferroni loss partially offset by the within-group gain (\S\ref{sec:sprint}).
A fully non-exchangeable extension would need coefficients depending on the hypotheses' identities; the elementary symmetric polynomial structure would not carry over directly, though the monotonicity principle might.

\emph{Relationship to false discovery rate control.}
The family-wise error rate and the false discovery rate \citep{BH95} control different quantities (${\rm pr}(V > 0) \le \alpha$ versus ${\rm E}\{V/\max(R, 1)\} \le q$) and suit different regimes: family-wise control when even a single false rejection is costly and $K$ is moderate \citep{Bretz09, Camerer18, list16}, false-discovery-rate control for large-$K$ screening where individual false positives are tolerable.
The choice is orthogonal to our contribution: within the exchangeable product model, SPOT makes computationally accessible the procedure that maximises the specified average-power criterion under strong family-wise error control, using model-based likelihood weighting that produces the 15--83\% gains in the truncated-normal scaling experiment of Section~\ref{sec:sim_scaling}.
The Benjamini--Hochberg procedure, by contrast, adapts through a data-dependent rejection rank and costs only $O(K\log K)$, needing only valid (super-uniform) true-null $p$-values under mutual independence of the full $p$-value vector, or the appropriate positive-regression-dependence (PRDS) condition, whereas our method requires the exchangeable product model with a common non-increasing $g$ (Supplementary Table~\ref{tab:supp_camerer_sensitivity} shows the decisions can depend on $g$).
The linear structure of \eqref{eq:power_general} suggests analogous dual formulations for other error metrics, including $k$-family-wise error rate \citep{romano}, false discovery rate, and weighted power objectives; whether the elementary symmetric polynomial structure and its monotonicity carry over is an open question.

\emph{Large $K$.}
While our simulations extend to $K{=}12$ and the largest application to $K{=}21$, the theory applies to arbitrary~$K$.
Here ``general-$K$'' refers to arbitrary-$K$ algebra and polynomial-time computation in the moderate-$K$ regime where strong family-wise control is typically relevant, not to the many-thousands-of-hypotheses scale of false-discovery-rate screening.
For very large $K$, two computational challenges arise: the elementary symmetric polynomials may overflow when the $g$-values span many orders of magnitude (log-domain arithmetic addresses this), and the per-iteration cost ($O(K^3 N \log(U_{\rm final}/\delta))$ with cached coefficient tensors, or $O(K^4 N + K^3 N \log(U_{\rm final}/\delta))$ when they must be recomputed each sweep) may become prohibitive (saddlepoint or importance-sampling approximations could reduce the variance or the required Monte Carlo size).

\section*{Supplementary material}
The supplementary material accompanying this article contains all proofs and the Monte Carlo analysis, the implementation and convergence diagnostics, the simulation and application details, and code documentation (Sections~\ref{sec:supp_proofs}--\ref{sec:supp_code}).

\section*{Data availability}
All data analysed are published summary statistics from the cited sources \citep{Camerer18, SPRINT15, OSC15} (the Reproducibility Project data file is at \url{https://osf.io/fgjvw/}). The code, data, and supplementary materials reproducing every result accompany this submission, and will also be made publicly available in a GitHub repository, archived with a digital object identifier, before acceptance, in accordance with the Series~B code policy.

\section*{Declaration of the use of generative AI and AI-assisted technologies}
We used generative AI tools (ChatGPT, Google Gemini, and Claude) for language editing and code formatting support only. All data, results and mathematical derivations are the authors' own work.

\section*{Funding}
Dubey acknowledges partial support from the Stewart Topper Fellowship at Georgia Tech.
Huo was partially supported by the National Science Foundation under Grant~2229876, the A.~Russell Chandler III Professorship, and the NIH-sponsored Georgia Clinical and Translational Science Alliance.

\section*{Competing interests}
The authors declare no competing interests.

\bibliographystyle{abbrvnat}
\bibliography{cite}

\clearpage

\setcounter{section}{0}
\setcounter{subsection}{0}
\setcounter{subsubsection}{0}
\setcounter{equation}{0}
\setcounter{table}{0}
\setcounter{figure}{0}
\renewcommand{\thesection}{S\arabic{section}}
\renewcommand{\thetable}{S\arabic{table}}
\renewcommand{\thefigure}{S\arabic{figure}}
\renewcommand{\theequation}{\arabic{equation}}
\renewcommand{\theHsection}{S.\arabic{section}}
\renewcommand{\theHsubsection}{\theHsection.\arabic{subsection}}
\renewcommand{\theHsubsubsection}{\theHsubsection.\arabic{subsubsection}}
\renewcommand{\theHtable}{S.\arabic{table}}
\renewcommand{\theHfigure}{S.\arabic{figure}}
\renewcommand{\theHequation}{S.\arabic{equation}}

\arxivrestoretitlecommands
\title{\Large Supplementary Material for:\\
Optimal multiple testing under family-wise error control:\\
elementary symmetric polynomials and a scalable algorithm}
\author{Prasanjit Dubey \quad and \quad Xiaoming Huo\\[4pt]
\normalsize H.~Milton Stewart School of Industrial and Systems Engineering,\\
\normalsize Georgia Institute of Technology, Atlanta, GA 30332, USA}
\date{}
\maketitle

This supplement contains proofs of all theoretical results stated in the main paper (Section~\ref{sec:supp_proofs}), the algorithm implementation, computational complexity, and convergence diagnostics (Section~\ref{sec:supp_impl}), the simulation models and scope of assumptions (Section~\ref{sec:supp_models}), sensitivity analyses (Section~\ref{sec:supp_sensitivity}), additional power curves for $K = 3, \ldots, 6$ (Section~\ref{sec:supp_power}), application details (Section~\ref{sec:supp_apps}), and documentation of the simulation code (Section~\ref{sec:supp_code}).
The code reproducing all results accompanies this submission and will be made publicly available in a GitHub repository, archived with a digital object identifier, before acceptance.
We use the notation of the main paper throughout.
Throughout, SPOT denotes the bisection coordinate-descent algorithm in Algorithm~\ref{alg:general_k} of the main paper.
In particular, $\mu_{-\gamma}$ denotes the vector $\mu$ with the $\gamma$th component removed.
Throughout, the significance level is fixed with $0 < \alpha < 1$, and the number of hypotheses $K$ is held fixed as the Monte Carlo sample size $N \to \infty$.

\section{Proofs}
\label{sec:supp_proofs}

\subsection{Proof of \texorpdfstring{Lemma~\ref{lem:general_k_coefficients}}{Lemma 1}}

\begin{lemmaS}{1}[Coefficient characterisation; restated]
\label{res:lemma1}
Under the product $p$-value model of Section~\ref{sec:pvalue} (Assumptions~\ref{as:assumption1}--\ref{as:assumption2}), the family-wise error rate constraint coefficients are
\[
b_{\gamma,k}(\vec{u}) = \gamma!\,(K{-}\gamma)!\; \prod_{j=1}^{k-1} g(u_j) \;\cdot\; e_{\gamma-k+1}\bigl\{g(u_{k+1}), \ldots, g(u_K)\bigr\},
\]
where $b_{\gamma,k} = 0$ whenever $\gamma < k{-}1$ or $\gamma{-}k{+}1 > K{-}k$.
The power coefficients are $a_k(\vec{u}) = (K{-}1)!\prod_{j=1}^K g(u_j)$, identical for all~$k$.
\end{lemmaS}

\begin{proof}
We derive $b_{\gamma,k}(\vec{u})$ by decomposing the family-wise error rate under configuration $\vec{h}_\gamma$ according to the role of each hypothesis.

Under $\vec{h}_\gamma$, there are $\gamma$ alternatives and $K - \gamma$ nulls.
The coefficient $b_{\gamma,k}(\vec{u})$ aggregates the contribution to the family-wise error rate of rejecting hypothesis $k$ (in the sorted order), across the configurations in which $k$ is the smallest-indexed null hypothesis in the sorted order, the null whose rejection triggers the family-wise error.

Fix a configuration where $k$ is the smallest-indexed null hypothesis in the sorted order.
This requires:
\begin{enumerate}[label=(\alph*)]
\item All indices $1, \ldots, k-1$ (those with smaller $p$-values) must be alternatives.
This is possible only if $k - 1 \le \gamma$, i.e., $\gamma \ge k - 1$.
Each such alternative contributes a factor of $g(u_j)$ to the joint density on $Q$.

\item Among the $K - k$ indices $k+1, \ldots, K$ (those with larger $p$-values), exactly $\gamma - (k-1) = \gamma - k + 1$ must be alternatives.
This is possible only if $0 \le \gamma - k + 1 \le K - k$.
\end{enumerate}

When $\gamma < k - 1$ or $\gamma - k + 1 > K - k$, no valid assignment exists and $b_{\gamma,k} = 0$.
Otherwise, the contribution from the $k - 1$ alternatives below $k$ is
\[
\prod_{j=1}^{k-1} g(u_j).
\]
The contribution from the $K - k$ indices above $k$ requires summing over all $\binom{K-k}{\gamma-k+1}$ ways to choose which $\gamma - k + 1$ of them are alternatives.
For each such selection $S \subseteq \{k+1, \ldots, K\}$ with $|S| = \gamma - k + 1$, the alternatives in $S$ contribute $\prod_{j \in S} g(u_j)$ while the nulls contribute $1$ (since null $p$-values are uniform).
Summing over all selections gives
\[
\sum_{\substack{S \subseteq \{k+1, \ldots, K\} \\ |S| = \gamma-k+1}} \prod_{j \in S} g(u_j) = e_{\gamma-k+1}\bigl\{g(u_{k+1}), \ldots, g(u_K)\bigr\}.
\]
The prefactor $\gamma!(K-\gamma)!$ accounts for the number of permutations of the $\gamma$ alternatives among themselves and the $K - \gamma$ nulls among themselves, arising from the reduction to the ordered simplex $Q$.

For the power coefficient, under $\vec{h}_K$ all hypotheses are alternatives, so
\[
a_k(\vec{u}) = K! \cdot K^{-1} \cdot \prod_{j=1}^K g(u_j) = (K-1)! \prod_{j=1}^K g(u_j),
\]
where $K!$ is the permutation factor and $K^{-1}$ comes from the averaging in the power definition. This is independent of $k$.
\end{proof}

\medskip\noindent\textit{Worked example for $K = 3$.}
To make the coefficient formula of Lemma~\ref{res:lemma1} concrete, we write out the full coefficient matrix for $K = 3$; let $g_j = g(u_j)$.
The constraints are indexed by $\gamma \in \{0, 1, 2\}$ and the positions by $k \in \{1, 2, 3\}$, giving
\begin{gather*}
b_{0,1} = 6, \quad b_{0,2} = 0, \quad b_{0,3} = 0, \\
b_{1,1} = 2(g_2 + g_3), \quad b_{1,2} = 2g_1, \quad b_{1,3} = 0, \\
b_{2,1} = 2g_2 g_3, \quad b_{2,2} = 2g_1 g_3, \quad b_{2,3} = 2g_1 g_2,
\end{gather*}
and power coefficients $a_k = 2!\prod_{j=1}^3 g_j = 2g_1 g_2 g_3$ for all $k$.
These match the explicit Lagrangian decomposition derived by case analysis in \citet{DubeyHuo2024K3}, confirming the general formula.

\subsection{Proof of \texorpdfstring{Theorem~\ref{thm:optimality_general_k}}{Theorem 1}}

\begin{theoremS}{1}[Optimality conditions; restated]
\label{res:theorem1}
Under Assumptions~\ref{as:assumption3}--\ref{as:assumption5}, a minimiser $\mu^*$ of the dual satisfies, for each $\gamma = 0, \ldots, K{-}1$:
\[
F_\gamma(\mu^*) := {\rm FWER}_\gamma(\vec{D}^{\mu^*}) = \alpha \quad \text{if } \mu_\gamma^* > 0, \qquad
F_\gamma(\mu^*) \le \alpha \quad \text{if } \mu_\gamma^* = 0.
\]
\end{theoremS}

\begin{proof}
This is complementary slackness for the linear program over policies.
Since the power and error representations are linear in $\vec{D}$, the dual objective is
\[
L(\vec{D}^\mu, \mu) = \Pi_K(\vec{D}^\mu) + \sum_{\gamma=0}^{K-1} \mu_\gamma\bigl\{\alpha - F_\gamma(\mu)\bigr\}.
\]
The dual function $h(\mu) = \max_{\vec{D}} L(\vec{D}, \mu)$ is convex in $\mu$ as a pointwise maximum of linear functions.
By Lemma~\ref{res:C1} and Corollary~\ref{res:C2} (established in the convergence section below, independently of this theorem), the maximising policy is unique for almost every $\vec{u}$ and $h$ is differentiable, the Danskin subgradient being the singleton $\nabla h(\mu^*)$ with $\partial h/\partial\mu_\gamma = \alpha - F_\gamma(\mu^*)$.
Because $h$ is minimised over $\mu \ge 0$, the Karush--Kuhn--Tucker conditions require $v := \nabla h(\mu^*)$ to satisfy $v_\gamma \ge 0$ and $v_\gamma\, \mu_\gamma^* = 0$ for every $\gamma$, with $v_\gamma = \alpha - F_\gamma(\mu^*)$.

If $\mu_\gamma^* > 0$, complementary slackness forces $v_\gamma = 0$, hence $F_\gamma(\mu^*) = \alpha$.
If $\mu_\gamma^* = 0$, the condition $v_\gamma \ge 0$ gives $F_\gamma(\mu^*) \le \alpha$.
\end{proof}

\subsection{Proof of \texorpdfstring{Theorem~\ref{thm:general_k_monotonicity}}{Theorem 2}}

\begin{theoremS}{2}[Global monotonicity; restated]
\label{res:theorem2}
In the product-model setting of Lemma~\ref{res:lemma1}, let $\mu,\mu' \in [0,\infty)^K$ satisfy $\mu \le \mu'$ componentwise.
Then $D_k^{\mu'}(\vec{u}) \le D_k^{\mu}(\vec{u})$ for all $k$ and $\vec{u} \in Q$, and consequently $F_\gamma(\mu') \le F_\gamma(\mu)$ for every $\gamma = 0, \ldots, K{-}1$.
\end{theoremS}

\begin{proof}
The argument proceeds in three steps, relying only on the non-negativity of the coefficients $b_{\gamma,k}$ established in Lemma~\ref{res:lemma1}.

\emph{Step~1: $R_i$ decreases in $\mu$.}
Since $\mu_\gamma' \ge \mu_\gamma$ for all $\gamma$ and $b_{\gamma,i}(\vec{u}) \ge 0$ for all $\gamma$, $i$ and $\vec{u}$, we have
\[
R_i(\mu', \vec{u}) = a_i(\vec{u}) - \sum_{\gamma=0}^{K-1} \mu_\gamma'\, b_{\gamma,i}(\vec{u}) \le a_i(\vec{u}) - \sum_{\gamma=0}^{K-1} \mu_\gamma\, b_{\gamma,i}(\vec{u}) = R_i(\mu, \vec{u})
\]
for every $i = 1, \ldots, K$ and $\vec{u} \in Q$.

\emph{Step~2: $l^*$ decreases in $\mu$.}
Define the partial sums $S_l(\mu) = \sum_{i=1}^l R_i(\mu, \vec{u})$ for $l \ge 1$ and $S_0 = 0$.
By Step~1, $S_l(\mu') \le S_l(\mu)$ for all $l \ge 1$, while $S_0(\mu') = S_0(\mu) = 0$.
Moreover, the gap $\Delta_l := S_l(\mu) - S_l(\mu')$ is non-decreasing in $l$ with $\Delta_0 = 0$, since $\Delta_l - \Delta_{l-1} = R_l(\mu) - R_l(\mu') \ge 0$ by Step~1.

We claim $l^*(\mu') \le l^*(\mu)$.
Define the ``acceptance set'' $A(\mu) = \{l \ge 0 : S_l(\mu) \ge S_j(\mu) \text{ for all } j\}$ and let $l^*(\mu) = \max A(\mu)$ (taking the largest maximiser as a tie-breaking convention).

For any $l \in A(\mu')$ with $l \ge 1$, we have $S_l(\mu') \ge S_0(\mu') = 0$, and thus $S_l(\mu) \ge S_l(\mu') \ge 0$.
This means $l$ achieves a non-negative value under $\mu$, so $\max_{j \ge 0} S_j(\mu) \ge 0$ and therefore $l^*(\mu) \ge 1$.

Now let $m = l^*(\mu')$.
Then $S_m(\mu) \ge S_m(\mu') \ge 0 = S_0(\mu)$, so the maximum of $\{S_l(\mu)\}_{l=0}^K$ is at least $S_m(\mu) \ge 0$.
Suppose for contradiction that $l^*(\mu) < m$.
Since $l^*(\mu) = \max A(\mu)$, the maximum property gives $S_{l^*(\mu)}(\mu) \ge S_m(\mu)$.
Conversely, using $\Delta_l$ non-decreasing and $m \in A(\mu')$,
\[
S_m(\mu) - S_{l^*(\mu)}(\mu) = \underbrace{\bigl[S_m(\mu') - S_{l^*(\mu)}(\mu')\bigr]}_{\ge\,0} + \underbrace{\bigl[\Delta_m - \Delta_{l^*(\mu)}\bigr]}_{\ge\,0} \;\ge\; 0.
\]
Combining these inequalities yields $S_m(\mu) = S_{l^*(\mu)}(\mu)$, hence $m \in A(\mu)$ with $m > l^*(\mu)$, contradicting $l^*(\mu) = \max A(\mu)$.
Therefore $l^*(\mu) \ge m = l^*(\mu')$.

\emph{Step~3: $F_\gamma$ decreases in $\mu$.}
Since $D_k^\mu(\vec{u}) = I(k \le l^*)$ and $l^*$ is non-increasing in $\mu$, $D_k^{\mu'}(\vec{u}) \le D_k^{\mu}(\vec{u})$ for all $k$ and $\vec{u}$.
The family-wise error rate functional is
\[
F_\gamma(\mu) = \int_Q \sum_{k=1}^K b_{\gamma,k}(\vec{u})\, D_k^\mu(\vec{u})\, d\vec{u}.
\]
Since $b_{\gamma,k} \ge 0$ and $D_k^{\mu'} \le D_k^\mu$, we obtain $F_\gamma(\mu') \le F_\gamma(\mu)$.
\end{proof}

\subsection{Proof of \texorpdfstring{Corollary~\ref{cor:unique_root}}{Corollary 1}}

\begin{corollaryS}{1}[Leftmost crossing point; restated]
\label{res:corollary1}
Suppose Assumptions~\ref{as:assumption3} and~\ref{as:assumption5} hold. For fixed $\mu_{-\gamma}$, the function $\mu_\gamma \mapsto F_\gamma(\mu_\gamma; \mu_{-\gamma})$ is non-increasing and continuous.
If $F_\gamma(0; \mu_{-\gamma}) > \alpha$, then the crossing point $\mu_\gamma^*(\mu_{-\gamma}) = \inf\{\mu_\gamma \ge 0 : F_\gamma(\mu_\gamma; \mu_{-\gamma}) \le \alpha\} \in (0, \infty)$ satisfies $F_\gamma(\mu_\gamma^*; \mu_{-\gamma}) = \alpha$, with $F_\gamma > \alpha$ for $\mu_\gamma < \mu_\gamma^*$ and $F_\gamma \le \alpha$ for $\mu_\gamma \ge \mu_\gamma^*$, and is computable by bisection.
\end{corollaryS}

\begin{proof}
Non-increasing follows from Theorem~\ref{res:theorem2} applied with $\mu$ and $\mu'$ differing only in coordinate $\gamma$.
If $\mu_\gamma < \mu_\gamma'$, then $\mu \le \mu'$ componentwise, so $F_\gamma(\mu') \le F_\gamma(\mu)$.

For the limiting behaviour as $\mu_\gamma \to \infty$: $F_\gamma(0; \mu_{-\gamma}) > \alpha$ by hypothesis.
For each $l$, write $B_{\gamma,l}(\vec{u}) = \sum_{i=1}^l b_{\gamma,i}(\vec{u}) \ge 0$, so that the cumulative net benefit splits as
\[
S_l(\mu, \vec{u}) = \sum_{i=1}^l R_i(\mu,\vec{u}) = C_l(\mu_{-\gamma}, \vec{u}) - \mu_\gamma\, B_{\gamma,l}(\vec{u}),
\]
where $C_l$ collects the terms that do not involve $\mu_\gamma$.
As $\mu_\gamma \to \infty$, every candidate $l$ with $B_{\gamma,l}(\vec{u}) > 0$ has $S_l \to -\infty$, so for almost every $\vec{u}$ the selected maximiser $l^*(\mu,\vec{u})$ eventually lies among the candidates with $B_{\gamma,l^*}(\vec{u}) = 0$.
Since $F_\gamma(\mu) = \int_Q \sum_{k=1}^K b_{\gamma,k}(\vec{u})\, D_k^\mu(\vec{u})\, d\vec{u} = \int_Q B_{\gamma,l^*(\mu,\vec{u})}(\vec{u})\, d\vec{u}$, and $B_{\gamma,l^*} \le \sum_{k=1}^K b_{\gamma,k}$ is integrable (at $\mu = \vec{0}$ every position is rejected, so $\int_Q \sum_k b_{\gamma,k}\,d\vec{u} = {\rm FWER}_\gamma(\vec{D} \equiv \vec{1}) \le 1$), dominated convergence gives $F_\gamma(\mu_\gamma; \mu_{-\gamma}) \to 0$.
This uses only the non-negativity of the coefficients, not strict positivity of $g$.

Define the crossing point
\[
\mu_\gamma^* := \inf\{\mu_\gamma \ge 0 : F_\gamma(\mu_\gamma; \mu_{-\gamma}) \le \alpha\}.
\]
By Corollary~\ref{res:C2} (Assumptions~\ref{as:assumption3} and~\ref{as:assumption5}), $F_\gamma(\cdot\,;\mu_{-\gamma}) = \alpha - \partial h/\partial\mu_\gamma$ is continuous. Since $F_\gamma(0;\mu_{-\gamma}) > \alpha$ and $F_\gamma(\mu_\gamma;\mu_{-\gamma}) \to 0 < \alpha$, the intermediate value theorem gives $\mu_\gamma^* \in (0,\infty)$ with $F_\gamma(\mu_\gamma^*;\mu_{-\gamma}) = \alpha$; non-increasing monotonicity then gives $F_\gamma(\mu_\gamma;\mu_{-\gamma}) > \alpha$ for every $\mu_\gamma < \mu_\gamma^*$ and $F_\gamma(\mu_\gamma;\mu_{-\gamma}) \le \alpha$ for every $\mu_\gamma \ge \mu_\gamma^*$, so $\mu_\gamma^*$ is the leftmost point of the sublevel set $\{\mu_\gamma \ge 0 : F_\gamma \le \alpha\}$.
This is the population crossing point. The finite-sample bisection in SPOT (Algorithm~\ref{alg:general_k}) acts instead on the empirical target $\hat F_\gamma$, a non-increasing step function (Theorem~\ref{res:theorem2} applies pathwise), refining bracketing intervals $[\mu_{\rm lo}, \mu_{\rm hi}]$ with $\hat F_\gamma(\mu_{\rm lo}) > \alpha \ge \hat F_\gamma(\mu_{\rm hi})$ to the prescribed tolerance.
\end{proof}

\subsection{Convergence (\texorpdfstring{Theorem~\ref{thm:algo_convergence}}{Theorem 3}), local rate (\texorpdfstring{Proposition~\ref{prop:linear_rate}}{Proposition 1}), and the Monte Carlo analysis}
\label{sec:supp_theory}

This section proves Theorem~\ref{thm:algo_convergence} (global convergence) and Proposition~\ref{prop:linear_rate} (local linear rate)
of the main paper, and establishes the Monte Carlo consistency of SPOT. We first show that the dual
objective is continuously differentiable (Lemma~\ref{res:C1}, Corollary~\ref{res:C2}), which makes the
exact population recursion an exact block-coordinate minimisation of a smooth convex function; we then prove
a \emph{global} convergence theorem (Lemma~\ref{res:C3}--Theorem~\ref{res:C5}); finally we bound the
error incurred by the Monte Carlo implementation (Assumption~\ref{res:C9}--Corollary~\ref{res:C12}). Throughout,
$R_i(\mu,\vec u) = a_i(\vec u) - \sum_{\gamma=0}^{K-1}\mu_\gamma\,b_{\gamma,i}(\vec u)$,
$S_l(\mu,\vec u) = \sum_{i=1}^l R_i(\mu,\vec u)$ with $S_0\equiv 0$,
$l^*(\mu,\vec u) = \max\bigl\{\argmax_{0\le l\le K} S_l(\mu,\vec u)\bigr\}$ (ties broken by the largest maximiser),
$D_i^\mu(\vec u) = I\{i\le l^*(\mu,\vec u)\}$,
$F_\gamma(\mu) = \int_Q\sum_{i=1}^K b_{\gamma,i}(\vec u)\,D_i^\mu(\vec u)\,d\vec u$,
and $h(\mu) = \max_{\vec D} L(\vec D,\mu)$ is the convex dual function of the main paper.

\subsubsection*{Smoothness of the dual objective}

For fixed $\mu$, define the \emph{tie set}
\[
T_\mu = \bigl\{\vec u\in Q : \textstyle\max_{0\le l\le K} S_l(\mu,\vec u)\ \text{is attained at two or more distinct } l\bigr\}.
\]
Off $T_\mu$ the maximiser $l^*(\mu,\vec u)$, hence the policy $\vec D^\mu(\vec u)$, is uniquely
determined. We use two standard facts.

\medskip\noindent\textbf{Fact A.} \emph{The zero set of a polynomial $P:\mathbb R^K\to\mathbb R$
that is not identically zero has Lebesgue measure zero.} (Induction on $K$: for $K=1$ a nonzero
polynomial has finitely many roots; for the step, write $P(v)=\sum_j c_j(v_1,\dots,v_{K-1})v_K^{\,j}$,
choose $j$ with $c_j\not\equiv 0$, apply the inductive hypothesis to $\{c_j=0\}$ and Fubini.)

\medskip\noindent\textbf{Fact B.} \emph{A Lipschitz map $\Psi:A\subseteq\mathbb R^K\to\mathbb R^K$
carries Lebesgue-null subsets of $A$ to Lebesgue-null sets.} (A cube of side $h$ has image of
diameter $\le\mathrm{Lip}(\Psi)\sqrt K\,h$, hence outer measure $\le C h^K$; sum over a cover.)

\begin{lemmaS}{C.1}[Decision boundaries are negligible]
\label{res:C1}
Under Assumption~\ref{as:assumption3}, for every fixed $\mu\in[0,\infty)^K$ the tie set $T_\mu$ has Lebesgue measure
zero in $Q$.
\end{lemmaS}

\begin{proof}
Since $T_\mu\subseteq\bigcup_{0\le l'<l\le K} E_{l,l'}$ with
$E_{l,l'}=\{\vec u\in Q: S_l(\mu,\vec u)=S_{l'}(\mu,\vec u)\}$ and the union is finite, it suffices to
show each $E_{l,l'}$ is null. Fix $l'<l$.

\emph{Step 1: the difference is a nonzero polynomial in the likelihood ratios.}
Write $v_j=g(u_j)$ and $\vec v=(v_1,\dots,v_K)$. By Lemma~\ref{res:lemma1}, $a_i(\vec u)=(K{-}1)!\,v_1v_2\cdots v_K$
for every $i$, and
$b_{\gamma,i}(\vec u)=\gamma!\,(K{-}\gamma)!\,\bigl(\prod_{j=1}^{i-1}v_j\bigr)\,e_{\gamma-i+1}(v_{i+1},\dots,v_K)$
is a polynomial in $\vec v$ of total degree $(i{-}1)+(\gamma{-}i{+}1)=\gamma\le K-1$. Hence
\begin{multline*}
S_l(\mu,\vec u)-S_{l'}(\mu,\vec u)=\sum_{i=l'+1}^l R_i(\mu,\vec u)\\
=(l-l')\,(K{-}1)!\,v_1v_2\cdots v_K-\sum_{i=l'+1}^l\sum_{\gamma=0}^{K-1}\mu_\gamma\,b_{\gamma,i}(\vec u)
=:P_{l,l'}(\vec v),
\end{multline*}
a polynomial in $\vec v$ whose coefficients depend on $(\mu,l,l',K)$ but not on $\vec u$. The
subtracted double sum has degree at most $K-1$ and so contains no multiple of the squarefree
degree-$K$ monomial $v_1v_2\cdots v_K$; therefore the coefficient of $v_1v_2\cdots v_K$ in $P_{l,l'}$
is $(l-l')\,(K{-}1)!\neq 0$, and $P_{l,l'}\not\equiv 0$.

\emph{Step 2: its zero set is negligible in $\vec v$.}
By Fact~A, $Z=\{\vec v\in\mathbb R^K:P_{l,l'}(\vec v)=0\}$ has Lebesgue measure zero.

\emph{Step 3: pull back to $\vec u$ without inflating measure.}
Let $\Phi(\vec u)=(g(u_1),\dots,g(u_K))$, so $E_{l,l'}=\Phi^{-1}(Z)$. Assumption~\ref{as:assumption3} gives, for all
$u\neq u'$, $|g(u)-g(u')|\ge c_3|u-u'|>0$; hence $g$ is injective, and with $v=g(u),\,v'=g(u')$,
$|g^{-1}(v)-g^{-1}(v')|=|u-u'|\le c_3^{-1}|v-v'|$, so $g^{-1}$ is $c_3^{-1}$-Lipschitz on $g([0,1])$.
Thus $\Phi$ is injective with Lipschitz inverse $\Psi(\vec v)=(g^{-1}(v_1),\dots,g^{-1}(v_K))$ on
$\Phi(Q)$, and $E_{l,l'}=\Psi(Z\cap\Phi(Q))$ is the Lipschitz image of a subset of the null set $Z$;
by Fact~B it is null. The finite union over $0\le l'<l\le K$ gives $|T_\mu|=0$.
\end{proof}

\begin{corollaryS}{C.2}[Continuity of $F_\gamma$ and smoothness of $h$]
\label{res:C2}
Under Assumptions~\ref{as:assumption3} and~\ref{as:assumption5}, each $F_\gamma$ is continuous on $[0,\infty)^K$. Consequently
$h$ is convex and continuously differentiable, with $\partial h/\partial\mu_\gamma(\mu)=\alpha-F_\gamma(\mu)$.
\end{corollaryS}

\begin{proof}
Fix $\mu$ and let $\mu^{(n)}\to\mu$. For $\vec u\notin T_\mu$ the maximiser $l^*(\mu,\vec u)$ is the
unique index attaining $\max_l S_l(\mu,\vec u)$, so $S_{l^*}(\mu,\vec u)>S_l(\mu,\vec u)$ for all
$l\neq l^*$; since each $S_l(\cdot,\vec u)$ is affine, hence continuous, in $\mu$, the same strict
inequalities hold at $\mu^{(n)}$ for all large $n$, giving $l^*(\mu^{(n)},\vec u)=l^*(\mu,\vec u)$ and
$D_i^{\mu^{(n)}}(\vec u)\to D_i^\mu(\vec u)$ for every $i$. By Lemma~\ref{res:C1}, $|T_\mu|=0$, so this holds for
almost every $\vec u$. Under Assumption~\ref{as:assumption5}, $g\le c_5$, so every $b_{\gamma,i}$ is bounded on $Q$ by the finite constant
$B_1:=\max_{\gamma,i}\sup_{\vec u\in Q}b_{\gamma,i}(\vec u)<\infty$; hence $|\sum_i b_{\gamma,i}D_i^{\mu^{(n)}}|\le K B_1$, a constant integrable envelope on the
bounded set $Q$. Dominated convergence yields $F_\gamma(\mu^{(n)})\to F_\gamma(\mu)$. By Danskin's
theorem applied to the pointwise maximum $h$ (as in the proof of Theorem~\ref{res:theorem1}),
$\partial h/\partial\mu_\gamma=\alpha-F_\gamma$; since $h$ is convex and these partial derivatives exist everywhere,
$h$ is differentiable \citep[Thm.~25.2]{Rockafellar1970}, with gradient continuous wherever it exists
\citep[Thm.~25.5]{Rockafellar1970}, so $h\in C^1$.
These open-domain results apply because $h(\mu)=\max_{\vec D}L(\vec D,\mu)$ is finite and convex on all of $\mathbb R^K$, and Lemma~\ref{res:C1}'s polynomial argument is unaffected by the sign of $\mu$, so the partial derivatives exist on an open set containing $[0,\infty)^K$.
\end{proof}

\subsubsection*{Global convergence}

For fixed $\gamma$ and $\mu_{-\gamma}$, write $\varphi(t)=h(t;\mu_{-\gamma})$ for the dual objective
along coordinate $\gamma$; then $\varphi'(t)=\alpha-F_\gamma(t;\mu_{-\gamma})$ is non-decreasing
(Corollary~\ref{res:C2} and Theorem~\ref{res:theorem2}), so $\varphi$ is convex.

\begin{lemmaS}{C.3}[Coordinate minimiser]
\label{res:C3}
Under Assumptions~\ref{as:assumption3} and~\ref{as:assumption5}, with $\gamma,\mu_{-\gamma}$ fixed: \emph{(i)} if $F_\gamma(0;\mu_{-\gamma})>\alpha$, the minimisers of $\varphi$ on $[0,\infty)$
form the level set $\{\mu_\gamma\ge0:F_\gamma(\mu_\gamma;\mu_{-\gamma})=\alpha\}$, whose leftmost point is the crossing
point $\mu_\gamma^*(\mu_{-\gamma})$ of Corollary~\ref{res:corollary1} and which is a single point precisely when $F_\gamma$ crosses the level $\alpha$ strictly (the no-plateau proviso); \emph{(ii)} if $F_\gamma(0;\mu_{-\gamma})<\alpha$, the unique minimiser is
$\mu_\gamma=0$; \emph{(iii)} if $F_\gamma(0;\mu_{-\gamma})=\alpha$, then $\mu_\gamma=0$ is a minimiser and the minimiser set is again $\{\mu_\gamma\ge0:F_\gamma(\mu_\gamma;\mu_{-\gamma})=\alpha\}$, an interval containing $0$.
\end{lemmaS}

\begin{proof}
(i) $\varphi'(0)=\alpha-F_\gamma(0)<0$, so $\mu_\gamma=0$ is not a minimiser; since $\varphi$ is convex with
non-decreasing derivative $\varphi'=\alpha-F_\gamma$, its minimisers form the level set
$\{\mu_\gamma\ge0:F_\gamma(\mu_\gamma)=\alpha\}$, whose leftmost point is the crossing point $\mu_\gamma^*$ of
Corollary~\ref{res:corollary1}; by monotonicity this level set is an interval, a single point exactly when the crossing is strict. (ii) As $F_\gamma$ is
non-increasing, $F_\gamma(t)\le F_\gamma(0)<\alpha$ for all $t\ge 0$, so
$\varphi'(t)=\alpha-F_\gamma(t)\ge\alpha-F_\gamma(0)>0$; thus $\varphi$ is strictly increasing and is
uniquely minimised at $t=0$. (iii) Here $F_\gamma(t)\le F_\gamma(0)=\alpha$ for all $t$, so $\varphi'\ge0$ and $\varphi$ is
non-decreasing; its minimisers are the initial interval on which $\varphi'$ vanishes, i.e., by monotonicity of $F_\gamma$,
the set $\{t\ge0:F_\gamma(t)=\alpha\}$, which contains $t=0$.
\end{proof}

\noindent\emph{Remark (plateaus).} The coordinate minimiser is unique precisely when the level set
$\{\mu_\gamma\ge0:F_\gamma(\mu_\gamma;\mu_{-\gamma})=\alpha\}$ is a single point.
Non-uniqueness occurs whenever $F_\gamma(\cdot;\mu_{-\gamma})$ is flat at level $\alpha$ over a nondegenerate
interval (a \emph{plateau}). This can arise at the left boundary (case (iii) of Lemma~\ref{res:C3}, where
$F_\gamma(0;\mu_{-\gamma})=\alpha$ and $F_\gamma$ remains at $\alpha$ on an interval $[0,\bar t\,]$) or in the interior of
$[0,\infty)$, where $F_\gamma$ is locally constant at $\alpha$ away from $0$; boundary tangency is not the only mechanism. This does not affect Theorem~\ref{res:C5}, whose hypotheses require each coordinate minimum only to
be \emph{attained}. Where the classical statement of \citet[Prop.~2.7.1]{Bertsekas1999}
is preferred, the absence of plateaus at level $\alpha$ can be imposed as an additional generic condition.

\begin{lemmaS}{C.4}[Coercivity and compact iterate set]
\label{res:C4}
For every $\mu\in[0,\infty)^K$,
\[
h(\mu)\;\ge\;L(\vec 0,\mu)\;=\;\alpha\sum_{\gamma=0}^{K-1}\mu_\gamma\;=\;\alpha\,\lVert\mu\rVert_1 .
\]
Consequently $h$ is coercive on $[0,\infty)^K$; the sublevel set
$\mathcal S_0=\{\mu\ge 0:h(\mu)\le h(\vec\mu^{(0)})\}$ is compact; and every iterate of
the exact population recursion of Theorem~\ref{res:C5} satisfies $\vec\mu^{(t)}\in\mathcal S_0$ and
$\lVert\vec\mu^{(t)}\rVert_1\le h(\vec\mu^{(0)})/\alpha$.
\end{lemmaS}

\begin{proof}
Taking the feasible policy $\vec D\equiv\vec 0$ in $h(\mu)=\max_{\vec D}L(\vec D,\mu)$ and using the
Lagrangian representation $L(\vec D,\mu)=\alpha\sum_{\gamma}\mu_\gamma+\int_Q\sum_i D_i(\vec u)R_i(\mu,\vec u)\,d\vec u$
of the main paper gives $h(\mu)\ge L(\vec 0,\mu)=\alpha\sum_\gamma\mu_\gamma$, since
the integral vanishes at $\vec D\equiv\vec 0$; as $\mu\ge 0$ this equals $\alpha\lVert\mu\rVert_1$.
Hence $h(\mu)\to\infty$ as $\lVert\mu\rVert_1\to\infty$, i.e.\ $h$ is coercive. Also $h$ is finite
(at $\mu=\vec 0$, $R_i=a_i\ge 0$, so $h(\vec 0)=\int_Q\sum_i a_i\,d\vec u=K!\int_Q\prod_{j}g(u_j)\,d\vec u<\infty$
by Assumption~\ref{as:assumption5}). By Corollary~\ref{res:C2} $h$ is continuous, so $\mathcal S_0$ is closed; coercivity makes
it bounded, hence compact. Each exact population coordinate update minimises $h$ along one coordinate, so
$h(\vec\mu^{(t)})\le h(\vec\mu^{(t-1)})\le\cdots\le h(\vec\mu^{(0)})$; thus $\vec\mu^{(t)}\in\mathcal S_0$
and $\alpha\lVert\vec\mu^{(t)}\rVert_1\le h(\vec\mu^{(t)})\le h(\vec\mu^{(0)})$.
\end{proof}

\begin{theoremS}{C.5}[Global convergence]
\label{res:C5}
Suppose Assumptions~\ref{as:assumption3}--\ref{as:assumption5} hold, and consider the exact population version of SPOT (Algorithm~\ref{alg:general_k} with each Monte Carlo target $\hat F_\gamma$ replaced by the exact $F_\gamma$ and each coordinate minimised exactly). Then its objective values satisfy
$h(\vec\mu^{(t)})\downarrow\min_{\mu\ge 0}h(\mu)$, and every limit point of $\{\vec\mu^{(t)}\}$ is an optimal
dual vector $\mu^*$. No contraction, strong-convexity, Lipschitz-gradient, or active-Jacobian-nonsingularity hypothesis is used.
\end{theoremS}

\begin{proof}
By Corollary~\ref{res:C2}, $h$ is convex and continuously differentiable on $[0,\infty)^K$, with
$\partial h/\partial\mu_\gamma=\alpha-F_\gamma$. By Lemma~\ref{res:C3}, for every fixed $\mu_{-\gamma}$ the coordinate
subproblem $\min_{\mu_\gamma\ge0}h(\mu_\gamma;\mu_{-\gamma})$ \emph{attains} its minimum, at the crossing
point $\mu_\gamma^*(\mu_{-\gamma})$ of Corollary~\ref{res:corollary1} when $F_\gamma(0;\mu_{-\gamma})>\alpha$, and at
$\mu_\gamma=0$ otherwise. By Lemma~\ref{res:C4} the iterates lie in the compact sublevel set $\mathcal S_0$. Thus
the exact population recursion is exact cyclic block-coordinate minimisation of a convex, continuously differentiable function (with attained coordinate minima) over the separable box $[0,\infty)^K$ with compact sublevel sets. Being convex and differentiable, $h$ is pseudoconvex, so Proposition~6 of \citet{GrippoSciandrone2000} applies on the compact sublevel set $\mathcal S_0$ and yields that every limit point of the iterates is a global minimiser of $h$, \emph{without} a unique-minimiser requirement; Theorem~4.1 of \citet{Tseng2001} gives the same conclusion, its pairwise-pseudoconvexity and regularity hypotheses holding trivially for a differentiable convex objective. Concretely, each limit point is a coordinate-wise minimum, which for the smooth
convex $h$ is a stationary point on $[0,\infty)^K$ and hence, by convexity, a global minimiser. Since the
objective is monotone non-increasing and bounded below and the iterates remain in $\mathcal S_0$,
$h(\vec\mu^{(t)})\downarrow\min_{\mu\ge 0}h(\mu)$ and at least one limit point exists; every limit point
minimises the dual, i.e.\ is an optimal $\mu^*$.

When the coordinate minima are unique (the no-plateau proviso of the remark below), the same conclusion
follows more elementarily from Proposition~2.7.1 of \citet{Bertsekas1999}.
\end{proof}

\noindent\emph{Remark (no-plateau proviso).} The proviso enters only the elementary Proposition~2.7.1 route,
whose uniqueness requirement serves to exclude Powell-type cycling; it demands that no coordinate section of
$F_\gamma$ be flat at level $\alpha$ over a nondegenerate interval. When the Bertsekas route is used, this can be
imposed as an additional generic assumption; the route through \citet[Proposition~6]{GrippoSciandrone2000} and
\citet[Theorem~4.1]{Tseng2001} used above requires nothing of the sort.

\begin{theoremS}{C.6}[Local linear rate; contraction hypothesis removed]
\label{res:C6}
Under Assumptions~\ref{as:assumption3}--\ref{as:assumption5}, let $\mu^*$ be a dual optimum with active set $A=\{\gamma:\mu^*_\gamma>0\}$, and assume:
\emph{(a)} strict complementarity, $F_\gamma(\mu^*)<\alpha$ for $\gamma\notin A$; \emph{(b)} $F=(F_0,\dots,F_{K-1})$
is continuously differentiable in a neighbourhood of $\mu^*$; and \emph{(c)} Assumption~\ref{res:C9}, the active
Jacobian $J=[\partial F_\gamma/\partial\mu_\delta(\mu^*)]_{\gamma,\delta\in A}$ is nonsingular. Then $\mu^*$ is the
unique dual optimum, $h$ is $C^{2}$ and strongly convex in the active coordinates on a neighbourhood of
$\mu^*$, and the iterates of the exact population recursion of Theorem~\ref{res:C5} converge $R$-linearly (Theorem~\ref{res:C5} first bringing them into the local neighbourhood): there exist $\rho\in(0,1)$ and $C<\infty$
with $\lVert\vec\mu^{(t)}-\mu^*\rVert\le C\rho^{\,t}$ for all large $t$.
\end{theoremS}

\noindent\emph{Remark.} The successive-iterate ratios reported in \S\ref{sec:convergence} of the main paper, ranging from about $0.1$ to $0.5$ after the initial transient, are finite-run descriptive counterparts of $\rho$, not a verification of the rate.

\begin{proof}
On a neighbourhood of $\mu^*$, by (b) the map $\nabla h=(\alpha-F_\gamma)_\gamma$ is $C^{1}$, so $h$ is $C^{2}$ with
Hessian $\nabla^2 h(\mu)=-\,\partial F/\partial\mu$. Convexity of $h$ gives $\nabla^2 h\succeq 0$; restricted to the
active coordinates the Hessian is $-J$, which by (c) is nonsingular, hence $-J\succ 0$. Thus $h$ is strongly
convex on the active block near $\mu^*$; for directions that increase an inactive coordinate, $\partial h/\partial\mu_\gamma(\mu^*)=\alpha-F_\gamma(\mu^*)>0$ by (a), so $h$ increases strictly there as well, making $\mu^*$ an isolated (strict) local minimiser on the orthant and hence, $h$ being convex, the unique global minimiser; Theorem~\ref{res:C5} then gives $\vec\mu^{(t)}\to\mu^*$. By (a) and continuity, once
$\vec\mu^{(t)}$ is close to $\mu^*$ the inactive coordinates satisfy $\partial h/\partial\mu_\gamma=\alpha-F_\gamma>0$
and are pinned at $0$ by their coordinate update; the tail of the SPOT recursion is therefore exact cyclic
coordinate minimisation of the $C^{2}$, strongly convex function $h$ restricted to the active block, with
Lipschitz gradient there. The rate now follows from a classical Gauss--Seidel argument.
Let $S$ denote the map taking the active block through one full sweep of exact coordinate minimisation of $h(\cdot\,;\mu_{A^c}=\mathbf 0)$; each coordinate update solves $\partial_\gamma h = 0$ given the already-updated earlier coordinates, and since the diagonal entries of the Hessian $H := \nabla_A^2 h(\mu^*_A)$ are positive, the implicit function theorem makes each update, hence $S$, continuously differentiable near $\mu^*_A$, with $S(\mu^*_A) = \mu^*_A$.
Differentiating the stationarity relations at the fixed point and writing $H = D + L + L^{\top}$, with $D$ diagonal and $L$ strictly lower triangular, gives $(D + L)\,S'(\mu^*_A) = -L^{\top}$, so $S'(\mu^*_A) = -(D + L)^{-1}L^{\top}$ is the Gauss--Seidel iteration matrix of $H$; for symmetric positive definite $H$ its spectral radius is strictly less than one (the Ostrowski--Reich theorem; see, e.g., \citealp[\S\S 10.1 and 10.3]{OrtegaRheinboldt2000}, also for the local convergence of nonlinear Gauss--Seidel).
Hence $\lVert\vec\mu^{(t)}_A-\mu^*_A\rVert \le C\rho^{\,t}$ for any $\rho$ strictly between the spectral radius and one and all large $t$, which yields the stated bound; \citet{LuoTseng1992} and \citet{BeckTetruashvili2013} give related global statements.
\end{proof}

\noindent\emph{Remark (linear rate).} Theorem~\ref{res:C6} is Proposition~\ref{prop:linear_rate} of the main paper: it delivers a linear
convergence rate under local $C^1$ smoothness, strict complementarity, and a nonsingular active
Jacobian (Assumption~\ref{res:C9}), a local non-degeneracy condition used in place of a global contraction hypothesis on the
update map (which would essentially presuppose the convergence it is meant to establish). The global convergence (with no rate assumption)
is Theorem~\ref{thm:algo_convergence}, proved as Theorem~\ref{res:C5} above.

\subsubsection*{Lipschitz stability of a single threshold, and why global Lipschitz is not claimed}

\begin{lemmaS}{C.7}[Lipschitz stability of one rejection threshold]
\label{res:C7}
Under Assumptions~\ref{as:assumption3}--\ref{as:assumption5} there is a finite constant $C_0=C_0(K,c_3,c_4,c_5)$ such that, for every
$j\in\{1,\dots,K\}$ and all $\mu,\mu'\in[0,\infty)^K$,
\[
\lambda\bigl(\{\vec u\in Q:R_j(\mu,\vec u)>0\}\,\triangle\,\{\vec u\in Q:R_j(\mu',\vec u)>0\}\bigr)\;\le\;C_0\,\lVert\mu-\mu'\rVert_1 .
\]
\end{lemmaS}

\begin{proof}
Write $v_k=g(u_k)$. By Lemma~\ref{res:lemma1}, $a_j=(K{-}1)!\,v_1\cdots v_K$, and, crucially, the coefficient
\[
b_{\gamma,j}(\vec u)=\gamma!(K{-}\gamma)!\,\Bigl(\prod_{s<j}v_s\Bigr)\,e_{\gamma-j+1}(v_{j+1},\dots,v_K)
\]
\emph{does not depend on $v_j$} (the null at position $j$ contributes nothing; Remark~\ref{rem:structure}). Hence, as a function of $v_j$ with the other
coordinates fixed, $R_j(\mu,\vec u)=a_j-\sum_\gamma\mu_\gamma b_{\gamma,j}$ is affine with slope
\[
\frac{\partial R_j}{\partial v_j}=(K{-}1)!\prod_{s\neq j}v_s\;\ge\;(K{-}1)!\,c_4^{\,K-1}=:\kappa>0
\]
uniformly in $\mu$ and in $\vec v\in[c_4,c_5]^K$ (Assumption~\ref{as:assumption4}). Set
$\varepsilon:=\sup_{\vec u}\lvert R_j(\mu,\vec u)-R_j(\mu',\vec u)\rvert=\sup_{\vec u}\bigl\lvert\sum_\gamma(\mu_\gamma-\mu'_\gamma)b_{\gamma,j}(\vec u)\bigr\rvert
\le B_1\lVert\mu-\mu'\rVert_1$, where $B_1=\max_{\gamma,j}\sup_{\vec u}b_{\gamma,j}(\vec u)<\infty$ by Assumption~\ref{as:assumption5}. If
$R_j(\mu,\vec u)$ and $R_j(\mu',\vec u)$ have opposite signs then $\lvert R_j(\mu,\vec u)\rvert\le\varepsilon$, so the
symmetric difference is contained in $\{\vec u\in Q:\lvert R_j(\mu,\vec u)\rvert\le\varepsilon\}$. Passing to $\vec v$,
this is $\Psi\bigl(\{\vec v\in\Phi(Q):\lvert R_j\rvert\le\varepsilon\}\bigr)$ with $\Psi=(g^{-1},\dots,g^{-1})$
the $c_3^{-1}$-Lipschitz inverse of Lemma~\ref{res:C1}. For each fixed $\vec v_{-j}$ the affine-in-$v_j$ function has
$\{v_j:\lvert R_j\rvert\le\varepsilon\}$ an interval of length $\le 2\varepsilon/\kappa$; since $\Phi(Q)\subseteq[c_4,c_5]^K$,
Fubini gives $\lambda_{\vec v}\{\lvert R_j\rvert\le\varepsilon\}\le(c_5-c_4)^{K-1}\,2\varepsilon/\kappa$, and the Lipschitz
image bound (Fact~B) multiplies this by at most $c_3^{-K}$. Collecting constants,
\[
\lambda\bigl(\{\lvert R_j(\mu,\cdot)\rvert\le\varepsilon\}\bigr)\le
\frac{2\,c_3^{-K}(c_5-c_4)^{K-1}}{(K{-}1)!\,c_4^{\,K-1}}\,\varepsilon
\le C_0\lVert\mu-\mu'\rVert_1,\quad
C_0:=\frac{2\,c_3^{-K}(c_5-c_4)^{K-1}B_1}{(K{-}1)!\,c_4^{\,K-1}} .
\qedhere
\]
\end{proof}

\begin{corollaryS}{C.8}[Local Lipschitz continuity of $F_\gamma$]
\label{res:C8}
Suppose Assumptions~\ref{as:assumption3}--\ref{as:assumption5} hold, together with the following \emph{local single-crossing condition}: there is a neighbourhood $U$ of $\mu^*$ in $[0,\infty)^K$ such that for every $\mu\in U$ and almost every $\vec u\in Q$,
\[
R_k(\mu,\vec u)>0\ \text{ for all }\ k\le l^*(\mu,\vec u)
\qquad\text{and}\qquad
R_k(\mu,\vec u)\le 0\ \text{ for all }\ k>l^*(\mu,\vec u),
\]
i.e., the net-benefit sequence $k\mapsto R_k(\mu,\vec u)$ changes sign at most once, from positive to non-positive.
Then, for every $k$ and every $\mu\in U$, $\{\vec u: l^*(\mu,\vec u)\ge k\}$ coincides with $\{\vec u: R_k(\mu,\vec u)>0\}$ up to a null set, each $F_\gamma$ is Lipschitz on $U$, and, if in addition $F_\gamma$ is differentiable at $\mu^*$ (the
non-degeneracy assumed in Proposition~\ref{prop:linear_rate} and Assumption~\ref{res:C9}), the Jacobian $J$ is well defined.
\end{corollaryS}

\begin{proof}
Fix $\mu\in U$ and $\vec u$ off the stated null set, and let $k\in\{1,\dots,K\}$.
If $R_k(\mu,\vec u)>0$ then $k\le l^*(\mu,\vec u)$, since $R_j\le 0$ for every $j>l^*$; conversely, if $k\le l^*(\mu,\vec u)$ then $R_k(\mu,\vec u)>0$ by the condition.
Hence $\{l^*\ge k\}=\{R_k(\mu,\cdot)>0\}$ up to a null set, so
$F_\gamma(\mu)=\int_Q\sum_k b_{\gamma,k}\,I\{k\le l^*(\mu,\cdot)\}\,d\vec u$ changes, between $\mu,\mu'\in U$, only on
$\bigcup_k(\{R_k(\mu,\cdot)>0\}\triangle\{R_k(\mu',\cdot)>0\})$; bounding $b_{\gamma,k}\le B_1$ and applying
Lemma~\ref{res:C7} to each $k$ gives $\lvert F_\gamma(\mu)-F_\gamma(\mu')\rvert\le K B_1 C_0\lVert\mu-\mu'\rVert_1$, the local
Lipschitz bound. (A monotone, locally Lipschitz $F_\gamma$ is differentiable at Lebesgue-almost every point;
differentiability at the specific optimum $\mu^*$, hence existence of $J$, is the non-degeneracy hypothesis of
Proposition~\ref{prop:linear_rate}, not established here.)
\end{proof}

\noindent\emph{Remark (global Lipschitz continuity is not established; the $O(1/t)$ route is not used).}
Corollary~\ref{res:C8} is deliberately local. A \emph{global} Lipschitz bound on $F_\gamma$ via the same route would require the
non-consecutive partial-sum differences $S_p-S_q$ $(p>q)$ to be transverse, but these are not affine in any single
coordinate and their gradients can vanish on the zero set through $\mu$-dependent cancellation, leaving, heuristically,
only anti-concentration bounds of Carbery--Wright type \citep{CarberyWright2001} under the transformed measure, suggesting H\"older- rather than Lipschitz-type control. We record this as motivation only; what we actually rely on is the concrete failure of single crossing below, not this bound. The
single-crossing mechanism genuinely fails beyond the local regime: with the $K{=}3$ coefficients above and the distinct, decreasing likelihood-ratio values $\vec v=(3,\,1,\,\tfrac12)$ (compatible with the injectivity of $g$ under Assumption~\ref{as:assumption3}), $\mu=(0,\,\tfrac25,\,\tfrac25)$ one gets
$R=(1.4,\,-0.6,\,0.6)$ and $S=(0,\,1.4,\,0.8,\,1.4)$, whose maximum is attained at the non-consecutive indices $l=1$ and $l=3$; the net-benefit sequence is therefore not
unimodal and the boundary of $\{l^*\ge 2\}$ is not a single-increment level set.
We therefore neither establish nor assume global Lipschitz continuity of $F_\gamma$.
This is why convergence is routed through the coordinate-minimisation theory of \citet[Proposition~6]{GrippoSciandrone2000} and \citet[Prop.~2.7.1]{Bertsekas1999} (Theorem~\ref{res:C5}), which needs no
Lipschitz gradient, with the linear rate supplied locally by Theorem~\ref{res:C6}. A global Lipschitz gradient would be
required for the $O(1/t)$ rate of \citet{BeckTetruashvili2013}, which we therefore do not claim.

\medskip
\noindent The local rate (Theorem~\ref{res:C6}) and the Monte Carlo analysis invoke one non-degeneracy condition at the dual optimum~$\mu^*$, with active set $A=\{\gamma:\mu^*_\gamma>0\}$.

\begin{assumptionS}{C.9}[Local non-degeneracy]
\label{res:C9}
The Jacobian $J=[\partial F_\gamma/\partial\mu_\delta(\mu^*)]_{\gamma,\delta\in A}$ is nonsingular.
\end{assumptionS}

\noindent This is the standard non-degeneracy condition underlying the local linear rate (Proposition~\ref{prop:linear_rate} of the main
paper); only \emph{nonsingularity} of $J$ is needed here, a local condition used in place of a global contraction hypothesis.

\subsubsection*{Population and Monte Carlo versions of SPOT}

SPOT (Algorithm~\ref{alg:general_k}) replaces each population target $F_\gamma(\mu)={\rm FWER}_\gamma(\vec D^\mu)$ by a Monte Carlo
average. For each $\gamma$, draw $\vec X^{(1)},\dots,\vec X^{(N)}$ i.i.d.\ under configuration $\vec h_\gamma$
(the first $\gamma$ hypotheses alternatives, the remaining $K-\gamma$ true nulls) once, before the outer loop,
independently across configurations, and reuse that batch for every update of coordinate $\gamma$: all bisection
evaluations, across all sweeps (common random numbers). Each $\hat F_\gamma$ is therefore a single fixed empirical
function throughout the run, and the outer recursion is deterministic given the batches; this is exactly how the
accompanying implementation operates. Each $\vec X^{(n)}$ induces sorted $p$-values
$\vec u^{(n)}\in Q$ and a set $\mathcal N^{(n)}\subseteq\{1,\dots,K\}$ of sorted positions occupied by true
nulls, with $|\mathcal N^{(n)}|=K-\gamma$. The policy $\vec D^\mu$ rejects sorted ranks $1,\dots,l^*(\mu,\vec u^{(n)})$,
so its number of false rejections is $V^{(n)}(\vec D^\mu)=|\mathcal N^{(n)}\cap\{1,\dots,l^*(\mu,\vec u^{(n)})\}|$.
The estimator is the empirical family-wise error frequency
\begin{equation}
\label{eq:fhat_correct}
\hat F_\gamma(\mu)=\frac1N\sum_{n=1}^N\psi_\gamma(\mu,\vec X^{(n)}),\qquad
\psi_\gamma(\mu,\vec X)=I\{V(\vec D^\mu)>0\}=I\{\,l^*(\mu,\vec u)\ge\rho(\vec X)\,\},
\end{equation}
where $\rho(\vec X)=\min\mathcal N(\vec X)$ is the smallest sorted rank of a true null. Because only $\gamma$
hypotheses are alternatives, any $\gamma+1$ sorted positions include at least one null, so $\rho(\vec X)\in\{1,\dots,\gamma+1\}$.

\medskip\noindent\emph{Remark (consistency of the Monte Carlo target).}
By construction,
\[
\mathbb E\,\psi_\gamma(\mu,\vec X)=\mathbb P_{\vec h_\gamma}\{V(\vec D^\mu)>0\}={\rm FWER}_\gamma(\vec D^\mu)=F_\gamma(\mu),
\]
so $\hat F_\gamma$ is unbiased for the same $F_\gamma$ defined in the main paper by the linear form
$F_\gamma(\mu)=\int_Q\sum_k b_{\gamma,k}(\vec u)\,D_k^\mu(\vec u)\,d\vec u$; the two representations of
${\rm FWER}_\gamma$ agree \citep[eq.~(7)]{RHPA22}. The estimator is \emph{not} $N^{-1}\sum_n I\{l^*(\mu,\vec u^{(n)})\ge\gamma+1\}$:
that would equal $\mathbb P(l^*\ge\gamma+1)$, which omits the event that a true null is rejected while sitting
below an alternative, and understates $F_\gamma$ for $\gamma\ge 1$. For instance, with the $K{=}3$ coefficients
of the main paper under $\vec h_1$ (one alternative, two nulls, sorted density $2\{g(u_1)+g(u_2)+g(u_3)\}$),
$\int_Q\sum_k b_{1,k}D_k^\mu\,d\vec u
=\int_Q\bigl[\,2\{g(u_2)+g(u_3)\}\,I\{l^*{=}1\}+2\{g(u_1)+g(u_2)+g(u_3)\}\,I\{l^*{\ge}2\}\bigr]d\vec u$,
whereas $\mathbb P_{\vec h_1}(l^*\ge 2)=\int_Q 2\{g(u_1)+g(u_2)+g(u_3)\}\,I\{l^*{\ge}2\}\,d\vec u$; the two differ by
the nonnegative term $\int_Q 2\{g(u_2)+g(u_3)\}\,I\{l^*{=}1\}\,d\vec u$, exactly the probability that the rank-one
rejection is a null.

\begin{lemmaS}{C.10}[Uniform Monte Carlo error]
\label{res:C10}
Suppose Assumptions~\ref{as:assumption3}--\ref{as:assumption5} hold, with $K$ fixed as $N\to\infty$ and $0<\alpha<1$. Fix $\gamma$ and a bounded set $M\subseteq[0,\infty)^K$. Then $\hat F_\gamma$ in \eqref{eq:fhat_correct} satisfies
\[
\sup_{\mu\in M}\bigl|\hat F_\gamma(\mu)-F_\gamma(\mu)\bigr|=O_p(N^{-1/2}),\qquad
\sqrt N\,(\hat F_\gamma-F_\gamma)\ \rightsquigarrow\ \mathbb G_\gamma\ \text{ in }\ \ell^\infty(M),
\]
where $\mathbb G_\gamma$ is a tight mean-zero Gaussian process with covariance
$\mathrm{Cov}\{\psi_\gamma(\mu,\cdot),\psi_\gamma(\mu',\cdot)\}$; and at any $\mu$ with $F_\gamma(\mu)=\alpha$,
$\mathrm{var}\{\hat F_\gamma(\mu)\}=\alpha(1-\alpha)/N$.
\end{lemmaS}

\begin{proof}
The variance statement is immediate: $\psi_\gamma(\mu,\cdot)$ is a Bernoulli variable with mean $F_\gamma(\mu)$, so
$\mathrm{var}\{\hat F_\gamma(\mu)\}=F_\gamma(\mu)\{1-F_\gamma(\mu)\}/N$, which is $\alpha(1-\alpha)/N$ when $F_\gamma(\mu)=\alpha$.
For the uniform statements it suffices to show the class $\mathcal F_\gamma=\{\psi_\gamma(\mu,\cdot):\mu\in M\}$ is a
Donsker class of $\{0,1\}$-valued functions. Being Donsker means $\sqrt N(\hat F_\gamma-F_\gamma)\rightsquigarrow\mathbb G_\gamma$
in $\ell^\infty(M)$ with the stated covariance, and the asymptotic tightness of this process gives
$\sup_{\mu\in M}|\hat F_\gamma-F_\gamma|=O_p(N^{-1/2})$.

The members of $\mathcal F_\gamma$ are indicators of the sets $A_\mu=\{\vec X:l^*(\mu,\vec u)\ge\rho(\vec X)\}$. Because
$\rho$ does not depend on $\mu$ and takes values in $\{1,\dots,\gamma+1\}$,
\[
A_\mu=\bigcup_{m=1}^{\gamma+1}\Bigl(\{\rho=m\}\cap\{l^*(\mu,\vec u)\ge m\}\Bigr),
\]
a fixed finite union of intersections with the $\mu$-free sets $\{\rho=m\}$. Finite unions and intersections of
Vapnik--Chervonenkis (VC) classes with fixed sets are VC \citep[Lemma~2.6.17]{vanderVaartWellner1996}, so it suffices to show that for each $m$
the class $\mathcal C_m=\{\{l^*(\mu,\cdot)\ge m\}:\mu\in M\}$ is VC.

Fix $m\in\{1,\dots,\gamma+1\}$. Off the null tie set (Lemma~\ref{res:C1}, Lebesgue-null and hence $\mathbb P_{\vec h_\gamma}$-null,
since under Assumptions~\ref{as:assumption4}--\ref{as:assumption5} the law of $\vec u$ is absolutely continuous), $l^*(\mu,\vec u)$ is the unique maximiser
of $l\mapsto S_l(\mu,\vec u)$, so
\[
\{l^*(\mu,\vec u)\ge m\}=\bigcup_{l=m}^{K}\ \bigcap_{l'=0}^{m-1}\{\vec u:\Delta_{l,l'}(\mu,\vec u)\ge 0\},\qquad
\Delta_{l,l'}=S_l-S_{l'} .
\]
For $l>l'$, $\Delta_{l,l'}(\mu,\vec u)=c_{l,l'}(\vec u)-\sum_{\delta=0}^{K-1}\mu_{\delta}\,d_{l,l',\delta}(\vec u)$ with the fixed
coefficient functions $c_{l,l'}=\sum_{i=l'+1}^{l}a_i$ and $d_{l,l',\delta}=\sum_{i=l'+1}^{l}b_{\delta,i}$ (and $\Delta_{l,l'}=-\Delta_{l',l}$).
Thus, as $\mu$ ranges over $M$, the function $\vec u\mapsto\Delta_{l,l'}(\mu,\vec u)$ ranges within the fixed
$(K{+}1)$-dimensional vector space of functions spanned by $\{c_{l,l'},d_{l,l',0},\dots,d_{l,l',K-1}\}$. By the
sub-graph criterion for finite-dimensional function spaces \citep[Lemma~2.6.15]{vanderVaartWellner1996}, the class
$\{\{\Delta_{l,l'}(\mu,\cdot)\ge 0\}:\mu\in M\}$ is VC of index at most $K+3$. Since $\{l^*(\mu,\cdot)\ge m\}$ is a fixed
Boolean combination of at most $\binom{K+1}{2}$ such sets, $\mathcal C_m$ is VC \citep[Lemma~2.6.17]{vanderVaartWellner1996}. Hence
$\mathcal F_\gamma$ is a VC-subgraph class with envelope $1$; VC classes have polynomially bounded uniform covering
numbers \citep[Thm.~2.6.7]{vanderVaartWellner1996} and are therefore Donsker \citep[Thm.~2.5.2]{vanderVaartWellner1996}.
\end{proof}

\subsubsection*{Consistency and the $\sqrt N$-rate of the Monte Carlo solution}

\begin{theoremS}{C.11}[Root-$N$ error bounds and the centred limit]
\label{res:C11}
Suppose Assumptions~\ref{as:assumption3}--\ref{as:assumption5} and~\ref{res:C9} hold, together with strict complementarity at the dual optimum $\mu^*$
(so $F_\gamma(\mu^*)<\alpha$ for $\gamma\notin A$, where $A=\{\gamma:\mu^*_\gamma>0\}$) and, in addition to the local single-crossing condition of
Corollary~\ref{res:C8} (hence its local Lipschitz and symmetric-difference bounds), that $F=(F_0,\dots,F_{K-1})$ is $C^1$ near $\mu^*$. Let $\hat\mu$ be the output of SPOT (Algorithm~\ref{alg:general_k}) on the
Monte Carlo estimators $\hat F_\gamma$ of \eqref{eq:fhat_correct}; write $\hat\Phi_{N,\gamma}(\mu):=\min\{\mu_\gamma,\,\alpha-\hat F_\gamma(\mu)\}$
for the empirical complementarity residual and $r_N:=\max_\gamma\lvert\hat\Phi_{N,\gamma}(\hat\mu)\rvert$ for its achieved value. Let $\delta_N$ denote the bisection tolerance and suppose $\sup_N\delta_N\le\bar\delta<\infty$.

\emph{(a) Root-$N$ bounds under the implemented stopping rule.} If $r_N=O_p(N^{-1/2})$ (the regime of the implemented rule, whose
residual target is $\{\alpha(1-\alpha)/N\}^{1/2}$), then
$\lVert\hat\mu-\mu^*\rVert=O_p(N^{-1/2})$; in particular $\hat\mu\xrightarrow{p}\mu^*$,
$\hat\mu_\gamma=O_p(N^{-1/2})$ for every $\gamma\notin A$, and $F_\gamma(\hat\mu)=\alpha+O_p(N^{-1/2})$ for $\gamma\in A$ while
$F_\gamma(\hat\mu)<\alpha$ with probability tending to one for $\gamma\notin A$.

\emph{(b) Centred limit under an idealised stopping rule.} If moreover $r_N=o_p(N^{-1/2})$ (a stronger output condition than the implemented rule, imposed for the centred limit rather than proved to be attained; see the remark following the theorem), then $\hat\mu_\gamma=o_p(N^{-1/2})$ for every $\gamma\notin A$, the thresholded
set $\{\gamma:\hat\mu_\gamma>\tau_N\}$ equals $A$ with probability tending to one for any $\tau_N\to0$ with $\sqrt N\,\tau_N\to\infty$, and,
with $J=[\partial F_\gamma/\partial\mu_\delta(\mu^*)]_{\gamma,\delta\in A}$,
\[
\sqrt N(\hat\mu_A-\mu^*_A)=-\,J^{-1}\sqrt N\,(\hat F_A-F_A)(\mu^*)+o_p(1)\ \rightsquigarrow\ \mathcal N\!\bigl(0,\,J^{-1}\Sigma J^{-\top}\bigr),
\]
where $\Sigma=\lim_N N\,\mathrm{Cov}\{(\hat F_\gamma(\mu^*))_{\gamma\in A}\}$ has diagonal $\Sigma_{\gamma\gamma}=\alpha(1-\alpha)$.
\end{theoremS}

\noindent Part~(a), together with the uniform Monte Carlo bound of Lemma~\ref{res:C10}, is Theorem~\ref{thm:mc_accuracy} of the main paper; part~(b) supplies the centred Gaussian limit and support recovery mentioned there.

\noindent\emph{Remark (the idealised residual).} The condition $r_N=o_p(N^{-1/2})$ in part (b) is a stronger output condition than the implemented one-standard-error stopping rule; it is imposed for the centred limit, not proved to be attained by the recursion, and a finite-sample diagnostic cannot certify an asymptotic rate. It refers to the full-vector residual $r_N=\max_\gamma\lvert\hat\Phi_{N,\gamma}(\hat\mu)\rvert$ recomputed at the final iterate, which bounds all coordinates simultaneously (avoiding any per-coordinate bound that later sweeps could perturb).

\begin{proof}
Encode the Karush--Kuhn--Tucker conditions for $\min_{\mu\ge0}h(\mu)$ ($\mu_\gamma\ge0$, $\alpha-F_\gamma(\mu)\ge0$,
$\mu_\gamma\{\alpha-F_\gamma(\mu)\}=0$) by the complementarity map $\Phi_\gamma(\mu)=\min\{\mu_\gamma,\ \alpha-F_\gamma(\mu)\}$,
so that $\hat\Phi_{N}$ is its plug-in version. Since $\partial h/\partial\mu_\gamma=\alpha-F_\gamma$
(Corollary~\ref{res:C2}), $\Phi(\mu)=0$ iff $\mu$ is a KKT point iff, $h$ being convex, $\mu$ minimises the dual; by Theorem~\ref{res:C6} that
minimiser is unique, $\mu^*$.

\emph{Boundedness.} For each $\gamma$ choose a deterministic $U_\gamma$ with $F_\gamma(U_\gamma;\vec 0_{-\gamma})<\alpha/2$, possible
because $F_\gamma(\cdot\,;\vec 0_{-\gamma})\to0$ (proof of Corollary~\ref{res:corollary1}); let the bisection tolerance satisfy $\sup_N\delta_N\le\bar\delta<\infty$
and define the deterministic box $\mathcal B=\prod_\gamma[0,U_\gamma+\bar\delta]$. Then $\mu^*\in\mathcal B$: its inactive coordinates are zero, while an active $\mu^*_\gamma>U_\gamma$ would give, by monotonicity (Theorem~\ref{res:theorem2}), $F_\gamma(\mu^*)\le F_\gamma(U_\gamma;\vec 0_{-\gamma})<\alpha/2$, contradicting $F_\gamma(\mu^*)=\alpha$. By the pathwise monotonicity of Theorem~\ref{res:theorem2} in \emph{all}
coordinates, $\hat F_\gamma(U_\gamma;\mu_{-\gamma})\le\hat F_\gamma(U_\gamma;\vec 0_{-\gamma})\le F_\gamma(U_\gamma;\vec 0_{-\gamma})+
\sup_{\mathcal B}\lvert\hat F_\gamma-F_\gamma\rvert<\alpha$ with probability tending to one (Lemma~\ref{res:C10}), uniformly in $\mu_{-\gamma}\ge0$,
so every bisection crossing for coordinate $\gamma$ lies in $[0,U_\gamma]$; the returned midpoint can exceed the crossing by at most $\delta_N\le\bar\delta$, so all iterates, including $\hat\mu$, lie in $\mathcal B$ with probability tending to one.

\emph{(a).} As $b\mapsto\min\{a,b\}$ is $1$-Lipschitz,
$\sup_{\mathcal B}\lVert\hat\Phi_N-\Phi\rVert_\infty\le\max_\gamma\sup_{\mathcal B}\lvert\hat F_\gamma-F_\gamma\rvert=O_p(N^{-1/2})$
by Lemma~\ref{res:C10}, whence $\lVert\Phi(\hat\mu)\rVert_\infty\le r_N+O_p(N^{-1/2})=O_p(N^{-1/2})$.
On $\mathcal B$, $\Phi$ is continuous with unique zero $\mu^*$, so
$\inf\{\lVert\Phi(\mu)\rVert_\infty:\mu\in\mathcal B,\ \lVert\mu-\mu^*\rVert\ge\varepsilon\}>0$ for every $\varepsilon>0$ and
$\lVert\Phi(\hat\mu)\rVert_\infty\xrightarrow{p}0$ forces $\hat\mu\xrightarrow{p}\mu^*$.
Near $\mu^*$ the map $\Phi$ admits a local error bound: writing $x=\mu-\mu^*$, for $\gamma\notin A$ the slack
$\alpha-F_\gamma(\mu)$ is bounded away from zero, so $\Phi_\gamma(\mu)=\mu_\gamma=x_\gamma$, while for $\gamma\in A$,
$\mu_\gamma$ is bounded away from zero and the $C^1$ expansion gives
$\Phi_\gamma(\mu)=\alpha-F_\gamma(\mu)=-[Jx_A]_\gamma-[Bx_{A^c}]_\gamma+o(\lVert x\rVert)$, with $B$ the (bounded) matrix of inactive
partial derivatives. Writing $I=A^c$, the inactive block gives $x_I=\Phi_I(\mu)$ and the active block $\Phi_A(\mu)=-Jx_A-Bx_I+o(\lVert x\rVert)$; since $J$ is nonsingular (Assumption~\ref{res:C9}),
\[
\lVert x_A\rVert\le\lVert J^{-1}\rVert\bigl\{\lVert\Phi_A(\mu)\rVert+\lVert B\rVert\,\lVert\Phi_I(\mu)\rVert+o(\lVert x\rVert)\bigr\}.
\]
Combining with $\lVert x_I\rVert=\lVert\Phi_I(\mu)\rVert$ and shrinking the neighbourhood to absorb the $o(\lVert x\rVert)$ remainder yields constants $C,r>0$ such that $\lVert\mu-\mu^*\rVert\le C\lVert\Phi(\mu)\rVert$ whenever $\lVert\mu-\mu^*\rVert_\infty\le r$; equivalently $\lVert\Phi(\mu)\rVert_\infty\ge c\lVert\mu-\mu^*\rVert_\infty$ on the same neighbourhood, with $c=C^{-1}$ up to the fixed constants relating the $\ell_2$ and $\ell_\infty$ norms on $\mathbb R^K$.
Hence $\lVert\hat\mu-\mu^*\rVert=O_p(N^{-1/2})$.
For $\gamma\notin A$ the empirical slack $\alpha-\hat F_\gamma(\hat\mu)\xrightarrow{p}\alpha-F_\gamma(\mu^*)>0$, so with probability
tending to one $\hat\Phi_{N,\gamma}(\hat\mu)=\hat\mu_\gamma$ and therefore $\hat\mu_\gamma\le r_N$; this gives
$\hat\mu_\gamma=O_p(N^{-1/2})$ for every $\gamma\notin A$.
Finally, local Lipschitz continuity of $F_\gamma$ (Corollary~\ref{res:C8}) gives
$\lvert F_\gamma(\hat\mu)-\alpha\rvert=\lvert F_\gamma(\hat\mu)-F_\gamma(\mu^*)\rvert=O_p(N^{-1/2})$ for $\gamma\in A$, and
$F_\gamma(\hat\mu)\xrightarrow{p}F_\gamma(\mu^*)<\alpha$ for $\gamma\notin A$.

\emph{(b).} Under $r_N=o_p(N^{-1/2})$ the inactive bound of (a) sharpens to $\hat\mu_\gamma=o_p(N^{-1/2})$ for every $\gamma\notin A$;
since $\hat\mu_\gamma\xrightarrow{p}\mu^*_\gamma>0$ for $\gamma\in A$, any threshold $\tau_N\to0$ with $\sqrt N\tau_N\to\infty$
separates the two groups with probability tending to one, which is the stated support recovery.
For the expansion, note that with probability tending to one $\hat\mu_\gamma>r_N$ for $\gamma\in A$, so
$\hat\Phi_{N,\gamma}(\hat\mu)=\alpha-\hat F_\gamma(\hat\mu)$ and $\lvert\hat F_\gamma(\hat\mu)-\alpha\rvert\le r_N=o_p(N^{-1/2})$.
By Lemma~\ref{res:C10} the class $\{\psi_\gamma(\mu,\cdot):\gamma\in A,\ \mu\in\mathcal B\}$ is Donsker, so
$\mu\mapsto\sqrt N(\hat F_A-F_A)(\mu)$ is stochastically equicontinuous and, with $\hat\mu\xrightarrow{p}\mu^*$,
\[
\hat F_A(\hat\mu)-F_A(\hat\mu)=(\hat F_A-F_A)(\mu^*)+o_p(N^{-1/2}).
\]
Equicontinuity is in the $L_2(P)$ semimetric: for $\mu,\mu'$ near $\mu^*$,
$\mathbb E\{\psi_\gamma(\mu,\vec X)-\psi_\gamma(\mu',\vec X)\}^2=\mathbb P(A_\mu\,\triangle\,A_{\mu'})\le C\lVert\mu-\mu'\rVert_1\to0$,
because the law of $\vec u$ under $\vec h_\gamma$ has a bounded density on $Q$ (Assumptions~\ref{as:assumption4}--\ref{as:assumption5}) and the rejection
regions differ on a set of Lebesgue measure $O(\lVert\mu-\mu'\rVert_1)$ by the argument of Corollary~\ref{res:C8}; the classes are
pointwise measurable, being indicators of finite Boolean combinations of sets varying continuously with the finite-dimensional
parameter $\mu$, so no measurability difficulties arise.
By the $C^1$ assumption (with derivative existence guaranteed there, not by Corollary~\ref{res:C8}, which supplies only Lipschitz continuity),
\[
F_A(\hat\mu)-\alpha\mathbf 1=J(\hat\mu_A-\mu^*_A)+B\,\hat\mu_{A^c}+o_p(\lVert\hat\mu-\mu^*\rVert)
=J(\hat\mu_A-\mu^*_A)+o_p(N^{-1/2}),
\]
using $\lVert\hat\mu_{A^c}\rVert=o_p(N^{-1/2})$ and, from (a), $\lVert\hat\mu-\mu^*\rVert=O_p(N^{-1/2})$.
Combining the three displays with $\lvert\hat F_A(\hat\mu)-\alpha\mathbf 1\rvert\le r_N=o_p(N^{-1/2})$ gives
$J(\hat\mu_A-\mu^*_A)=-(\hat F_A-F_A)(\mu^*)+o_p(N^{-1/2})$, i.e., in the manner of the Z-estimator master theorem
\citep[Thm.~5.21]{vanderVaart1998},
\begin{align*}
\sqrt N(\hat\mu_A-\mu^*_A)&=-J^{-1}\sqrt N(\hat F_A-F_A)(\mu^*)+o_p(1)
\rightsquigarrow\ \mathcal N(0,J^{-1}\Sigma J^{-\top})
\end{align*}
by the multivariate central limit theorem. Finally $\psi_\gamma(\mu^*,\cdot)$ is
Bernoulli$(\alpha)$ for $\gamma\in A$, so $\Sigma_{\gamma\gamma}=\alpha(1-\alpha)$.
\end{proof}

\noindent\emph{Remark (exact zeros).} Although part (b) only bounds the inactive multipliers, the implemented update sets
$\hat\mu_\gamma=0$ exactly whenever $\hat F_\gamma(0;\cdot)\le\alpha$ at the time coordinate $\gamma$ is visited, so in practice the
inactive coordinates are typically returned as exact zeros; the theorem does not rely on this.

\noindent\emph{Note on the covariance $\Sigma$.} The diagonal entries $\Sigma_{\gamma\gamma}=\alpha(1-\alpha)$ are unambiguous. The
off-diagonals $\Sigma_{\gamma\gamma'}$ ($\gamma\ne\gamma'$, both active) depend on how the Monte Carlo draws are coupled across the
distinct configurations $\vec h_\gamma,\vec h_{\gamma'}$: independent batches make $\Sigma$ diagonal, while a shared uniform stream
(common random numbers across configurations) induces nonzero off-diagonals. Either way the sandwich $J^{-1}\Sigma J^{-\top}$ is
correct once the coupling (hence $\Sigma$) is specified. Our implementation draws the configuration batches independently,
so $\Sigma$ is diagonal with $\Sigma_{\gamma\gamma}=\alpha(1-\alpha)$. Corollary~\ref{res:C12} uses only the
$\gamma$-marginal rate and is unaffected by the coupling.

\begin{corollaryS}{C.12}[Realised family-wise error rate and conservative targeting]
\label{res:C12}
Under the conditions of Theorem~\ref{res:C11}(a): for each $\gamma\in A$, $F_\gamma(\hat\mu)=\alpha+O_p(N^{-1/2})$, and for $\gamma\notin A$,
$F_\gamma(\hat\mu)<\alpha$ with probability tending to one. Fix $0<\eta<1/2$ and let $z_\eta$ be defined by $\Pr(Z>z_\eta)=\eta$, $Z\sim N(0,1)$. If the bisection instead targets the reduced level
$\alpha_N=\alpha-z_\eta\{\alpha(1-\alpha)/N\}^{1/2}$, with the reduced-target complementarity residual
\[
\hat\Phi_{N,\gamma}^{(\alpha_N)}(\mu)=\min\{\mu_\gamma,\ \alpha_N-\hat F_\gamma(\mu)\},\qquad
r_N^{(\alpha_N)}=\lVert\hat\Phi_N^{(\alpha_N)}(\hat\mu)\rVert_\infty,
\]
satisfying $r_N^{(\alpha_N)}=o_p(N^{-1/2})$, then
$\mathrm{pr}\{F_\gamma(\hat\mu)\le\alpha\}\to1-\eta$ for each $\gamma\in A$; replacing $\eta$ by $\eta/K$ gives simultaneous control,
\[
\liminf_{N\to\infty}\mathrm{pr}\{F_\gamma(\hat\mu)\le\alpha\ \text{for all}\ \gamma\in A\}\ge1-\eta.
\]
Moreover, if $\tilde\mu$ is a nominal-target output satisfying $r_N=O_p(N^{-1/2})$ as in Theorem~\ref{res:C11}(a), the power cost of the reduced target is
\[
\lvert\Pi_K(D^{\hat\mu})-\Pi_K(D^{\tilde\mu})\rvert=O_p(N^{-1/2}).
\]
\end{corollaryS}

\noindent Corollary~\ref{res:C12} is Corollary~\ref{cor:conservative} of the main paper.

\begin{proof}
By Corollary~\ref{res:C8}, $F_\gamma$ is Lipschitz near $\mu^*$, so $F_\gamma(\hat\mu)-F_\gamma(\mu^*)=O_p(\lVert\hat\mu-\mu^*\rVert)=O_p(N^{-1/2})$
by Theorem~\ref{res:C11}(a); as $F_\gamma(\mu^*)=\alpha$ ($\gamma\in A$) this is the first claim, and $F_\gamma(\mu^*)<\alpha$ gives the second.
For the targeted level, let $\hat\mu$ now denote the reduced-target output, that is SPOT (Algorithm~\ref{alg:general_k}) run with level $\alpha_N$, whose achieved residual is $r_N^{(\alpha_N)}=\lVert\hat\Phi_N^{(\alpha_N)}(\hat\mu)\rVert_\infty=o_p(N^{-1/2})$. Since $\alpha_N\to\alpha$, we have $\alpha_N>\alpha/2$ for all sufficiently large $N$, so repeating the boundedness argument of Theorem~\ref{res:C11} with target $\alpha_N$ gives $\hat\mu\in\mathcal B$ with probability tending to one. Evaluating the population complementarity map $\Phi$ at the nominal level $\alpha$ at this $\hat\mu$, the $1$-Lipschitz dependence of $b\mapsto\min\{a,b\}$ on each argument gives
\[
\lVert\Phi(\hat\mu)\rVert_\infty\le r_N^{(\alpha_N)}+\lvert\alpha_N-\alpha\rvert+\max_\gamma\sup_{\mathcal B}\lvert\hat F_\gamma-F_\gamma\rvert=O_p(N^{-1/2}),
\]
the three terms being $o_p(N^{-1/2})$, $O(N^{-1/2})$ (as $\lvert\alpha_N-\alpha\rvert=z_\eta\{\alpha(1-\alpha)/N\}^{1/2}$), and $O_p(N^{-1/2})$ (Lemma~\ref{res:C10}) respectively. Since $\Phi$ is continuous on the fixed box $\mathcal B$ with unique zero $\mu^*$, the compact-separation argument of Theorem~\ref{res:C11}(a) forces $\hat\mu\xrightarrow{p}\mu^*$, and its local error bound then upgrades this to $\lVert\hat\mu-\mu^*\rVert=O_p(N^{-1/2})$, keeping it in the regime where the expansions below apply.
Then, for active $\gamma$, $\hat\mu_\gamma\xrightarrow{p}\mu^*_\gamma>0$, so with probability tending to one the reduced-target residual reduces to $\hat\Phi_{N,\gamma}^{(\alpha_N)}(\hat\mu)=\alpha_N-\hat F_\gamma(\hat\mu)$, whence $\hat F_\gamma(\hat\mu)=\alpha_N+o_p(N^{-1/2})$ (from $r_N^{(\alpha_N)}=o_p(N^{-1/2})$); so by stochastic equicontinuity (Lemma~\ref{res:C10}),
$F_\gamma(\hat\mu)=\alpha_N+\{F_\gamma-\hat F_\gamma\}(\hat\mu)+o_p(N^{-1/2})=\alpha_N-\{\hat F_\gamma(\mu^*)-\alpha\}+o_p(N^{-1/2})$. Writing
$Z_N=\hat F_\gamma(\mu^*)-\alpha$, so that $\sqrt N\,Z_N/\{\alpha(1-\alpha)\}^{1/2}\rightsquigarrow\mathcal N(0,1)$,
\[
\mathrm{pr}\{F_\gamma(\hat\mu)\le\alpha\}=\mathrm{pr}\{Z_N\ge\alpha_N-\alpha+o_p(N^{-1/2})\}
=\mathrm{pr}\Bigl\{\tfrac{\sqrt N\,Z_N}{\{\alpha(1-\alpha)\}^{1/2}}\ge-z_\eta+o_p(1)\Bigr\}\to1-\eta,
\]
since $\alpha_N-\alpha=-z_\eta\{\alpha(1-\alpha)/N\}^{1/2}$. This makes precise the conservative reduced-level targeting of \S\ref{sec:mc_error} of the main paper.
For its power cost, let $\tilde\mu$ denote the nominal-target empirical output of SPOT (Algorithm~\ref{alg:general_k} at level $\alpha$) and $\hat\mu$ the reduced-target output; both lie within $O_p(N^{-1/2})$ of $\mu^*$ by Theorem~\ref{res:C11}(a) and the bound above, so
\[
\lVert\tilde\mu-\hat\mu\rVert\le\lVert\tilde\mu-\mu^*\rVert+\lVert\hat\mu-\mu^*\rVert=O_p(N^{-1/2}).
\]
The two induced policies therefore differ on a set of Lebesgue measure $O_p(\lVert\tilde\mu-\hat\mu\rVert)=O_p(N^{-1/2})$ (Lemma~\ref{res:C7} and the single-crossing identity of Corollary~\ref{res:C8}), and with the power weights $a_k$ bounded under Assumption~\ref{as:assumption5} the average-power difference is $O_p(N^{-1/2})$.
\end{proof}

\noindent\emph{Remark (simultaneity).} Corollary~\ref{res:C12} controls each active configuration marginally; simultaneous
$1-\eta$ control over all $\gamma\in A$ follows by replacing $\eta$ with $\eta/K$ (a union bound over the at most $K$
active constraints, which gives simultaneous control at level at least $1-\eta$), as noted in the
main paper, or from the joint Gaussian limit of Theorem~\ref{res:C11} on calibrating a joint (max-)Gaussian critical value for the active coordinates.

\section{SPOT implementation, computational complexity, and convergence diagnostics}
\label{sec:supp_impl}

\subsection{Bisection subroutine}
The bisection subroutine for updating coordinate $\gamma$ in SPOT (Algorithm~\ref{alg:general_k}) proceeds as follows.
Given the current vector $\mu$ with all coordinates except $\gamma$ fixed, and the configuration-$\gamma$ batch of $N$ Monte Carlo samples $\{\vec{u}^{(n)}\}_{n=1}^N$, drawn once before the outer loop and reused for every update of coordinate $\gamma$ across all sweeps (common random numbers):
\begin{enumerate}
\item[(a)] Evaluate $\hat F_\gamma(0)$. If $\hat F_\gamma(0)\le\alpha$, return $0$.
\item[(b)] Otherwise set $\mathit{lo} = 0$, $\mathit{hi} = U_{\max}$; while $\hat{F}_\gamma(\mathit{hi}) > \alpha$, set $\mathit{lo} = \mathit{hi}$ and $\mathit{hi} = 2\mathit{hi}$. This terminates because $\hat{F}_\gamma(\mu_\gamma) \to 0$ as $\mu_\gamma \to \infty$, and leaves $\hat F_\gamma(\mathit{lo})>\alpha\ge\hat F_\gamma(\mathit{hi})$.
\item[(c)] \textbf{While} $\mathit{hi} - \mathit{lo} > \delta$: set $\mathit{mid} = (\mathit{lo} + \mathit{hi})/2$ and compute $\hat{F}_\gamma(\mathit{mid})$ from the $N$ samples; if $\hat{F}_\gamma(\mathit{mid}) > \alpha$ set $\mathit{lo} = \mathit{mid}$, else set $\mathit{hi} = \mathit{mid}$.
\item[(d)] Return $(\mathit{lo} + \mathit{hi})/2$.
\end{enumerate}
Because Theorem~\ref{thm:general_k_monotonicity} of the main paper applies pathwise to every $\vec{u}^{(n)}$, the use of common random numbers ensures that $\hat{F}_\gamma(\mu_\gamma)$ is non-increasing in $\mu_\gamma$ for any fixed batch, so the bisection is well-defined even with finite Monte Carlo samples.

\medskip\noindent\textit{Per-iteration cost.}
Each outer iteration performs $K$ bisection updates, each of $O(\log(U_{\rm final}/\delta))$ function evaluations, where $U_{\rm final}$ is the final upper bracket after any doubling.
The coefficient tensors $a_k$ and $b_{\gamma,k}$ are precomputed from the fixed batch at cost $O(K^3 N)$ per configuration, $O(K^4 N)$ in total; when memory permits they are cached across all sweeps at $O(K^3 N)$ storage (the $K$ configuration tensors $b_{\gamma,k}$, each of size $K^2 N$), after which each bisection evaluation costs only $O(K^2 N)$ via the linear combination $R_i = a_i - \sum_\gamma \mu_\gamma b_{\gamma,i}$, so each outer iteration costs $O(K^3 N \log(U_{\rm final}/\delta))$.
When they cannot be cached and are recomputed each sweep, the per-iteration cost rises to $O(K^4 N + K^3 N \log(U_{\rm final}/\delta))$.

\medskip\noindent\textit{Implementation settings.}
All runs use the upper bracket $U_{\max} = 50$ and the sweep cap $T_{\max}$ set to $100$ for the scaling and sensitivity runs, $400$ for the $112$-setting robustness grid, and $200$ for the applications; only one out-of-model Student-$t$ run in the grid reached its cap.

\subsection{Convergence diagnostics}
\label{sec:supp_convergence}
Figure~\ref{fig:supp_convergence} plots the convergence trajectory $\|\mu^{(t)} - \mu^{(t-1)}\|_2$ on a logarithmic scale for several $K$ values under the truncated normal model at $\theta = -2.0$.
With each configuration's Monte Carlo batch fixed, the recursion is deterministic conditional on the batches. In the plotted runs, the successive-sweep changes decrease approximately geometrically until the stopping tolerance is met, a descriptive counterpart of the population linear rate in Proposition~\ref{prop:linear_rate} rather than a consequence of fixing the batches.

\begin{figure}[htbp]
\centering
\includegraphics[width=25pc]{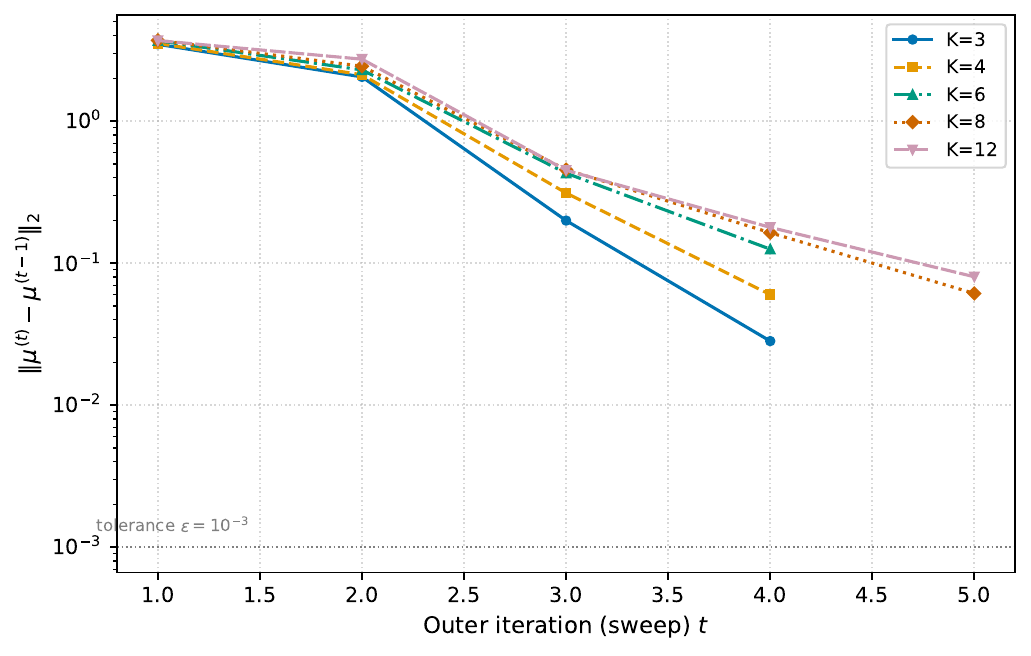}
\caption{Convergence of SPOT (Algorithm~\ref{alg:general_k}): $\|\mu^{(t)} - \mu^{(t-1)}\|_2$ versus outer iteration $t$ on a logarithmic scale, for several $K$ values under the truncated normal model at $\theta = -2.0$.
The approximately linear decrease is consistent with geometric convergence, with successive-iterate ratios of about $0.1$ to $0.5$ after the initial transient.}
\label{fig:supp_convergence}
\alttext{Successive multiplier changes decrease approximately geometrically for all five hypothesis counts.}
\end{figure}
\FloatBarrier

Table~\ref{tab:supp_timing} reports the wall-clock time and outer-iteration count for each $K$.
With bisection tolerance $\delta = 10^{-4}$ and outer tolerance $\varepsilon = 10^{-3}$, the residual criterion is met within $4$ to $5$ sweeps for every $K \le 12$, and the successive-iterate ratios stay bounded well below one (Figure~\ref{fig:supp_convergence}), a descriptive counterpart of the linear rate of Proposition~\ref{prop:linear_rate} rather than a verification of it.
Across the full $112$-setting grid the criterion was met in $92$ settings (median $5$ sweeps among these $92$, and $7$ sweeps across all $112$).
The remaining $20$, all in the out-of-model Student-$t$ family, terminated at residuals between one and three Monte Carlo standard errors ($7 \times 10^{-4}$ to $2 \times 10^{-3}$): $19$ by the iterate-change safeguard and one at the sweep cap.
Of these, the $92$ meeting the implemented one-standard-error stopping threshold satisfy that finite-sample diagnostic; the $20$ safeguard exits did not, though their residuals of one to three standard errors are themselves of $N^{-1/2}$ scale, and a single finite-$N$ run neither establishes nor refutes the asymptotic condition $r_N = O_p(N^{-1/2})$. Formal application of Theorem~\ref{res:C11} additionally requires an asymptotic sequence together with all its model and local regularity conditions (strict complementarity, single crossing, smoothness, and a nonsingular $J$), which we do not verify here, most of these settings being out-of-model.
For near-null alternatives the empirical targets can jump across $\alpha$ between iterates, so the residual need not fall below the Monte Carlo floor even as the sweeps stabilise; we therefore classify these as slow or coupled empirical terminations rather than as evidence of a particular limit. Per-setting diagnostics are stored with the results.

\section{Simulation models and scope of assumptions}
\label{sec:supp_models}

The simulation study of \S\ref{sec:simulations} of the main paper uses four models spanning different tail behaviours and violations of the boundedness assumption.
\begin{enumerate}
\item[(a)] \emph{Truncated normal.}  $H_{0k}: X_k \sim N(0,1)$ versus $H_{Ak}: X_k \sim N(\theta,1)$, both truncated to $[-6, 6]$, with $\theta \in \{-1, -1.5, -2, -2.5, -3, -3.5, -4\}$.
The truncation makes $g$ bounded and bounded away from zero on $[0,1]$, giving Assumptions~\ref{as:assumption5} and~\ref{as:assumption4}. Assumption~\ref{as:assumption3} holds as well: $g(u) = C\exp\{\theta\,F_0^{-1}(u) - \theta^2/2\}$ for a constant $C > 0$, so $g'(u) = \theta\,g(u)/f_0\{F_0^{-1}(u)\}$, and on the compact truncated domain $g$ is bounded below and $f_0$ bounded above, whence $|g'(u)| \ge |\theta|\,(\inf g)/(\sup f_0) > 0$; since $g$ is monotone this gives the lower-Lipschitz bound $|g(u) - g(u')| \ge c_3|u - u'|$.
\item[(b)] \emph{Mixture normal.}  $H_{Ak}: X_k \sim 0.5\,N(\theta,1) + 0.5\,N(-\theta,1)$, with $p$-values $u_k = 2\Phi(-|X_k|)$ and $\theta$ on the same grid $\{-1, -1.5, \ldots, -4\}$ as in~(a).
This is a two-sided, multimodal alternative.
The density $g(u) = \exp(-\theta^2/2)\cosh\{\theta\,\Phi^{-1}(1 - u/2)\}$ is unbounded as $u \to 0$, violating Assumption~\ref{as:assumption5}; it also violates Assumption~\ref{as:assumption3}, having vanishing slope $g'(1) = 0$ at $u = 1$, so no global lower-Lipschitz constant exists.
\item[(c)] \emph{Student-$t$ (out-of-model).}  $H_{Ak}: X_k \sim t_{\rm df}$, with ${\rm df} \in \{2, 4, 6, \ldots, 20\}$, representing heavy-tailed alternatives.
The uncapped likelihood ratio $g(u) = f_t\{\Phi^{-1}(1-u/2); {\rm df}\}/\phi\{\Phi^{-1}(1-u/2)\}$ is unbounded near $u = 0$, violating Assumption~\ref{as:assumption5}; it is also U-shaped rather than non-increasing, decreasing only for $u < 2\{1 - \Phi(1)\} \approx 0.32$, so the reduction to $p$-value-ordered policies does not apply and the computed policy is not guaranteed optimal for this model; we retain it purely as an out-of-model stress test.
For numerical stability, SPOT uses the working ratio $\min\{g(u),\exp(45)\}$ while evaluation samples remain drawn from the uncapped $t_{\rm df}$ alternative. Thus the capped function is neither the exact likelihood ratio nor a normalised alternative density; the cap affects about $1\%$ of alternative samples at ${\rm df} = 2$, $0.05\%$ at ${\rm df} = 4$, and a negligible fraction beyond.
\item[(d)] \emph{Beta.}  Under $H_{Ak}$, $u_k \sim {\rm Beta}(\theta, 1)$ with $\theta \in \{0.8, 0.6, 0.4, 0.2\}$.
The density $g(u) = \theta u^{\theta - 1}$ is decreasing and unbounded as $u \to 0$ for these $\theta < 1$, so the model lies outside Assumption~\ref{as:assumption5}.
\end{enumerate}

These $28$ model--parameter combinations, evaluated at each of $K \in \{3, 4, 5, 6\}$, form the $112$-setting robustness grid of \S\ref{sec:sim_robustness} of the main paper.

Only the truncated normal satisfies all of Assumptions~\ref{as:assumption3}--\ref{as:assumption5}, so the global (population) convergence guarantee of Theorem~\ref{thm:algo_convergence} applies to it directly; the local linear rate requires the further conditions of Theorem~\ref{res:C6} (strict complementarity, local $C^1$ smoothness, and a nonsingular active Jacobian), and the Monte Carlo consistency statements those of Theorem~\ref{res:C11}, none of which we verify a priori for any specific $g$. The mixture, Student-$t$, and beta models, and the applications of \S\ref{sec:realdata} (which use the mixture model), are therefore numerical studies outside even the global assumptions. That SPOT nonetheless terminates at small residuals for all these models (\S\ref{sec:supp_convergence}) suggests the regularity assumptions are sufficient but not necessary for it to work.

\section{Sensitivity analyses}
\label{sec:supp_sensitivity}

\emph{Sensitivity to the significance level $\alpha$.}
Repeating the truncated normal experiment at $\theta = -2.0$, $K = 6$ with $\alpha \in \{0.01, 0.05, 0.10\}$ shows the power advantage over Hommel is stable across levels and largest at the most stringent one, where the tighter constraints make the joint allocation of the error budget most consequential; the full comparison, including the realised family-wise error rates, is reported in Table~\ref{tab:sensitivity_alpha} of the main paper (\S\ref{sec:sensitivity_N}).

\emph{Sensitivity to the Monte Carlo sample size.}
For $K = 6$ under the truncated normal at $\theta = -2.0$, we run SPOT with $N_{\rm opt} \in \{50{,}000,\; 100{,}000,\; 200{,}000\}$ and evaluate each resulting policy with a common evaluation sample of $N_{\rm eval} = 50{,}000$.
The resulting $\Pi_6$ values are $0.625$, $0.623$, and $0.625$ respectively, with maximum empirical FWER of $0.053$, $0.052$, and $0.052$.
The power variation is within Monte Carlo uncertainty, and the FWER variation is small and consistent with Monte Carlo plus optimisation error; this supports adequacy of $N_{\rm opt} = 100{,}000$ for this $K = 6$ setting.

\section{Additional power curves}
\label{sec:supp_power}

Figures~\ref{fig:supp_K3}--\ref{fig:supp_K6} show power curves ($\Pi_K$ and $\Pi_{\rm any}$) for $K \in \{3, 4, 5, 6\}$ across all four distributional models: truncated normal, mixture normal, Student-$t$, and beta.
The $K = 3$ and $K = 6$ cases also appear as Figures~\ref{fig:power_curves_K3} and~\ref{fig:power_curves_K6} of the main paper; they are repeated here so that the full sequence $K \in \{3, 4, 5, 6\}$ can be read together.
In every average-power panel (left column), the policy computed by SPOT (solid black line) attains the highest average power $\Pi_K$ among all procedures (Bonferroni, Holm, Hochberg, Hommel) across the entire signal range; on the secondary any-discovery metric $\Pi_{\rm any}$ (right column), which is not the optimised objective, the standard procedures can slightly exceed SPOT for the out-of-model Student-$t$ alternative at larger degrees of freedom (see \S\ref{sec:sim_robustness} of the main paper).

\begin{figure}[htbp]
\centering
\includegraphics[width=0.931\linewidth]{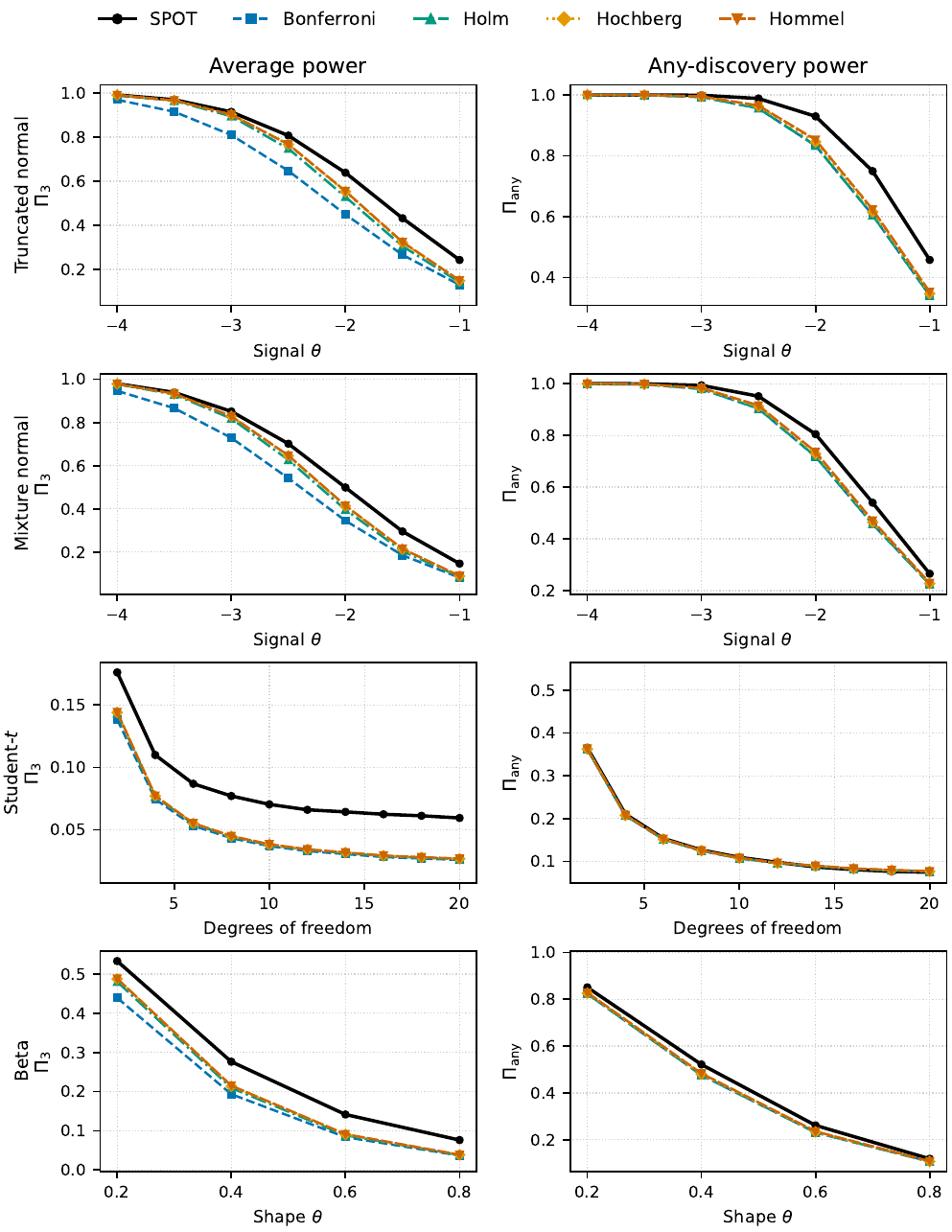}
\caption{Power curves at $K = 3$: average power $\Pi_3$ (left column) and any-discovery power $\Pi_{\rm any}$ (right column) for the truncated-normal, mixture-normal, Student-$t$, and beta alternatives (rows). This is Figure~\ref{fig:power_curves_K3} of the main paper. SPOT (solid black) attains the highest average power at every signal strength in all four models.}
\label{fig:supp_K3}
\alttext{Left panels: SPOT leads average power throughout. Right panels: SPOT leads except for near ties under Student t.}
\end{figure}

\begin{figure}[htbp]
\centering
\includegraphics[width=0.931\linewidth]{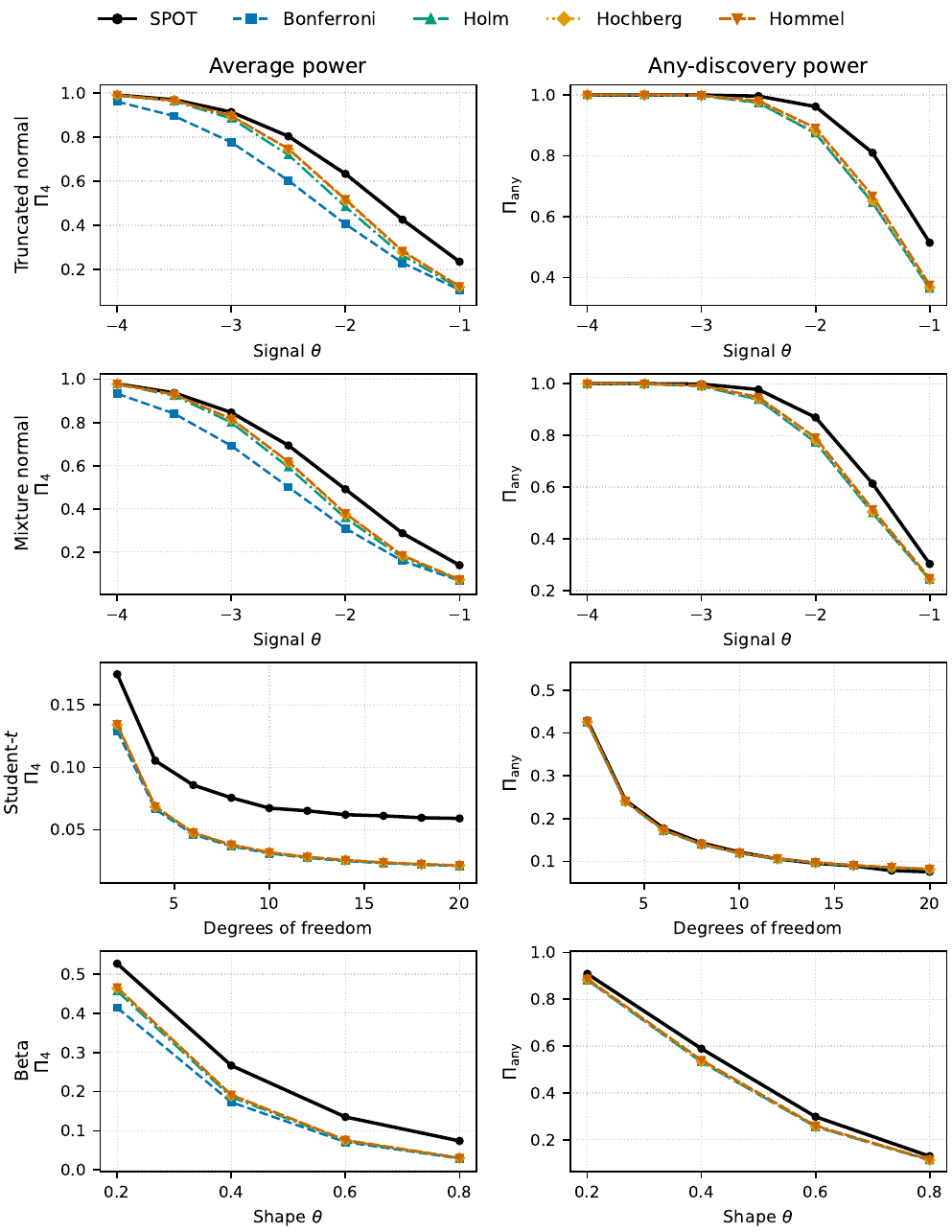}
\caption{Power curves at $K = 4$: average power $\Pi_4$ (left column) and any-discovery power $\Pi_{\rm any}$ (right column) for the truncated-normal, mixture-normal, Student-$t$, and beta alternatives (rows). SPOT (solid black) attains the highest average power at every signal strength in all four models.}
\label{fig:supp_K4}
\alttext{Left panels: SPOT leads average power throughout. Right panels: SPOT leads except for near ties under Student t.}
\end{figure}

\begin{figure}[htbp]
\centering
\includegraphics[width=0.931\linewidth]{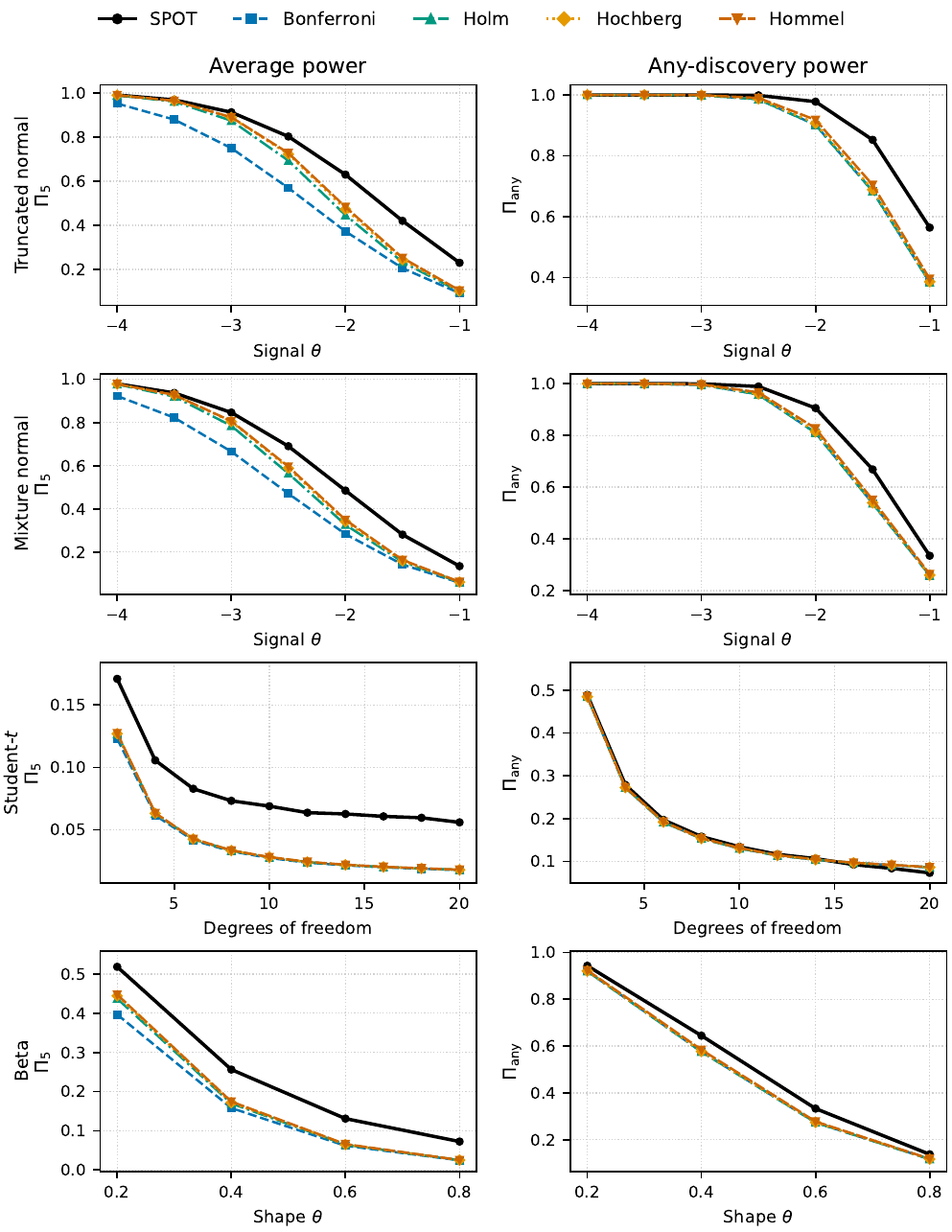}
\caption{Power curves at $K = 5$: average power $\Pi_5$ (left column) and any-discovery power $\Pi_{\rm any}$ (right column) for the truncated-normal, mixture-normal, Student-$t$, and beta alternatives (rows). SPOT (solid black) attains the highest average power at every signal strength in all four models.}
\label{fig:supp_K5}
\alttext{Left panels: SPOT leads average power throughout. Right panels: baselines approach or slightly surpass SPOT under Student t.}
\end{figure}

\begin{figure}[htbp]
\centering
\includegraphics[width=0.931\linewidth]{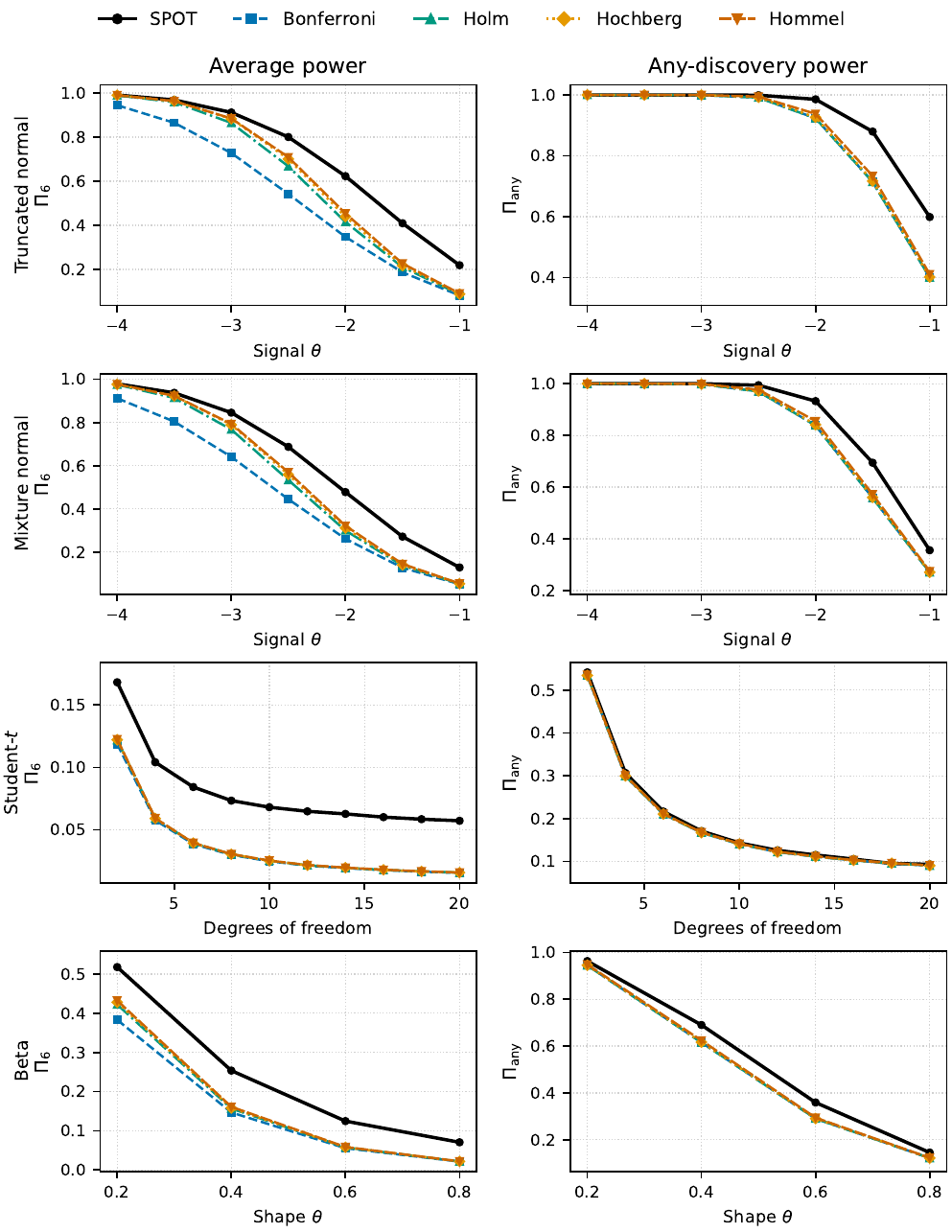}
\caption{Power curves at $K = 6$: average power $\Pi_6$ (left column) and any-discovery power $\Pi_{\rm any}$ (right column) for the truncated-normal, mixture-normal, Student-$t$, and beta alternatives (rows). This is Figure~\ref{fig:power_curves_K6} of the main paper. SPOT (solid black) attains the highest average power at every signal strength in all four models.}
\label{fig:supp_K6}
\alttext{Left panels: SPOT leads average power throughout. Right panels: baselines sometimes lead under Student t.}
\end{figure}

\clearpage
\section{Application details}
\label{sec:supp_apps}

This section collects the study-level tables and the secondary-application details for \S\ref{sec:realdata} of the main paper.

\subsection{Replicability of social science experiments}
Table~\ref{tab:supp_camerer} gives the full study-level rejection decisions for the $21$ Stage-1 replications of \citet{Camerer18}, and Table~\ref{tab:supp_camerer_sensitivity} the sensitivity of the rejection count to the assumed common effect $\theta$.
The policy returned by SPOT rejects nine studies: the seven rejected by Holm/Hochberg/Hommel (the six with $p = 0.001$ and Derex et al.\ at $p = 0.003$) plus the Karpicke--Blunt ($p = 0.006$) and Nishi et al.\ ($p = 0.011$) replications, the two additional rejections reflecting the strength-borrowing effect of \S\ref{sec:why_gap_grows} of the main paper.

At $\theta = 1.5$ the twelve rejected hypotheses coincide exactly with the twelve studies that replicated individually in Stage~1 of \citet{Camerer18} (every Stage-1 $p$-value below $0.05$ is in fact below $0.026$). The dependence of the decisions on the assumed $g$ argues for reporting such a sensitivity analysis whenever SPOT is used in practice.

\subsection{SPRINT blood pressure trial}
The Systolic Blood Pressure Intervention Trial \citep[SPRINT,][]{SPRINT15} randomised $9{,}361$ participants to intensive ($<120$ mmHg) versus standard ($<140$ mmHg) blood pressure targets.
We run SPOT with $K = 5$ on the five individual cardiovascular endpoints (excluding the primary composite, which overlaps with the individual components), with reported two-sided $p$-values: myocardial infarction ($p = 0.19$), non-MI acute coronary syndrome ($p = 0.99$), stroke ($p = 0.50$), heart failure ($p = 0.002$), and cardiovascular death ($p = 0.005$).
The reported $p$-values are two-sided, so we use the two-sided mixture-normal density with an assumed common standardised effect $\theta = 2.5$; the decisions are identical for $\theta = 2.0$ and $3.0$.
This application deliberately sits outside the product model in two respects: the five endpoints are measured on the same participants, so their $p$-values are dependent, and the endpoint-specific effects are heterogeneous (hazard ratios $0.83$, $1.00$, $0.89$, $0.62$, and $0.57$), whereas the model posits independent $p$-values with one common effect.
All standard methods reject heart failure and cardiovascular death ($p = 0.002$ and $0.005$); the policy returned by SPOT rejects exactly the same two endpoints, for every $\theta$ considered, because the null-like endpoints (acute coronary syndrome at $p = 0.99$, stroke at $p = 0.50$) drive the cumulative net benefit $\sum_{i=1}^{l} R_i$ negative beyond $l = 2$, so the borderline myocardial-infarction endpoint ($p = 0.19$) is not lifted.

\medskip\noindent\textit{Hierarchical variant.}
The hierarchical approach of \S\ref{sec:conclusion} of the main paper, suggested in the discussion of \citet{RHPA22}, accommodates such heterogeneity by design.
We partition the five endpoints into two groups, Group~A (heart failure, cardiovascular death; $K_A = 2$) and Group~B (myocardial infarction, acute coronary syndrome, stroke; $K_B = 3$), and run SPOT within each group at level $\alpha/2 = 0.025$.
For this illustration the grouping is data-informed; in practice the partition must be prespecified from prior clinical evidence, and even then the guarantee is model-based, since the endpoints are measured on the same participants.
SPOT returns both Group-A rejections and no Group-B rejections, matching the whole-family analysis while confining the exchangeability assumption to groups within which it is more plausible.

\subsection{Psychology replication studies}
The Reproducibility Project: Psychology \citep{OSC15} replicated $100$ studies from three psychology journals.
We select $K = 6$ social-psychology studies from the \emph{Journal of Personality and Social Psychology} with similar original effect sizes (approximately $r \in [0.21, 0.27]$): the replications of Albarrac\'{\i}n et al.\ ($p = 0.079$), Fischer et al.\ ($p = 0.141$), Shnabel and Nadler ($p = 0.234$), Centerbar et al.\ ($p = 0.322$), Correll ($p = 0.374$), and Cox et al.\ ($p = 0.469$), with replication sample sizes between $105$ and $200$; the $p$-values and sample sizes are from the project's public data file.
All six were originally significant ($p < 0.05$) but none replicated at $p < 0.05$.
We model the alternative with the two-sided mixture-normal density at $\theta = 2.0$, as in the other applications; the ${\rm Beta}(0.6, 1)$ model, $g(u) = 0.6\,u^{-0.4}$, gives the same no-rejection decision.
No standard method rejects any hypothesis, and SPOT likewise returns no rejections: the replication $p$-values are large and show no collective pattern of significance; even the smallest ($p = 0.079$) is far from the Bonferroni threshold $0.05/6 \approx 0.008$. There is no strong signal from which the likelihood-based rule can borrow strength, in contrast to the Camerer application.

\section{Simulation code}
\label{sec:supp_code}

All simulations are implemented in Python using NumPy and SciPy.
The package \texttt{general\_k} provides:
\begin{enumerate}
\item[(a)] \texttt{coefficients.py}: Elementary symmetric polynomial computation and FWER coefficient functions ($a_k$, $b_{\gamma,k}$, $R_i$), including vectorised batch versions.
\item[(b)] \texttt{algorithm.py}: Implementation of SPOT with bisection updates (Algorithm~\ref{alg:general_k} in the main paper): one Monte Carlo batch is drawn per configuration before the outer loop and reused across all sweeps (fixed empirical target functions, matching Theorem~\ref{res:C11}); the primary stopping rule is the empirical complementarity residual falling below $\{\alpha(1-\alpha)/N\}^{1/2}$, with iterate-change criteria as safeguards that can trigger an exit without meeting it, and per-sweep residuals are recorded as diagnostics; the upper bisection bracket is doubled until the crossing is bracketed, and a fast variant precomputes and caches the coefficient tensors.
\item[(c)] \texttt{target.py}: Monte Carlo FWER estimation and bisection target functions with common random numbers.
\item[(d)] \texttt{decision.py}: Optimal decision rule computing $l^* = \max\bigl(\argmax_{0 \le l \le K} \sum_{i=1}^l R_i\bigr)$, ties broken by the largest maximiser (the paper's convention).
\item[(e)] \texttt{alternatives.py}: Alternative models and working likelihood-ratio functions (truncated normal, mixture normal, Student-$t$, beta), with direct sampling and closed-form evaluation in log space where needed; the Student-$t$ log-likelihood ratio is capped at $45$ (equivalently, the likelihood ratio is capped at $\exp(45)$), as documented in Supplementary Section~\ref{sec:supp_models}.
\item[(f)] \texttt{baselines.py}: Standard FWER procedures (Bonferroni, Holm, Hochberg, Hommel).
\item[(g)] \texttt{simulations.py}: Simulation framework for power and FWER evaluation.
\end{enumerate}

Three scripts reproduce all numerical results:
\texttt{run\_all\_experiments.py} (the model $\times$ $K$ $\times$ signal grid and power-curve figures; \texttt{--jobs} optionally fans the independent settings out over worker processes),
\texttt{run\_scaling\_convergence.py} (the scaling and convergence tables and figures, run sequentially on a single core so the reported timings are meaningful, together with the $\alpha$- and $N_{\rm opt}$-sensitivity runs), and
\texttt{run\_applications.py} (the three applications, with the source data and their provenance documented in the script header).
Usage:
\begin{verbatim}
pip install -e '.[plots,dev]'
python run_all_experiments.py --k 3 4 5 6 --models trunc mixture t beta
python run_scaling_convergence.py
python run_applications.py
\end{verbatim}

The test suite (91 tests in \texttt{tests/}) validates the coefficient computation, the decision rule, the coordinate-descent and bisection routines, the Monte Carlo target functions, and the baseline procedures.

\medskip\noindent\textit{Computational environment.}
All computations were performed on a MacBook Pro with an Intel Core i7-1068NG7 processor (2.3\,GHz, 4 cores) and 32\,GB RAM, running macOS~15.
The software environment uses Python~3.12.7 with NumPy~1.26.4, SciPy~1.14.1, and Matplotlib~3.9.2.
Document typesetting was performed with pdfTeX (TeX~Live~2021).
All reported timings are from single-core runs.

\clearpage

\begin{table}[htbp]
\caption{Convergence of SPOT (Algorithm~\ref{alg:general_k}) under the truncated normal model at $\theta = -2.0$.
All runs used $N_{\rm opt} = 100{,}000$ Monte Carlo samples, upper bracket $U_{\max} = 50$, bisection tolerance $\delta = 10^{-4}$, and the residual stopping criterion with iterate-change tolerance $\varepsilon = 10^{-3}$; every scaling run stopped by the residual criterion.
Times are single-core: `Opt.' is the coordinate-descent (optimisation) time, the direct measure of algorithmic cost, and `Total' additionally includes the evaluation phase and the baseline procedures.}
\label{tab:supp_timing}
\begin{center}
\begin{tabular}{@{}rrrrr@{}}
\toprule
$K$ & Opt.\ (s) & Total (s) & Outer iterations & $\Pi_K$ \\
\midrule
3 & $3.1$ & $11$ & 4 & $0.638$ \\
4 & $5.0$ & $16$ & 4 & $0.633$ \\
5 & $7.4$ & $22$ & 4 & $0.630$ \\
6 & $9.9$ & $32$ & 4 & $0.623$ \\
8 & $23$ & $63$ & 5 & $0.617$ \\
10 & $30$ & $80$ & 5 & $0.614$ \\
12 & $48$ & $133$ & 5 & $0.610$ \\
\bottomrule
\end{tabular}
\end{center}
\end{table}

\begin{table}[htbp]
\caption{Rejection decisions for the $21$ Stage-1 replications of \citet{Camerer18} at $\alpha = 0.05$, sorted by $p$-value (their Supplementary Table~3; values reported as $<0.001$ set to $0.001$, marked $\dagger$).
R, rejection; $\cdot$, failure to reject; Holm, Hochberg, and Hommel coincide here.
SPOT is run with the two-sided mixture-normal alternative at $\theta = 2.0$.}
\label{tab:supp_camerer}
\begin{center}
\setlength{\tabcolsep}{5pt}\renewcommand{\arraystretch}{1.1}
\begin{tabular}{@{}lcccc@{}}
\toprule
Study & $p$-value & Bonf. & Holm--Hommel & SPOT \\
\midrule
Aviezer et al. (2012) & $0.001\,\dagger$ & R & R & R \\
Gneezy et al. (2014) & $0.001\,\dagger$ & R & R & R \\
Hauser et al. (2014) & $0.001\,\dagger$ & R & R & R \\
Kovacs et al. (2010) & $0.001\,\dagger$ & R & R & R \\
Morewedge et al. (2010) & $0.001\,\dagger$ & R & R & R \\
Wilson et al. (2014) & $0.001\,\dagger$ & R & R & R \\
Derex et al. (2013) & $0.003$ & $\cdot$ & R & R \\
Karpicke and Blunt (2011) & $0.006$ & $\cdot$ & $\cdot$ & R \\
Nishi et al. (2015) & $0.011$ & $\cdot$ & $\cdot$ & R \\
Balafoutas and Sutter (2012) & $0.022$ & $\cdot$ & $\cdot$ & $\cdot$ \\
Ackerman et al. (2010) & $0.024$ & $\cdot$ & $\cdot$ & $\cdot$ \\
Janssen et al. (2010) & $0.025$ & $\cdot$ & $\cdot$ & $\cdot$ \\
Pyc and Rawson (2010) & $0.089$ & $\cdot$ & $\cdot$ & $\cdot$ \\
Shah et al. (2012) & $0.150$ & $\cdot$ & $\cdot$ & $\cdot$ \\
Sparrow et al. (2011) & $0.265$ & $\cdot$ & $\cdot$ & $\cdot$ \\
Kidd and Castano (2013) & $0.273$ & $\cdot$ & $\cdot$ & $\cdot$ \\
Duncan et al. (2012) & $0.279$ & $\cdot$ & $\cdot$ & $\cdot$ \\
Rand et al. (2012) & $0.366$ & $\cdot$ & $\cdot$ & $\cdot$ \\
Gervais and Norenzayan (2012) & $0.416$ & $\cdot$ & $\cdot$ & $\cdot$ \\
Lee and Schwarz (2010) & $0.455$ & $\cdot$ & $\cdot$ & $\cdot$ \\
Ramirez and Beilock (2011) & $0.716$ & $\cdot$ & $\cdot$ & $\cdot$ \\
\bottomrule
\end{tabular}
\end{center}
\end{table}

\begin{table}[htbp]
\caption{Sensitivity of the Camerer et al.\ analysis to the assumed common effect $\theta$ of the two-sided mixture-normal alternative.
Standard methods reject $7$ ($6$ for Bonferroni) regardless of $\theta$.}
\label{tab:supp_camerer_sensitivity}
\begin{center}
\begin{tabular}{@{}cc@{}}
\toprule
Assumed effect $\theta$ & Rejections $l^*$ \\
\midrule
$1.5$ & 12 \\
$2.0$ & 9 \\
$2.5$ & 8 \\
$3.0$ & 8 \\
\bottomrule
\end{tabular}
\end{center}
\end{table}

\end{document}